\documentclass[10pt,english]{article}

\usepackage[T1]{fontenc}
\usepackage[latin9]{inputenc}
\usepackage[a4paper]{geometry}
\usepackage{color}
\usepackage{babel}
\usepackage{float}
\usepackage{amsmath}
\usepackage{amsthm}
\usepackage{amssymb}
\usepackage{graphicx}
\usepackage{setspace}
\usepackage[authoryear]{natbib}
\usepackage[unicode=true,pdfusetitle,
bookmarks=true,bookmarksnumbered=false,bookmarksopen=false,
breaklinks=false,pdfborder={0 0 1},backref=false,colorlinks=false]
{hyperref}

\makeatletter
\theoremstyle{plain}
\newtheorem{assumption}{\protect\assumptionname}
\theoremstyle{plain}
\newtheorem{lem}{\protect\lemmaname}
\theoremstyle{plain}
\newtheorem{thm}{\protect\theoremname}

\usepackage{lscape}
\usepackage{changepage}
\date{}

\makeatother

\providecommand{\assumptionname}{Assumption}
\providecommand{\lemmaname}{Lemma}
\providecommand{\theoremname}{Theorem}

\begin{document}
\title{Identification and Estimation of SVARMA models with Independent and
Non-Gaussian Inputs}
\author{Bernd Funovits}

\maketitle
\thispagestyle{empty}

\section*{Proposed Running Head}

Non-Gaussian SVARMA Identification

\section*{Affiliation}

\begin{singlespace}
\textbf{University of Helsinki}

Faculty of Social Sciences

Discipline of Economics

P. O. Box 17 (Arkadiankatu7)

FIN-00014 University of Helsinki
\end{singlespace}

and

\begin{singlespace}
\textbf{TU Wien}

Institute of Statistics and Mathematical Methods in Economics

Econometrics and System Theory

Wiedner Hauptstr. 8

A-1040 Vienna
\end{singlespace}

\section*{E-mail}

bernd.funovits@helsinki.fi

\pagebreak{}

\thispagestyle{empty}

\section*{Abstract}

This paper analyzes identifiability properties of structural vector
autoregressive moving average (SVARMA) models driven by independent
and non-Gaussian shocks. It is well known, that SVARMA models driven
by Gaussian errors are not identified without imposing further identifying
restrictions on the parameters. Even in reduced form and assuming
stability and invertibility, vector autoregressive moving average
models are in general not identified without requiring certain parameter
matrices to be non-singular. Independence and non-Gaussianity of the
shocks is used to show that they are identified up to permutations
and scalings. In this way, typically imposed identifying restrictions
are made testable. Furthermore, we introduce a maximum-likelihood
estimator of the non-Gaussian SVARMA model which is consistent and
asymptotically normally distributed.

Keywords: Structural vector autoregressive moving-average models,
non-Gaussianity, Identifiability

JEL classification: C32, C51, E52

\pagebreak{}

\section{Introduction}

\setcounter{page}{1}

Recently, \citet{LMS_svarIdent16} and \citet{GourierouxZakoianRenne17}
have shown that structural vector autoregressive (SVAR) models driven
by independent non-Gaussian components are identified up to scaling
and permutations which makes the typically imposed identifying restrictions
testable. If the error terms driving the economy are Gaussian or (cross-sectionally)
merely uncorrelated (as opposed to independent), one has to resort
to identifying restrictions in order to conclude on the fundamental
shocks driving the economy. From analysis in terms of second moments,
the true shocks can be identified only up to multiplication with orthogonal
matrices (which all lead to the same second moments of the  observed
process). Non-Gaussianity combined with cross-sectional independence,
however, allows to identify the shocks up to permutations and scalings.
In particular, infinitely many linear combinations of shocks generating
the same second moments are reduced to a finite set of linear combinations
generating the same distributional outcome. It is thus possible to
employ a data-driven approach instead of a story-telling approach.
Most importantly, the identifying (story-imposed) restrictions are
made testable when using the (data-driven) non-Gaussian SVARMA approach. 

Structural econometric analysis is usually conducted with SVAR models.
The situation for structural VARMA models driven by independent non-Gaussian
shocks is more complicated because one has to take additional identifiability
restrictions on the parameter space into account. In this paper, spectral
factorization techniques are employed to generalize the SVAR results
by \citet{LMS_svarIdent16} to the SVARMA case. While the literature
on SVAR models is abundant, see \citet{KilianLut17} and references
therein, the contributions regarding SVARMA models are easier to keep
track of, see, e.g., \citet{BoubacarFrancq11} and \citet{GourierouxMR_svarma19}.
In structural econometric analysis, the impulse response function
(IRF) and variance decompositions are the primary objects of interest
\citep{luet05,KilianLut17}. Especially in macroeconometrics, where
data is sometimes available only at quarterly instances, it is of
paramount importance to use a parsimoniously parameterized models
(like e.g. SVARMA models) for which the IRF and other can be obtained
straight-forwardly. It is widely known that SVARMA models are superior
to SVAR models in this respect, see, e.g., \citet{HannanDeistler12}.
Moreover, the articles \citet{Poskitt16}, \citet{PoskittYao17},
\citet{RaghavanAthSilvapulle16}, \citet{AthVahid08}, and \citet{AthVahid_JTSA_08}
provide ample evidence and make a strong point for using VARMA models
instead of VAR models for econometric analysis.

In a recent contribution, \citet{GourierouxMR_svarma19} consider
the dynamic identification problem in SVARMA models. While their focus
is a general treatment of whether it is possible to identify the root
location of determinantal roots of the associated MA polynomial matrix
in the structural VARMA case, we focus here on the precise derivation
of the properties of the maximum likelihood (ML) estimator of the
fundamental representation, including the first and second partial
derivatives with respect to all system and noise parameters.

One (perceived) disadvantage of VARMA models is increased complexity
of the estimation procedure compared to VAR models. Two rebuttals
are in order. First, there are many sophisticated (e.g. non-linear
threshold) VAR models whose estimation is arguably more involved than
the one of VARMA models. Second, there are many stable and openly
available software implementations which should put the complexities
of estimation of VAR and VARMA models on the same level. Examples
for implementations in the R software environment \citet{Rcore} are
\citet{Scherrer_rldm}, \citet{Tsay13,Tsay_R_MTS} and \citet{Gilbert_R_dse},
see also \citet{ScherrerDeistler2019_handbook} for a comparison and
further comments on these packages, and in MATLAB \citet{Gomez_matlab_15,Gomez16}.
The estimation procedure described in this article is implemented
in R and can be installed with the command \texttt{devtools::install\_github(``bfunovits/svarma\_id'')}
in the R console.

The rest of the paper is structured as follows. In section 2, the
SVARMA model is introduced. In section 3, the identification result
is stated and proved. In section 4, the maximum likelihood (ML) estimator
is derived and shown to be consistent and asymptotically normal. In
section 5, we illustrate the method. Proofs and technical details
are available in the Online Appendix.

We use $z$ as a complex variable as well as the backward shift operator
on a stochastic process, i.e. $z\left(y_{t}\right)_{t\in\mathbb{Z}}=\left(y_{t-1}\right)_{t\in\mathbb{Z}}$
and define $i=\sqrt{-1}$. The transpose of an $\left(m\times n\right)$-dimensional
matrix $A$ is denoted by $A'$. The column-wise vectorization of
$A\in\mathbb{R}^{m\times n}$ is denoted by $vec\left(A\right)\in\mathbb{R}^{mn\times1}$
and for a square matrix $B\in\mathbb{R}^{n\times n}$ we denote with
$vecd{^\circ}\left(B\right)\in\mathbb{R}^{n(n-1)}$ the vectorization
where the diagonal elements of $B$ are left out. The $n$-dimensional
identity matrix is denoted by $I_{n}$, an $n$-dimensional diagonal
matrix with diagonal elements $\left(a_{1},\ldots,a_{n}\right)$ is
denoted by $\text{diag}\left(a_{1},\ldots,a_{n}\right)$, and the
inequality $">0"$ means positive definiteness in the context of matrices.
The column vector $\iota_{i}$ has a one at positions $i$ and zeros
everywhere else. The expectation of a random variable with respect
to a given probability space is denoted by $\mathbb{E}\left(\cdot\right)$.
Convergence in probability and in distribution are denoted by $\xrightarrow{p}$
and $\xrightarrow{d}$, respectively. Partial derivatives $\left.\frac{\partial f(x)}{\partial x}\right|_{x=x_{0}}$
of a real-valued function $f(x)$ evaluated at a point $x_{0}\in\mathbb{R}^{k}$
are denoted by $f_{x}\left(x_{0}\right)$ and considered columns.

\section{\label{sec:Model}Model}

We start from an $n$-dimensional VARMA system 
\begin{equation}
\underbrace{\left(I_{n}-a_{1}z-\cdots a_{p}z^{p}\right)}_{=a(z)}y_{t}=\underbrace{\left(I_{n}+b_{1}z+\cdots+b_{q}z^{q}\right)}_{=b(z)}B\varepsilon_{t},\quad a_{i},b_{i}\in\mathbb{R}^{n\times n}.\label{eq:system}
\end{equation}
The shocks $\left(\varepsilon_{t}\right)_{t\in\mathbb{Z}}$ driving
the system are identically and independently distributed (i.i.d.)
in cross-section and time, have zero mean, and diagonal covariance
matrix $\Sigma^{2}$ with positive diagonal elements $\sigma_{i}^{2}$,
whose positive square root is in turn denoted by $\sigma_{i}$ . To
simplify presentation, we also introduce the column vector $\sigma=\left(\sigma_{1},\ldots,\sigma_{n}\right)'$
and $\Sigma=\text{diag}\left(\sigma_{1},\ldots,\sigma_{n}\right)$,
as well as $x_{t-1}'=\left(y_{t-1}',\ldots,y_{t-p}'\right)$ and $s_{t-1}'=\left(\varepsilon_{t-1}'B',\ldots,\varepsilon_{t-q}'B'\right)$
such that equation \eqref{sec:Model} can be written as 
\[
y_{t}=\left(a_{1},\ldots,a_{p}\right)x_{t-1}+\left(b_{1},\ldots,b_{q}\right)s_{t-1}+B\varepsilon_{t}.
\]
We assume that the stability condition 
\begin{equation}
\det\left(a(z)\right)\neq0,\ \left|z\right|\leq1,\label{eq:stability}
\end{equation}
and the strict invertibility condition 
\begin{equation}
\det\left(b(z)\right)\neq0,\ \left|z\right|\leq1\label{eq:invertibility}
\end{equation}
hold, and that $B$ is invertible and has ones on its diagonal. Furthermore,
we assume that the polynomial matrices $a(z)$ and $b(z)$ are left-coprime\footnote{Two matrix polynomials are called left-coprime if $\left(a(z),b(z)\right)$
is of full row rank for all $z\in\mathbb{C}$. For equivalent definitions
see \citet{HannanDeistler12} Lemma 2.2.1 on page 40.} and that $\left(a_{p},b_{q}\right)$ is of full rank\footnote{The stability, invertibility, coprimeness, and full-rank assumptions
on the parameters in $a(z)$ and $b(z)$ can be relaxed. Imposing
them, allows us to focus on the essential part of this contribution:
To reduce the class of observational equivalence in terms of second
moments from the orthogonal matrices to permutation matrices in the
context of SVARMA models.}. Note that this full rank assumption is over-identifying in the sense
that some rational transfer function cannot be parameterized by any
VARMA(p,q) system which satisfies this assumption, see \citet{Hannan71}
or \citet{HannanDeistler12}, Chapter 2.7 on page 77.

The stationary solution $\left(y_{t}\right)_{t\in\mathbb{Z}}$ of
the system \eqref{eq:system} is called an ARMA process.

We follow \citet{Rothenberg71} to define identifiability of parametric
models. The external characteristic of the stationary solution $\left(y_{t}\right)_{t\in\mathbb{Z}}$
of \eqref{eq:system} is the probability distribution function (or
a subset of corresponding moments). A particular system \eqref{eq:system}
is described by the parameters of \eqref{eq:system} which satisfy
assumptions \eqref{eq:stability} and \eqref{eq:invertibility} as
well as the coprimeness assumption, the full rank assumption and the
assumptions on $B$ and $D^{2}$. The model is then characterized
by the set of all a priori possible systems which we will call internal
characteristics. Two systems of the form \eqref{eq:system} are called
observationally equivalent if they imply the same external characteristics
of $\left(y_{t}\right)_{t\in\mathbb{Z}}$. A system is identifiable
if there is no other observationally equivalent system. The identifiability
problem is concerned with the existence of an injective function from
the internal characteristics to the external characteristics\footnote{The inverse of this function, i.e. from the external to the internal
characteristics, is called the identifying function.}, see \citet{DeistlerSeifert78} for a more detailed discussion.

The classical (non-)identifiability issues where the external characteristics
are described by the second moments of $\left(y_{t}\right)_{t\in\mathbb{Z}}$
are best understood in terms of the spectral density of the stationary
solution of \eqref{eq:system}. The spectral density, i.e. the Fourier
transform of the autocovariance function $\gamma(s)=\mathbb{E}\left(y_{t}y_{t-s}'\right),\ s\in\mathbb{Z},$
of $\left(y_{t}\right)_{t\in\mathbb{Z}}$ , is 
\[
f(z)=a(z)^{-1}b(z)B\Sigma^{2}B'b'\left(\frac{1}{z}\right)a'\left(\frac{1}{z}\right)^{-1}=k(z)\left(B\Sigma^{2}B'\right)k'\left(\frac{1}{z}\right),
\]
evaluated at $z=e^{-i\lambda}$, $\lambda\in\left[-\pi,\pi\right]$,
where $k(z)=a(z)^{-1}b(z)=\sum_{j=0}^{\infty}k_{j}z^{j},\ k(0)=I_{n}$,
and $k(z)B$ corresponds to the transfer function relating the output
$y_{t}$ to the $\left(\varepsilon_{t}\right)_{t\in\mathbb{Z}}$ .

On the one hand, transforming the pair $\left(B,\Sigma\right)$ with
an orthogonal matrix\footnote{A square matrix is orthogonal if $QQ'=Q'Q=I_{n}$.}
$Q$ to $\left(B\Sigma Q\Sigma_{1}^{-1},\Sigma_{1}\right)$, where
$\Sigma_{1}$ is a diagonal matrix such that the diagonal elements
of $B$ are equal to one, generates the same spectral density because
$B_{1}\Sigma_{1}^{2}B_{1}'=B\Sigma^{2}B'$ where $B_{1}=B\Sigma Q\Sigma_{1}^{-1}$.
Hence, the class of observational equivalence is at least $\frac{n(n-1)}{2}$-dimensional.
On the other hand, it is easy to see \citep[page 66]{Hannan70} that
two spectral factors\footnote{A spectral factor $l(z)$ is a rational matrix function for which
$l(z)l'\left(\frac{1}{z}\right)$, evaluated at the unit circle, is
equal to the spectral density.} of the form $k(z)B\Sigma=a(z)^{-1}b(z)B\Sigma$, where $a(z)$ and
$b(z)$ satisfy \eqref{eq:stability} and \eqref{eq:invertibility}
as well as the coprimeness assumption, the full rank assumption and
where $B$ and $\Sigma$ satisfy the assumptions outlined above, obtained
from the spectral density corresponding to the stationary solution
of \eqref{eq:system} are related through orthogonal matrices. This
means that any other spectral factor is of the form $a(z)^{-1}b(z)B\Sigma Q$
where $Q$ is an orthogonal matrix. By normalizing the diagonal elements
of $B\Sigma Q$, we obtain a new pair $\left(B_{1},\Sigma_{1}\right)$
of the required form. Hence, the class of observational equivalence
is $\frac{n(n-1)}{2}$-dimensional. This result, however, only uses
second moment information and not the full distribution of the stochastic
process $\left(\varepsilon_{t}\right)_{t\in\mathbb{Z}}$.

We will show in the next section that if the inputs $\left(\varepsilon_{t}\right)$
to \eqref{eq:system} are non-Gaussian and independent, the spectral
factors are related by permutation matrices (modulo sign). Thus, we
reduce the class of observational equivalence from the group of orthogonal
matrices to the group of (signed) permutations. 

\section{\label{sec:identification_scheme}Identification of the Instantaneous
Shock Transmission}

In this section, we first use the cross-sectional independence and
non-Gaussianity of the components of the shocks $\varepsilon_{t}$
for identifying the matrix $B$ up to permutation and scaling of its
columns.Finally, we discuss advantages and disadvantages of various
rules for choosing a particular permutation and scaling.

The assumptions on the error term $\varepsilon_{t}=\left(\varepsilon_{1,t},\ldots,\varepsilon_{n,t}\right)$
are the same as in \citet{LMS_svarIdent16}, the essential one being
that the components (at one point in time) are mutually independent
and that at most one of them has a Gaussian marginal distribution.

\begin{assumption}
\label{assu:non_gaussianIID}We assume the following.
\begin{enumerate}
\item The error process $\varepsilon_{t}=\left(\varepsilon_{1,t},\ldots,\varepsilon_{n,t}\right)$
is a sequence of i.i.d. random vectors. Each component $\varepsilon_{i,t},\ i\in\left\{ 1,\ldots,n\right\} $
has zero mean and positive variance. 
\item For any (fixed) point in time, the components of $\varepsilon_{t}$
are mutually independent and at most one of the components has a Gaussian
marginal distribution. 
\end{enumerate}
\end{assumption}
In order to strengthen intuition as to how non-Gaussianity and independence
help reducing the size of the class of observational equivalence,
consider the following example featuring two identically and independently
uniformly distributed random variables. Rotating these two variables
45 degrees (with rotation matrix $\frac{1}{\sqrt{2}}\left(\begin{smallmatrix}1 & 1\\
1 & -1
\end{smallmatrix}\right)$) leads to marginal distributions which are ``more Gaussian'' (e.g.
measured by the absolute value of the excess kurtosis) than the original
variables. This suggests that searching for linear combinations that
lead to ``maximally non-Gaussian'' variables might pin down a rotation.
In the following, we present a formal approach.

\subsection{Fixing a Rotation}

The theoretical background for reducing the class of observational
equivalence from orthogonal matrices to (signed) permutations is provided
by the following lemma. It allows to conclude from the independence
of the sums of independent variables on the distribution of the underlying
summands. In particular, it is useful to conclude on the coefficients
pertaining to the summands if one makes additional assumptions on
the distribution of the summands.

We use
\begin{lem}[\citet{Kagan73}, Theorem 3.1.1]
\label{lem:Kagan} Let $X_{1},\ldots X_{n}$ be independent (not
necessarily identically distributed) random variables, and define
$Y_{1}=\sum_{i=1}^{n}a_{i}X_{i}$ and $Y_{2}=\sum_{i=1}^{n}b_{i}X_{i}$
where $a_{i}$ and $b_{i}$ are constants. If $Y_{1}$ and $Y_{2}$
are independent, then the random variables $X_{j}$ for which $a_{j}b_{j}\neq0$
are all normally distributed. 
\end{lem}
In the following, Lemma \ref{lem:Kagan} is used to conclude on the
columns of $M$ in $\varepsilon_{t}=M\varepsilon_{t}^{*}$, where
$M=B^{-1}B^{*}$, where both $\varepsilon_{t}$ and $\varepsilon_{t}^{*}$
are assumed to be (cross-sectionally) independent and non-Gaussian.
The components of $\varepsilon_{t}$ correspond to $Y_{1},\ Y_{2}$,
the components of $\varepsilon_{t}^{*}$ correspond to $X_{1},\ldots,X_{n}$.
E.g., for component 1 and 2 of $\varepsilon_{t}$ we have $\varepsilon_{1,t}=\left(m_{11},\ldots,m_{1n}\right)\varepsilon_{t}^{*}$
and $\varepsilon_{2,t}=\left(m_{21},\ldots,m_{2n}\right)\varepsilon_{t}^{*}$.
If any pair of coefficients $\left(m_{1k},m_{2k}\right)$ satisfies
$m_{1k}m_{2k}\neq0$, then the corresponding component $\varepsilon_{k,t}^{*}$
is Gaussian according to the Lemma. By Assumption 1, at most one component
of $\varepsilon_{t}^{*}$ is allowed to have a Gaussian marginal distribution.
It follows that there cannot be another pair $\left(m_{1l},m_{2l}\right),\ l\neq k,$
that satisfies $m_{1l}m_{2l}\neq0$. In particular, there is (at most)
one non-zero coefficient in the scalar product $\left\langle m_{1,\bullet},m_{2,\bullet}\right\rangle =m_{1k}m_{2k}\neq0$,
where $m_{i,\bullet}$ denotes the $i$-th row of $M$. If $\left\langle m_{1,\bullet},m_{2,\bullet}\right\rangle =m_{1k}m_{2k}\neq0$,
we obtain a contradiction to the assumption that $\mathbb{E}\left(\varepsilon_{1,t}\varepsilon_{2,t}\right)=0$
because from the fact that one (exactly one) component $\varepsilon_{k,t}^{*}$
is Gaussian and $\varepsilon_{i,t}=m_{i,\bullet}\begin{pmatrix}\varepsilon_{1,t}^{*} & \cdots & \varepsilon_{n,t}^{*}\end{pmatrix}^{'}$
we obtain that $\mathbb{E}\left(\varepsilon_{1,t}\varepsilon_{2,t}\right)=m_{1,\bullet}D^{*}m_{2,\bullet}'=d_{k}^{*}m_{1k}m_{2k}\neq0$.
It thus follows that all pairs $\left(m_{1k},m_{2k}\right)$ satisfy
$m_{1k}m_{2k}=0$. Since this argument holds for all pairs in $\varepsilon_{1,t},\ldots,\varepsilon_{n,t}$,
it follows that every column contains at most one non-zero element.
Finally, non-singularity implies that every column contains exactly
one non-zero element.

Now we are ready to prove
\begin{thm}
The set of observationally equivalent ARMA systems of the form in
section \ref{sec:Model} is described by the set of matrices $PD$
where $P$ is a permutation matrix and $D$ a diagonal matrix with
non-zero diagonal entries. 
\end{thm}
\begin{proof}
Consider two systems \eqref{eq:system}, say $\left(a(z),b(z);B,\Sigma\right)$
and $\left(a^{*}(z),b^{*}(z);B^{*},\Sigma^{*}\right)$ whose stationary
solutions have the same spectral density (or equivalently the same
second moments), in particular $B\Sigma^{2}B'=B^{*}\Sigma^{*2}B^{*'}$.
Written differently, we consider 
\[
y_{t}=\left(a_{1},\ldots,a_{p}\right)x_{t-1}+\left(b_{1},\ldots,b_{q}\right)s_{t-1}+B\varepsilon_{t}
\]
and 
\[
y_{t}=\left(a_{1}^{*},\ldots,a_{p}^{*}\right)x_{t-1}+\left(b_{1}^{*},\ldots,b_{q}^{*}\right)s_{t-1}^{*}+B^{*}\varepsilon_{t}^{*},
\]
post-multiply $\left(x_{t-1}',s_{t-1}'\right)$ and $\left(x_{t-1}',s_{t-1}^{*'}\right)$
respectively, where $s_{t-1}^{*'}=\left(\varepsilon_{t-1}^{*'}B^{*'},\ldots,\varepsilon_{t-q}^{*'}B^{*'}\right)$,
and take expectations such that 
\begin{gather}
\left(\gamma_{1},\ldots,\gamma_{p},k_{1}BD^{2}B',\ldots,k_{q}BD^{2}B'\right)=\nonumber \\
\begin{pmatrix}a_{1} & \cdots & a_{p} & b_{1} & \cdots & b_{q}\end{pmatrix}\mathbb{E}\left(\begin{pmatrix}x_{t-1}\\
s_{t-1}
\end{pmatrix}\begin{pmatrix}x_{t-1}' & s_{t-1}'\end{pmatrix}\right)\label{eq:YW_sys1}
\end{gather}
and 
\begin{gather}
\left(\gamma_{1},\ldots,\gamma_{p},k_{1}B^{*}D^{*2}B^{*'},\ldots,k_{q}B^{*}D^{*2}B^{*'}\right)-\cdots\nonumber \\
\cdots-\begin{pmatrix}a_{1}^{*} & \cdots & a_{p}^{*} & b_{1}^{*} & \cdots & b_{q}^{*}\end{pmatrix}\mathbb{E}\left(\begin{pmatrix}x_{t-1}\\
s_{t-1}^{*}
\end{pmatrix}\begin{pmatrix}x_{t-1}' & s_{t-1}^{*'}\end{pmatrix}\right)=0.\label{eq:YW_sys2}
\end{gather}
Since the stationary solution $\left(y_{t}\right)_{t\in\mathbb{Z}}$
of \eqref{eq:system} depends only on past inputs, we obtain that
the right-hand-side of the equation is zero. The square matrices in
\eqref{eq:YW_sys1} and \eqref{eq:YW_sys2} are non-singular due to
the coprimeness assumption on $\left(a(z),b(z)\right)$ and the full-rank
assumption on $\left(a_{p},b_{q}\right)$, compare \citet{deistler83}.
The elements in this matrix correspond either to autocovariances or
can be obtained as, e.g., $\mathbb{E}\left(y_{t-1}\varepsilon_{t-1}^{'}B'\right)=\mathbb{E}\left[\left(\sum_{j=0}^{\infty}k_{j}B\varepsilon_{t-1-j}\right)\varepsilon_{t-1}^{'}B'\right]=k_{0}B\Sigma^{2}B'$.

Now, it follows that $\begin{pmatrix}a_{1} & \cdots & a_{p} & b_{1} & \cdots & b_{q}\end{pmatrix}=\begin{pmatrix}a_{1}^{*} & \cdots & a_{p}^{*} & b_{1}^{*} & \cdots & b_{q}^{*}\end{pmatrix}$
and $B\varepsilon_{t}=B^{*}\varepsilon_{t}^{*}$ because both equation
system involve the same second moments (in particular $B\Sigma^{2}B'=B^{*}\Sigma^{*2}B^{*'}$).
The remainder of the proof follows from what was discussed below
Lemma \ref{lem:Kagan}.
\end{proof}
While the proof above is easily understandable for readers who know
the paper \citet{LMS_svarIdent16}, the following proof uses less
matrix algebra.
\begin{proof}
A different way to prove this theorem uses spectral factorization
arguments (in the guise of linear projections and the Wold representation
theorem). Starting from the stationary solution $\left(y_{t}\right)_{t\in\mathbb{Z}}$
of \eqref{eq:system}, we project $y_{t}$ on its infinite past in
order to obtain the linear innovation $v_{t}$, i.e. $y_{t}-Proj\left(y_{t}|y_{t-1},y_{t-2},\ldots\right)=v_{t}$.
Note that $Proj\left(y_{t}|y_{t-1},y_{t-2},\ldots\right)=Proj\left(y_{t}|B\varepsilon_{t-1},B\varepsilon_{t-2},\ldots\right)=Proj\left(y_{t}|v_{t-1},v_{t-2},\ldots\right)$
because the linear space spanned by the components of $\left\{ y_{t-1},y_{t-2},\ldots\right\} $
coincides with the linear space spanned by the components of $\left\{ v_{t-1},v_{t-2},\ldots\right\} $
and the linear space spanned by the components of $\left\{ B\varepsilon_{t-1},B\varepsilon_{t-2},\ldots\right\} $\footnote{Note that not only the projection is unique (as follows from the projection
theorem) but also the representation in the given basis $\left(y_{t-1},\ldots,y_{t-p},B\varepsilon_{t-1},\ldots,B\varepsilon_{t-q}\right)$
because of the assumptions that $\left(a(z),b(z)\right)$ be left-coprime
and that $\left(a_{p},b_{q}\right)$ be of full rank.}. Now, knowing that the inputs $\varepsilon_{t}$ to \eqref{eq:system}
are not only uncorrelated but also independent and non-Gaussian, we
factorize the covariance matrix of $v_{t}$ as $\mathbb{E}\left(v_{t}v_{t}'\right)=B\mathbb{E}\left(\varepsilon_{t}\varepsilon_{t}'\right)B'$
where $\mathbb{E}\left(\varepsilon_{t}\varepsilon_{t}'\right)=\Sigma^{2}$
is diagonal. It is obvious that it is impossible to distinguish between
$B^{*}\varepsilon_{t}^{*}=B\Sigma^{*}Q\Sigma^{-1}\varepsilon_{t}$
and $B\varepsilon_{t}$ for any orthogonal matrix $Q$ by second moments
only. The rest of the proof is the same as above. 
\end{proof}
The difference in these two proofs is as follows. In the first proof,
we use model \eqref{eq:system} together with its assumptions earlier,
i.e. we write down the ARMA equation, take expectations, and obtain
that any two independent error terms satisfy $\varepsilon_{t}^{*}=\left(B^{*}\right)^{-1}B\varepsilon_{t}$.
In the second proof, we focus firstly on the linear innovations $v_{t}$
and only use the fact that the error terms are independent when it
comes to parameterizing the covariance matrix of the innovations.
Note that the second proof suggests that as soon as one can identify
the true inputs (irrespective of the model), one may use cross-sectional
independence and non-Gaussianity to reduce the equivalence class of
orthogonal matrices to the one of (signed) permutation matrices.

\subsection{Identification Scheme: Choosing a Unique Permutation and Scaling}

In this section, we describe how to pick one particular permutation
and scaling from the class of observational equivalence described
in the previous section. In order to do this, we describe different
identification schemes, i.e. rules for choosing a particular permutation
and scaling of the matrix $B$. 

We start by repeating two identification schemes presented in \citet{LMS_svarIdent16}
(which are in turn based on \citet{IlmonenPaindaveine11} and \citet{HallinMehta15}).
The \textit{first identification} scheme, which is convenient for
deriving asymptotic properties and which we refer to as \textbf{identification
scheme A}, consists in firstly scaling all columns of $B$ such that
their norm is equal to one, secondly, permutating the columns such
that the absolute value of each diagonal element is larger than the
absolute value of all elements in the same row with a higher column
index, and finally scaling all columns of $B$ such that the diagonal
elements are equal to one\footnote{Note that in the derivation of the ML estimator, we impose only that
the diagonal elements of $B$ be equal to one. Thus, the restrictions,
in general, do not suffice to pin down the particular permutation
and scaling for $B$. However, the fact that the observationally equivalent
points in the parameter space are discrete ensures the existence of
a consistent root, i.e. the solution of the first order conditions
obtained from taking derivatives of the standardized log-likelihood
function. Should the gradient descent algorithm return a $B$ matrix
which does not satisfy the identification scheme, it can be easily
transformed such that the identification scheme is satisfied. The
companion R-package to this article transforms the $B$ matrix such
that all restrictions described here are satisfied.}. The \textit{second identification scheme} consists of the same first
two steps but instead of scaling the columns in the last step such
that their diagonal elements are equal to one, it is required that
the diagonal elements are positive. Sometimes, the second identification
scheme turns out to be more flexible, for example when testing hypotheses
involving diagonal elements. Regarding the derivation of asymptotic
properties, however, one would need to maximize the constrained (log-)
likelihood function where the restrictions that the columns of $B$
have length one are taken into account.

It is important to realize that the transformations used in the identification
schemes described above, exist not on the whole parameter space but
only on a topologically large set in the parameter set. For details,
see Proposition 2 in \citet{LMS_svarIdent16} including an example
of a matrix or which the above identification schemes are not defined.
The \textit{third identification} \textit{scheme}, similar to the
one in \citet{ChenBickel05} on page 3626, does not exclude any non-singular
matrix $B$ and is defined by the following transformations. Firstly,
the columns of $B$ are scaled to have norm equal to one. Secondly,
in each column, the element with largest absolute value is made positive.
Finally, the columns are ordered according to $\prec$ such that $c\prec d$
for two columns $c,d$ of $B$ if and only if there exists a $k\in\left\{ 1,\ldots,n\right\} $
such that $c_{k}<d_{k}$ and $c_{j}=d_{j}$ for all $j\in\left\{ 1,\ldots,k-1\right\} $.

Now that we have firstly obtained a discrete set of observationally
equivalent SVARMA systems and secondly provided different rules to
select a unique representative, we may proceed to local ML estimation
of the true underlying parameter.

\section{Parameter Estimation}

In this section, we treat local ML estimation of \eqref{eq:system}.
In particular, we prove local consistency and asymptotic normality
of the ML estimator (MLE).

In order to separate the essential ideas from technicalities, we start
by stating a theorem for local asymptotic normality of the MLE in
terms of (easily understandable and intuitive) high-level assumptions
on the densities of i.i.d. shocks. Next, we discuss (component-wise)
the densities of the i.i.d. shocks $\left(\varepsilon_{t}\right)$,
the admissible parameter space, and the (standardized) log-likelihood
function of our problem at hand. Last, we state a theorem for local
asymptotic normality of the MLE in terms of low-level integrability
and differentiability assumptions on the densities and verify the
high-level assumptions. The proofs and many technicalities (e.g. partial
derivatives of the likelihood function) which are similar to the
ones in \citet{LMS_svarIdent16} are deferred to the Online Appendix.

\subsection{\label{subsec:asy_lowlevel}Local Asymptotic Normality in terms of
High-Level Assumptions}

For the sake of clarity, and in order to understand where the low-level
assumptions on the densities that we will introduce in Assumption
\ref{assu:densities}below come into play, we state a theorem proving
local asymptotic normality in terms of high-level assumptions. In
the Online Appendix, we show how the low-level assumptions imply the
high-level assumptions. The (standardized) log-likelihood function
to be maximized is 
\begin{align}
L_{T}\left(\theta\right) & =\frac{1}{T}\sum_{t=1}^{T}l_{t}\left(\varepsilon_{t}(\theta).\theta\right)\label{eq:likelihood}
\end{align}
where $l_{t}\left(\varepsilon_{t}(\theta).\theta\right)=\log\left(f\left(\varepsilon_{t}(\theta).\theta\right)\right)$
are the individual contributions to the log-likelihood function and
$f(\cdot)$ is the (joint) density of a residuals $\varepsilon_{t}\left(\theta\right)$
which are obtained from a parametric model with parameter $\theta\in\Theta\subseteq\mathbb{R}^{k}$.

The following discussion builds on \citet{poetpruch97}. Firstly,
the existence of a sequence of solutions $\left(\hat{\theta}_{T}\right)$
of the first order condition $L_{\theta,T}\left(\theta\right)=0$
of the standardized log-likelihood function which converges almost
surely towards $\theta_{0}$ is required. This is essentially guaranteed
by the identification result in the previous section (and some technical
conditions), showing that the observationally equivalent points in
the parameter space are discrete (in the sense that there exist disjoint
open sets around each point of this kind). Furthermore, the score
of the individual contributions $l_{t}\left(\theta\right)$ to $L_{T}\left(\theta\right)$
has to satisfy a Central Limit Theorem (CLT) for martingale difference
sequences (MDS) and the Hessian of the individual contributions has
to satisfy a Uniform Law of Large Numbers (ULLN). If these conditions
are satisfied, the sequence $\sqrt{T}\left(\hat{\theta}_{T}-\theta_{0}\right)$
is asymptotically normal. To make this discussion more precise, we
state
\begin{thm}
\label{thm:generic_mle}For \eqref{eq:likelihood}, the following
conditions are assumed to be true:
\begin{enumerate}
\item \label{enu:generic_mle_solution}There exists a sequence of estimators
$\left(\hat{\theta}_{T}\right)$ converging almost surely to an interior
point $\theta_{0}\in\Theta$ for which $L_{\theta,T}\left(\hat{\theta}_{n}\right)=o_{P}\left(\frac{1}{\sqrt{T}}\right)$. 
\item \label{enu:generic_mle_stat_ergod}$\left(\varepsilon_{t}\right)$
is stationary and ergodic with density $f\left(x,\theta_{0}\right)$
\item \label{enu:generic_mle_mds}$l_{\theta,t}\left(\varepsilon_{t}(\theta).\theta\right)$
is an MDS.
\item \label{enu:generic_mle_density}For the parametric family $\left\{ f\left(x.\theta\right)\ |\ \theta\in\Theta\subseteq\mathbb{R}^{k},\ x\in\mathbb{R}^{n}\right\} $
of densities it holds that $f\left(x,\theta\right)>0$ for all $\left(x,\theta\right)$
and that $f\left(x,\theta\right)$ is twice continuously differentiable
with respect to $\theta$ in an open neighborhood centered at $\theta_{0}$
for all $x$.
\item \label{enu:generic_mle_clt_opg}The individual contributions to the
standardized log-likelihood function satisfy $\mathbb{E}\left(\left\Vert l_{\theta,t}\left(\varepsilon_{t}(\theta_{0}),\theta_{0}\right)\right\Vert ^{2}\right)<\infty$.
\item \label{enu:generic_mle_ulln}There exists a (non-singleton) compact
set $\Theta_{0}$ such that $\mathbb{E}\left(\sup_{\theta\in\Theta_{0}}\left\Vert l_{\theta\theta,t}\left(\varepsilon_{t}(\theta),\theta\right)\right\Vert \right)<\infty$.
\item \label{enu:generic_mle_non_singular}The Hessian matrix $\mathbb{E}\left(l_{\theta\theta,t}\left(\varepsilon_{t}(\theta_{0}),\theta_{0}\right)\right)$
is non-singular.
\item \label{enu:generic_mle_hess_equal2_opg}At the true parameter value,
the expectation of the outer product of the score is equal to the
negative expectation of the Hessian of the individual contribution
to the likelihood, i.e. $\mathbb{E}\left(l_{\theta\theta,t}\left(\varepsilon_{t}(\theta_{0}),\theta_{0}\right)\right)=-\mathbb{E}\left(l_{\theta,t}\left(\varepsilon_{t}(\theta_{0}),\theta_{0}\right)l_{\theta,t}\left(\varepsilon_{t}(\theta_{0}),\theta_{0}\right)'\right)$
holds.
\end{enumerate}
Under 1) to 8), we obtain that 
\[
\sqrt{T}\left(\hat{\theta}_{T}-\theta_{0}\right)\xrightarrow{d}\mathcal{N}\left(0,\left[\mathbb{E}\left(l_{\theta,t}\left(\varepsilon_{t}(\theta_{0}),\theta_{0}\right)l_{\theta,t}\left(\varepsilon_{t}(\theta_{0}),\theta_{0}\right)'\right)\right]^{-1}\right).
\]
\end{thm}
The basic idea consists in applying (component-wise) the mean value
theorem to $\left(L_{\theta,T}\left(\hat{\theta}_{T}\right),L_{\theta,T}\left(\theta_{0}\right)\right)$
such that one obtains asymptotically $\sqrt{T}\left(L_{\theta,T}\left(\hat{\theta}_{T}\right)-L_{\theta,T}\left(\theta_{0}\right)\right)=\bar{A}_{T}\sqrt{T}\left(\hat{\theta}_{T}-\theta_{0}\right)$
where the matrix $\bar{A}_{T}$ corresponds to the Hessian whose rows
are evaluated at the respective mean values. Point \ref{enu:generic_mle_solution})
is necessary for the existence of a consistent sequence $\left(\hat{\theta}_{T}\right)$
and is, together with point \ref{enu:generic_mle_stat_ergod}), \ref{enu:generic_mle_mds}),
\ref{enu:generic_mle_density}), and \ref{enu:generic_mle_clt_opg}),
required for the CLT for MDS. It follows that $\sqrt{T}L_{\theta,T}\left(\theta_{0}\right)$
and $-\bar{A}_{T}\sqrt{T}\left(\hat{\theta}_{T}-\theta_{0}\right)$
are asymptotically normal with the same asymptotic distribution, i.e.
$\mathcal{N}\left(0,\mathbb{E}\left(l_{\theta,t}\left(\varepsilon_{t}(\theta_{0}),\theta_{0}\right)l_{\theta,t}\left(\varepsilon_{t}(\theta_{0}),\theta_{0}\right)'\right)\right)$.
Moreover, one needs to ensure that the Hessian satisfies a ULLN\footnote{This means that $\sup_{\theta\in\Theta_{0}}\left\Vert \frac{1}{T}\sum_{t=1}^{T}l_{\theta\theta,t}(\theta)-\mathbb{E}\left(l_{\theta\theta,t}(\theta)\right)\right\Vert \rightarrow0,\ a.s.,$
and as a byproduct $\mathbb{E}\left(l_{\theta\theta,t}\left(\theta\right)\right)$
is continuous at $\theta_{0}$.} such that $\bar{A}_{T}$ converges towards the non-singular expectation
of the Hessian evaluated at the true parameter value which is moreover
equal to the negative of the expectation of the outer product of the
score. This is ensured by points \ref{enu:generic_mle_ulln}), \ref{enu:generic_mle_non_singular})
and \ref{enu:generic_mle_hess_equal2_opg})\footnote{Point \ref{enu:generic_mle_hess_equal2_opg}) is, e.g., implied by
requiring that $\int\sup_{\theta\in\Theta_{0}}\left\Vert l_{\theta,t}\left(\varepsilon_{t}(\theta),\theta\right)\right\Vert dx<\infty$
and $\int\sup_{\theta\in\Theta_{0}}\left\Vert l_{\theta\theta,t}\left(\varepsilon_{t}(\theta),\theta\right)\right\Vert dx<\infty$
but can also be obtained by less stringent assumptions.}, respectively.

Note that the covariance matrix can be consistently estimated by $-A_{T}^{-1}$,
where $A_{T}=\frac{1}{T}\sum_{t=1}^{T}\left(l_{\theta\theta,t}\left(\varepsilon_{t}\left(\hat{\theta}_{T}\right),\hat{\theta}_{T}\right)\right)$.
Under an additional condition, the outer product of the score can
be used as well:
\begin{thm}
If in addition to the assumptions of the above Theorem \ref{thm:generic_mle},
we assume that 
\[
\mathbb{E}\left(\sup_{\theta\in\Theta_{0}}\left\Vert l_{\theta,t}\left(\varepsilon_{t}(\theta),\theta\right)\right\Vert ^{2}\right)<\infty
\]
then we obtain that
\[
B_{T}=\frac{1}{T}\sum_{t=1}^{T}\left[l_{\theta,t}\left(\varepsilon_{t}\left(\hat{\theta}_{T}\right),\hat{\theta}_{T}\right)l_{\theta,t}\left(\varepsilon_{t}\left(\hat{\theta}_{T}\right),\hat{\theta}_{T}\right)'\right]\xrightarrow{p}\mathbb{E}\left[l_{\theta,t}\left(\varepsilon_{t}\left(\hat{\theta}_{T}\right),\hat{\theta}_{T}\right)l_{\theta,t}\left(\varepsilon_{t}\left(\hat{\theta}_{T}\right),\hat{\theta}_{T}\right)'\right].
\]
\end{thm}

\subsection{Parameter Space, Log-Likelihood Function, and Low-Level Assumptions}

In this section, we specialize the generic theorem stated in the previous
section for the problem at hand. First, we describe the parameter
space on which we optimize the log-likelihood function. Second, we
make assumptions on the densities of the components of $\varepsilon_{t}$.
This allows us to provide explicit expressions for the individual
contributions to the standardized log-likelihood function and its
first partial derivatives\footnote{The expression for the second partial derivatives as well as the tedious
but straightforward derivations are deferred to the Online Appendix.}. Third, we state integrability and dominance conditions on the first
and second partial derivatives of the densities of the components
of $\varepsilon_{t}$. Last, we verify that Assumptions \ref{assu:CLT}
and \ref{assu:ULLN} below (together with Assumptions \ref{assu:non_gaussianIID},
\ref{assu:densities}, and \ref{assu:paramSpace}) imply the ones
in the generic Theorem \ref{thm:generic_mle} and are thus sufficient
for consistency and local asymptotic normality of the MLE.

In order to introduce the parameter space for the SVARMA parameters,
we define $\pi=\left(\pi_{2},\pi_{3}\right)$ where  $\pi_{2}=vec\left(a_{1},\ldots,a_{p}\right)$,
and $\pi_{3}=vec\left(b_{1},\ldots,b_{q}\right)$. Compared with \citet{LMS_svarIdent16}
there is an additional sub-vector $\pi_{3}$ for the MA parameters
and we abstract in our model from the mean by setting it equal to
zero.
\begin{assumption}
\label{assu:paramSpace}The true parameter value $\theta_{0}$ belongs
to the permissible parameter space $\Theta=\Theta_{\pi}\times\Theta_{\beta}\times\Theta_{\sigma}\times\Theta_{\lambda},$
where
\begin{enumerate}
\item $\Theta_{\pi}=\Theta_{\pi_{2}}\times\Theta_{\pi_{3}}$ with $\Theta_{\pi_{2}}\subseteq\mathbb{R}^{n^{2}p}$
and $\Theta_{\pi_{3}}\subseteq\mathbb{R}^{n^{2}q}$ are such that
condition \eqref{eq:stability}, \eqref{eq:invertibility}, the coprimeness
assumption and the full rank assumption on $\left(a_{p},b_{q}\right)$
are satisfied, and
\item $\Theta_{\beta}=vecd\text{°}\left(\mathcal{B}\right)=\left\{ \beta\in\mathbb{R}^{n(n-1)}\,|\,\beta=vecd\text{°}\left(B\right)\text{ for some }B\in\mathcal{B}\right\} $.
The vector $\beta$ collects the off-diagonal elements of $B$.
\item For the scalings, $\Theta_{\sigma}=\mathbb{R}_{+}^{n}$ holds, and
\item for the additional parameters appearing in the component densities,
we have $\Theta_{\lambda}=\Theta_{\lambda_{1}}\times\cdots\times\Theta_{\lambda_{n}}\subseteq\mathbb{R}^{d}$
with $\Theta_{\lambda_{i}}\subseteq\mathbb{R}^{d_{i}}$ open for every
$i\in\left\{ 1,\ldots,n\right\} $ and $d=d_{1}+\cdots+d_{n}$.
\end{enumerate}
\end{assumption}
We also introduce the non-singleton compact and convex subset $\Theta_{0}=\Theta_{0,\pi}\times\Theta_{0,\beta}\times\Theta_{0,\sigma}\times\Theta_{0,\lambda}$
of the interior of $\Theta$ which contains the true parameter value
$\theta_{0}$.

Regarding the component densities of the i.i.d. shock process $\left(\varepsilon_{t}\right)$,
we have
\begin{assumption}
\label{assu:densities}For each $i\in\left\{ 1,\ldots,n\right\} $
the distribution of the error term $\varepsilon_{i,t}$ has a (Lebesgue)
density $f_{i,\sigma_{i}}\left(x;\lambda_{i}\right)=\sigma_{i}^{-1}f_{i}\left(\sigma_{i}^{-1}x;\lambda_{i}\right)$
which may also depend on a parameter vector $\lambda_{i}\in\mathbb{R}^{d_{i}}$.
\end{assumption}
Thus, the individual contributions in the (standardized) log-likelihood
function \eqref{eq:likelihood} are 
\begin{equation}
l_{t}\left(\theta\right)=\sum_{i=1}^{n}\log\left[f_{i}\left(\sigma_{i}^{-1}\iota_{i}^{'}B\left(\beta\right)^{-1}u_{t}\left(\theta\right);\lambda_{i}\right)\right]-\log\left\{ \left|\det\left[B\left(\beta\right)\right]\right|\right\} -\sum_{i=1}^{n}\log\left(\sigma_{i}\right),\label{eq:likelihood_individual}
\end{equation}
where $u_{t}\left(\theta\right)=y_{t}-a_{1}y_{t-1}-\cdots-a_{p}y_{t-p}-b_{1}B\left(\beta\right)\varepsilon_{t-1}\left(\theta\right)-\cdots-b_{q}B\left(\beta\right)\varepsilon_{t-q}\left(\theta\right)$.

The expressions for the partial derivatives of the individual contributions
to the standardized log-likelihood function are given as 
\begin{align*}
\frac{\partial l_{t}\left(\theta\right)}{\partial\pi_{2}} & =-x_{b,t-1}\left(\theta\right)B'\left(\beta\right)^{-1}\Sigma^{-1}e_{x,t}\left(\theta\right)\\
\frac{\partial l_{t}\left(\theta\right)}{\partial\pi_{3}} & =-w_{b,t-1}\left(\theta\right)^{'}\Sigma^{-1}e_{x,t}\left(\theta\right).\\
\frac{\partial l_{t}\left(\theta\right)}{\partial\beta} & =-H'\sum_{i=1}^{q}\left(B\left(\beta\right)^{-1}u_{t-i}\left(\theta\right)\otimes b_{i}'B'\left(\beta\right)^{-1}\Sigma^{-1}e_{x,t}\left(\theta\right)\right)\\
 & \qquad-H'\left(B\left(\beta\right)^{-1}u_{t}\left(\theta\right)\otimes B'\left(\beta\right)^{-1}\Sigma^{-1}e_{x,t}\left(\theta\right)\right)\\
 & \qquad-H'vec\left(B'\left(\beta\right)^{-1}\right)\\
\frac{\partial}{\partial\sigma}l_{t}\left(\theta\right) & =-\Sigma^{-2}\left[e_{x,t}\left(\theta\right)\odot\varepsilon_{t}\left(\theta\right)+\sigma\right]\\
\frac{\partial}{\partial\lambda}l_{t}\left(\theta\right) & =e_{\lambda,t}\left(\theta\right)
\end{align*}

where $x_{b,t-1}\left(\theta\right)=\left(x_{t-1}\otimes b'(z)^{-1}\right)$,
$w_{b,t-1}\left(\theta\right)=\left(w_{t-1}\left(\theta\right)\otimes b'(z)^{-1}\right)$,
$w_{t-1}\left(\theta\right)=\left(u'_{t-1}\left(\theta\right),\ldots,u'_{t-q}\left(\theta\right)\right)^{'}$,
the matrix $H\in\mathbb{R}^{n^{2}\times n(n-1)}$ consisting of zeros
and ones is implicitly defined by $vec\left(B(\beta)\right)=H\beta+vec\left(I_{n}\right)$
for $B$ in $\mathcal{B}$, and 
\[
e_{i,x,t}(\theta)=\frac{\partial}{\partial x}\log\left[f_{i}\left(\sigma_{i}^{-1}\iota_{i}^{'}B\left(\beta\right)^{-1}u_{t}\left(\theta\right);\lambda_{i}\right)\right]=\frac{f_{i,x}\left(\sigma_{i}^{-1}\varepsilon_{i,t}\left(\theta\right);\lambda_{i}\right)}{f_{i}\left(\sigma_{i}^{-1}\varepsilon_{i,t}\left(\theta\right);\lambda_{i}\right)}
\]
and 
\[
e_{i,\lambda_{i},t}(\theta)=\frac{\partial}{\partial\lambda_{i}}\log\left[f_{i}\left(\sigma_{i}^{-1}\iota_{i}^{'}B\left(\beta\right)^{-1}u_{t}\left(\theta\right);\lambda_{i}\right)\right]=\frac{f_{i,\lambda}\left(\sigma_{i}^{-1}\varepsilon_{i,t}\left(\theta\right);\lambda_{i}\right)}{f_{i}\left(\sigma_{i}^{-1}\varepsilon_{i,t}\left(\theta\right);\lambda_{i}\right)},
\]
with $f_{i,x}\left(x;\lambda_{i}\right)=\frac{\partial}{\partial x}f_{i}\left(x;\lambda_{i}\right)$
and $f_{i,\lambda_{i}}\left(x;\lambda_{i}\right)=\frac{\partial}{\partial\lambda_{i}}f_{i}\left(x;\lambda_{i}\right)$.
Evaluated at the truth, i.e. $\theta=\theta_{0}$, we have that $\varepsilon_{i,t}\left(\theta_{0}\right)=\varepsilon_{i,t}$
and 
\[
e_{i,x,t}=e_{i,x,t}(\theta_{0})=\left.\frac{\partial}{\partial x}\log\left[f_{i}\left(\sigma_{i}^{-1}\iota_{i}^{'}B\left(\beta\right)^{-1}u_{t}\left(\pi\right);\lambda_{i}\right)\right]\right|_{\theta=\theta_{0}}=\frac{f_{i,x}\left(\sigma_{i}^{-1}\varepsilon_{i,t};\lambda_{i,0}\right)}{f_{i}\left(\sigma_{i,0}^{-1}\varepsilon_{i,t};\lambda_{i,0}\right)}.
\]

In order to show that the scores are MDS, that the resulting covariance
matrix is finite at the true parameter point, that the expectation
of the supremum on $\Theta_{0}$ of the Hessian is finite, and that
the expectation of the outer product of the score is equal to the
negative expectation of the Hessian of the individual contribution
to the likelihood, we need the following two assumptions on the first
and second derivatives of the component densities.
\begin{assumption}
\label{assu:CLT}The following conditions hold for $i\in\left\{ 1,\ldots,n\right\} $.
\begin{enumerate}
\item For all $x\in\mathbb{R}$ and all $\lambda_{i}\in\Theta_{0,\lambda_{i}},\ f_{i}\left(x;\lambda_{i}\right)>0$
and $f_{i}\left(x;\lambda_{i}\right)$ is twice continuously differentiable
with respect to $\left(x;\lambda_{i}\right)$. 
\item The function $f_{i,x}\left(x;\lambda_{i,0}\right)$ is integrable
with respect to x, i.e., $\int\left|f_{i,x}\left(x;\lambda_{i,0}\right)\right|dx<\infty$
. 
\item For all $x\in\mathbb{R}$
\[
x^{2}\frac{f_{i,x}^{2}\left(x;\lambda_{i}\right)}{f_{i}^{2}\left(x;\lambda_{i}\right)}\ and\ \frac{\left\Vert f_{i,\lambda_{i}}\left(x;\lambda_{i}\right)\right\Vert ^{2}}{f_{i}^{2}\left(x;\lambda_{i}\right)}
\]
are dominated by $c_{1}\left(1+\left|x\right|^{c_{2}}\right)$ with
$c_{1},c_{2}\geq0$ and $\int\left|x\right|^{c_{2}}f_{i}\left(x;\lambda_{i,0}\right)dx<\infty$
\item $\int\sup_{\lambda_{i}\in\Theta_{0,\lambda_{i}}}\left\Vert f_{i,\lambda_{i}}\left(x;\lambda_{i,0}\right)\right\Vert dx<\infty$.
\end{enumerate}
\end{assumption}
and
\begin{assumption}
\label{assu:ULLN}The following conditions hold for $i\in\left\{ 1,\ldots,n\right\} $.
\begin{enumerate}
\item The functions $f_{i,xx}\left(x;\lambda_{i,0}\right)$ and $f_{i,x\lambda_{i}}\left(x;\lambda_{i,0}\right)$
are integrable with respect to $x$, i.e., 
\[
\int\left|f_{i,xx}\left(x;\lambda_{i,0}\right)\right|dx<\infty\ and\ \int\left\Vert f_{i,x\lambda_{i}}\left(x;\lambda_{i,0}\right)\right\Vert dx<\infty.
\]
\item $\int\sup_{\lambda_{i}\in\Theta_{0,\lambda_{i}}}\left\Vert f_{i,\lambda_{i}\lambda_{i}}\left(x;\lambda_{i,0}\right)\right\Vert dx<\infty$
\item For all $x\in\mathbb{R}$ and all $\lambda_{i}\in\Theta_{0,\lambda_{i}}$,
\[
\frac{f_{i,x}^{2}\left(x;\lambda_{i}\right)}{f_{i}^{2}\left(x;\lambda_{i}\right)}\text{ and }\left|\frac{f_{i,xx}\left(x;\lambda_{i}\right)}{f_{i}\left(x;\lambda_{i}\right)}\right|
\]
are dominated by $a_{0}\left(1+\left|x\right|^{a_{1}}\right)$, 
\[
\left\Vert \frac{f_{i,x\lambda_{i}}\left(x;\lambda_{i}\right)}{f_{i}\left(x;\lambda_{i}\right)}\right\Vert \text{and }\left\Vert \frac{f_{i,x}\left(x;\lambda_{i}\right)}{f_{i}\left(x;\lambda_{i}\right)}\frac{f_{i,\lambda_{i}}\left(x;\lambda_{i}\right)}{f_{i}\left(x;\lambda_{i}\right)}\right\Vert 
\]
are dominated by $a_{0}\left(1+\left|x\right|^{a_{2}}\right)$, 
\[
\left\Vert \frac{f_{i,\lambda_{i}}\left(x;\lambda_{i}\right)}{f_{i}\left(x;\lambda_{i}\right)}\right\Vert ^{2}\text{and }\left\Vert \frac{f_{i,\lambda_{i}\lambda_{i}}\left(x;\lambda_{i}\right)}{f_{i}\left(x;\lambda_{i}\right)}\right\Vert 
\]
are dominated by $a_{0}\left(1+\left|x\right|^{a_{3}}\right),$ with
$a_{0},a_{1},a_{2},a_{3}\geq0$ such that $\int\left(\left|x\right|^{2+a_{1}}+\left|x\right|^{1+a_{2}}+\left|x\right|^{a_{3}}\right)f_{i}\left(x;\lambda_{i,0}\right)dx<\infty$.
\end{enumerate}
\end{assumption}
 In combination, these assumptions allow to prove (in the Online
Appendix)
\begin{thm}
\label{thm:MLE}Under assumptions 2-5, there exists a sequence of
maximizers $\hat{\theta}_{T}$ of (\ref{eq:likelihood}) such that
$\sqrt{T}\left(\hat{\theta}_{T}-\theta_{0}\right)$ converges in distribution
to $\mathcal{N}\left\{ 0,\mathbb{E}\left[l_{\theta,t}\left(\theta_{0}\right)l_{\theta,t}'\left(\theta_{0}\right)\right]^{-1}\right\} $.
\end{thm}

\section{Empirical Application}

\subsection{Impulse Response Functions}

Often, the goal of macroeconometric analyses is gaining an understanding
of the impact of structural economic shocks on the observable variables.
This is usually done through analysis of the impulse response function
(IRF) or the analysis of variance decompositions (see e.g. \citet[Chapter 9.4 and 11.7 ]{luet05}
and \citet[Chapter 4]{KilianLut17}).

We comment on the differences between obtaining them from SVARMA or
SVAR representations. First, note that calculation of the IRF is as
straightforward as in the SVAR case after one has obtained the estimates
of the structural parameters. One possibility is representing the
system in state space form \citep[page 15]{HannanDeistler12}. Then,
the impulse responses and variance decompositions are obtained in
the same way as in the SVAR case, see e.g. \citep[page 108]{KilianLut17}.

If, furthermore, the object of interest is the impulse response function
it is hard to come up with reasons favoring SVAR models over SVARMA
models. While in theory one may approximate SVARMA models (or even
``infinite VAR models'') by SVAR models, it is well known that the
approximation is bad in many practically relevant cases. This was
emphasized in a macroeconometric context by \citet{Ravenna07} and
\citet{PoskittYao17}. \citet{Ravenna07} decomposes the error when
SVARMA models are approximated by SVAR models into a truncation error
and an identification error, pertaining to the parameters describing
the economic shocks. \citet{PoskittYao17} decompose the truncation
error introduced in \citet{Ravenna07} further into an estimation
and approximation error and argue that both are large for commonly
used lag lengths and sample sizes. They conclude that ``using VAR($n$)
may not be justified unless $n$ and {[}the sample size{]} $T$ are
enormous''. Obviously, these errors carry over to the IRF which is
a non-linear transformation of the structural parameters.

\subsection{Empirical Application}

To illustrate the developed methods, we estimate a three equation
macroeconomic model and analyze its impulse response function. A more
detailed analysis can be found in the vignette of the R-package associated
with this article.

\subsubsection{Data}

We use the FRED database of the Federal Reserve Bank of St. Louis
and retrieve series for the unemployment gap $n_{t}$, i.e. we subtract
the unemployment rate (UNRATE) from the natural unemployment rate
(NROU), inflation $\pi_{t}$ (lagged differences of GDPDEF), and the
effective federal funds rate $R_{t}$ (FEDFUNDS). The observation
period starts with Q3 1954 and ends with Q1 2019, thus there are 259
observations.
\begin{center}
\begin{figure}[h]
\includegraphics[width=17cm]{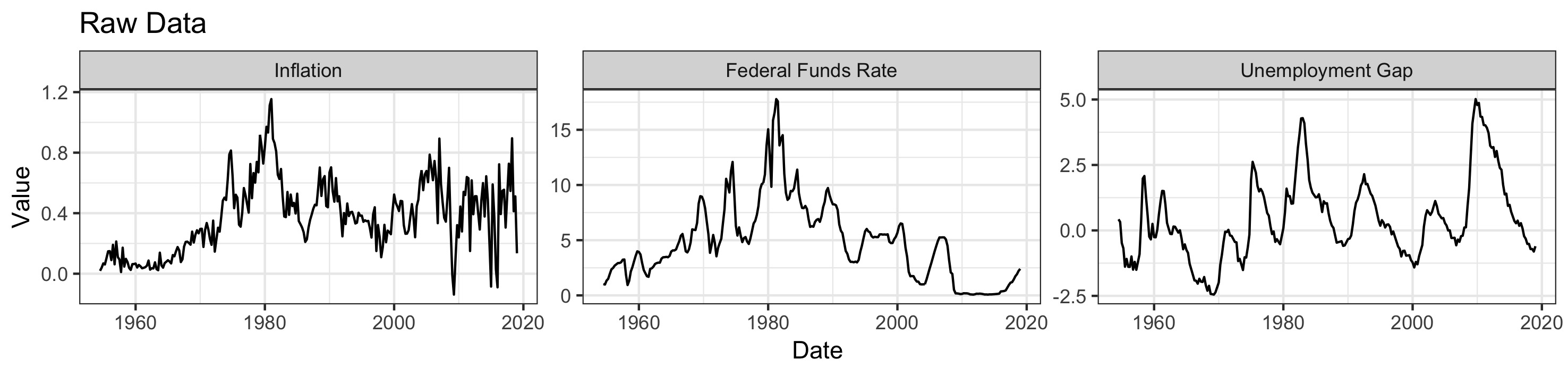}\caption{Raw data}
\end{figure}
\par\end{center}

\subsubsection{Estimation Procedure}

In order to select appropriate integer-valued parameters $p$ and
$q$, we estimate a number of VARMA models with the \texttt{MTS}-package
and select the one with the smallest AIC value. It is worth noting
that neither the \texttt{dse}-package nor the \texttt{MTS}-package
enforces the stability condition \eqref{eq:stability} or the invertibility
condition \eqref{eq:invertibility}. In case of unstable or non-invertible
determinantal roots of the $a(z)$ or $b(z)$ matrix polynomials,
we mirror these roots outside the unit circle and adjust the error
covariance accordingly. Based on the AIC value, the Ljung-Box test,
and the McLeod-Li test \citet{MahdiMcLeod_R_portes}, we choose $p=2$
and $q=2$. Moreover, the distribution of the shocks of the initial
model in Figure \ref{fig:shocks_initial_model} and Figure \ref{fig:qqplot_initial_model}
suggest, and the Jarque-Bera test indicates that the individual series
are not normally distributed.

\textcolor{red}{}
\begin{figure}
\textcolor{red}{\includegraphics[height=6cm]{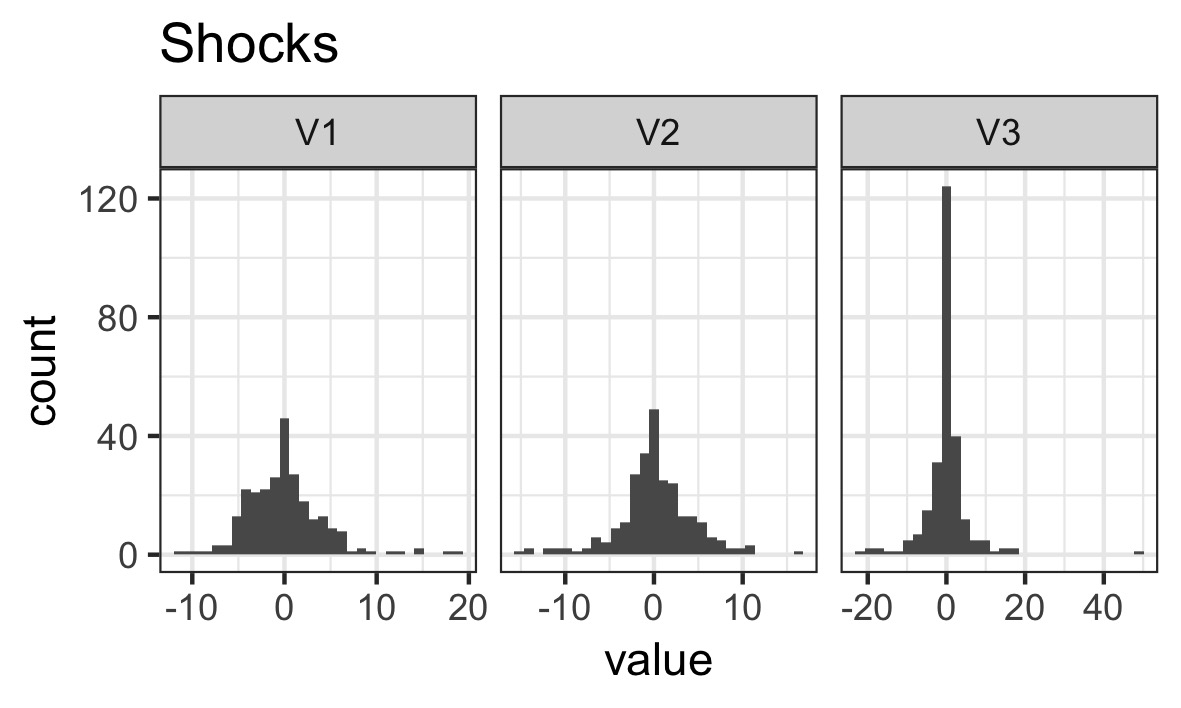}}

\caption{\textcolor{red}{\label{fig:shocks_initial_model}}Histogram of shocks
$\hat{\varepsilon}_{t}=\Sigma^{-1}B^{-1}\hat{a}(z)^{-1}\hat{b}(z)y_{t}$
based on initial model}
\end{figure}

\begin{figure}
\includegraphics[height=6cm]{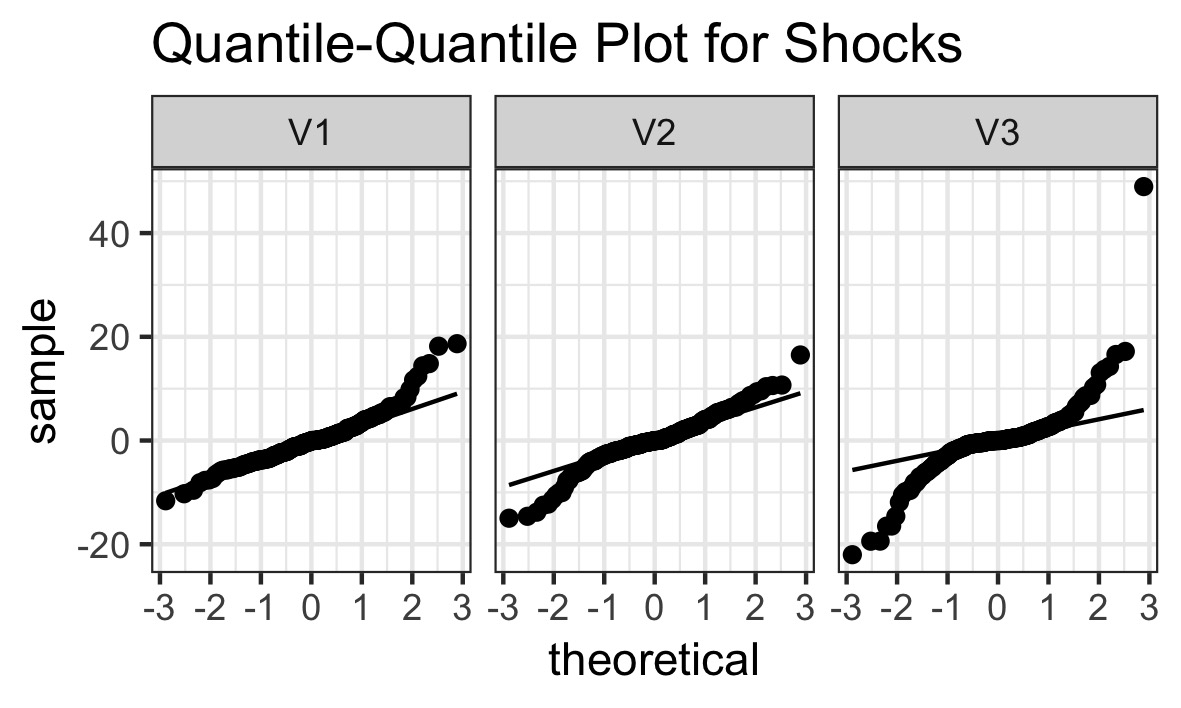}

\caption{\label{fig:qqplot_initial_model}Quantile-quantile plot of shocks
$\hat{\varepsilon}_{t}=\Sigma^{-1}B^{-1}\hat{a}(z)^{-1}\hat{b}(z)y_{t}$
based on initial model}
\end{figure}

We proceed thus under the assumption that the errors of the VARMA(2,2)
model are independent and not normally distributed. The individual
series seem leptokurtic and we assume that they follow a Laplace distribution.
The full conditional likelihood is subsequently estimated using the
same identification scheme as in \citet{LMS_svarIdent16} for fixing
a particular permutation and scaling. The estimates for $B$ and $\sigma$
are 
\[
\hat{B}=\begin{pmatrix}1 & 0.1224 & -0.1282\\
-0.0168 & 1 & 0.0107\\
0.0280 & 0.175 & 1
\end{pmatrix}
\]
and $\hat{\sigma}=\left(0.0685,\,0.0315,\,0.14\right)$, the respective
(bootstrapped) standard deviations are 
\[
\hat{\mathbb{V}}_{bs}\left(\hat{B}\right)=\begin{pmatrix}0 & 0.1204 & 0.01496\\
0.02816 & 0 & 0.0156\\
0.06903 & 0.0928 & 0
\end{pmatrix}
\]
and $\hat{\mathbb{V}}_{bs}\left(\hat{\sigma}\right)=\left(0.00427,\,0.00199,\,0.01419\right)$.

\subsubsection{Impulse Response Function}

In order to interpret the results, we calculate the impulse response
function together with bootstrapped confidence intervals (1000 bootstrap
replications).\textcolor{red}{}
\begin{figure}[H]
\includegraphics[width=16cm]{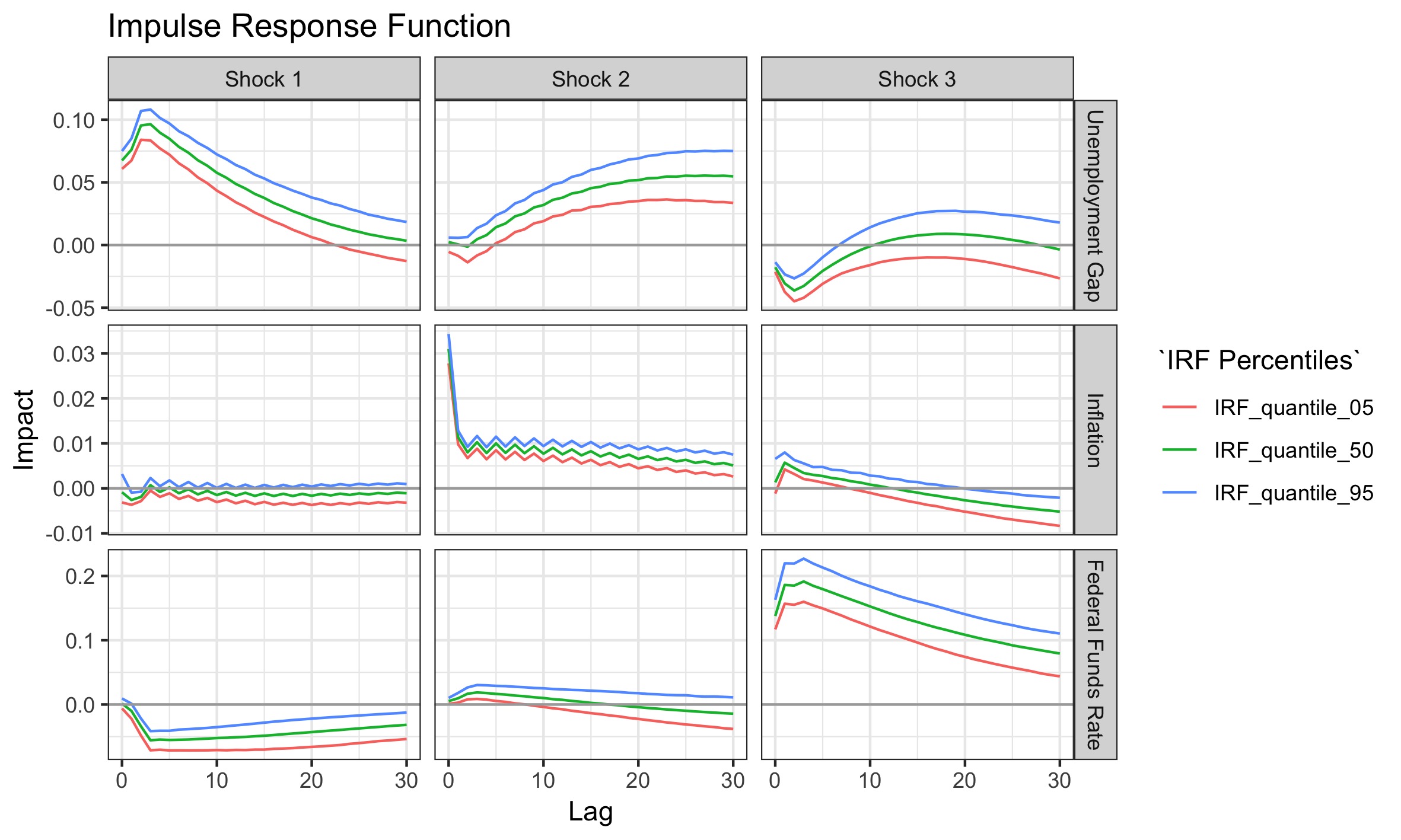}

\caption{Impulse responses}
\end{figure}

We interpret the third shock as a (in the longer term) contractionary
monetary policy shock since it is the only one to which the interest
rate reacts significantly (judging from the confidence bands). In
the medium term, the unemployment gap widens (the economy shrinks)
and initially there is a positive response of inflation to Shock 3.
The first shock is interpreted as negative demand shock because the
economy shrinks (widening of the unemployment gap) and prices fall
slightly. We interpret the remaining Shock 2 as negative supply shock
since the economy shrinks (widening unemployment gap) with increasing
prices.

\section{Acknowledgements}

Financial support by the Research Funds of the University of Helsinki
as well as by funds of the Oesterreichische Nationalbank (Austrian
Central Bank, Anniversary Fund, project number: 17646) is gratefully
acknowledged. For bootstrapping, the \textit{Finnish Grid and Cloud
Infrastructure} with persistent identifier \textit{urn:nbn:fi:research-infras-2016072533}
was used. Juho Koistinen, Mika Meitz, Markku Lanne, and Wolfgang Scherrer
provided helpful comments on various versions of this article.

\section{Conclusion}

In this article, we showed that stable and invertible SVARMA models
\eqref{eq:system} driven by independent and non-Gaussian shocks are
identifiable up to permutation and scaling. This result extends the
identifiability results regarding structural VAR models in \citet{LMS_svarIdent16}.
SVARMA models capture (macroeconomic) dynamics more parsimoniously
than SVAR models and are therefore advantageous in situations with
relatively small sample sizes.

\pagebreak{}

\appendix
\setcounter{page}{1}
\renewcommand{\thepage}{A-\arabic{page}} 
\newgeometry{left=1cm,right=1cm,top=2cm,bottom=2cm}

\section{Introduction to the Online Appendix}

We obtain the first and second partial derivatives of the individual
contributions $l_{t}(\theta)$ to the standardized log-likelihood
function $L_{T}\left(\theta\right)$ and show that Assumptions 2 to
5 imply that the premises of Theorem \ref{thm:generic_mle} are satisfied.
The rest of this Online Appendix is structured as follows.

After introducing some notation in the rest of this section, we calculate
the first partial derivatives of $l_{t}(\theta)$ in Section \ref{sec:first_partial_deriv}.
In Section \ref{sec:score_MDS} we show that the first partial derivatives
of $l_{t}(\theta)$ are an MDS. In Section \ref{sec:opg_expression}
we derive the expression for the expectation of the outer product
of the score at the true parameter value $\theta_{0}$ and in Section
\ref{sec:opg_finite} we verify that it is finite. This verifies that
the score satisfies a CLT for MDS. 

In Section \ref{sec:hessian_expr} we derive the matrix of second
partial derivatives of $l_{t}\left(\theta\right)$ while in Section
\ref{sec:hessian_ulln} we verify that the expectation of the supremum
(over a compact and convex set in the parameter space which contains
the true parameter value) of the Hessian is finite. This verifies
that the Hessian satisfies a ULLN.

In Section \ref{sec:hessian_equal_opg} it is verified that the expectation
of the Hessian equals the expectation of the outer product of the
score, evaluated respectively at the true parameter value $\theta_{0}$.
Together with the above, all statements in Theorem \ref{thm:generic_mle}
are verified. 

\subsection{Notation and System Representations}

The individual contribution at time $t$ to the (standardized) log-likelihood
function, i.e. equation \eqref{eq:likelihood_individual}, is here
repeated as

\[
l_{t}\left(\theta\right)=\sum_{i=1}^{n}\log\left[f_{i}\left(\sigma_{i}^{-1}\varepsilon_{i,t}\left(\theta\right);\lambda_{i}\right)\right]-\log\left\{ \det\left[B\left(\beta\right)\right]\right\} -\sum_{i=1}^{n}\log\left(\sigma_{i}\right),
\]
where $\varepsilon_{i,t}\left(\theta\right)=\iota_{i}^{'}B\left(\beta\right)^{-1}u_{t}\left(\theta\right).$

\paragraph{Derivatives of the component densities.}

For the first partial derivatives of $l_{t}\left(\theta\right)$,
the expressions
\[
e_{i,x,t}(\theta)=\frac{\partial}{\partial x}\log\left[f_{i}\left(\sigma_{i}^{-1}\iota_{i}'B\left(\beta\right)^{-1}u_{t}\left(\theta\right);\lambda_{i}\right)\right]=\frac{f_{i,x}\left(\sigma_{i}^{-1}\iota_{i}'B\left(\beta\right)^{-1}u_{t}\left(\theta\right);\lambda_{i}\right)}{f_{i}\left(\sigma_{i}^{-1}\iota_{i}'B'\left(\beta\right)^{-1}u_{t}\left(\theta\right);\lambda_{i}\right)}
\]
and
\[
e_{i,\lambda_{i},t}(\theta)=\frac{\partial}{\partial\lambda_{i}}\log\left[f_{i}\left(\sigma_{i}^{-1}\iota_{i}'B\left(\beta\right)^{-1}u_{t}\left(\theta\right);\lambda_{i}\right)\right]=\frac{f_{i,\lambda}\left(\sigma_{i}^{-1}\iota_{i}'B\left(\beta\right)^{-1}u_{t}\left(\theta\right);\lambda_{i}\right)}{f_{i}\left(\sigma_{i}^{-1}\iota_{i}'B\left(\beta\right)^{-1}u_{t}\left(\theta\right);\lambda_{i}\right)},
\]
where $f_{i,x}\left(x;\lambda_{i}\right)=\frac{\partial}{\partial x}f_{i}\left(x;\lambda_{i}\right)$
and $f_{i,\lambda_{i}}\left(x;\lambda_{i}\right)=\frac{\partial}{\partial\lambda_{i}}f_{i}\left(x;\lambda_{i}\right)$
will be used extensively. The corresponding versions for all components
are $e_{x,t}\left(\theta\right)=\left(e_{1,x,t}\left(\theta\right),\ldots,e_{n,x,t}\left(\theta\right)\right)'$
of dimension $n$ and $e_{\lambda,t}\left(\theta\right)=\left(e_{1,\lambda_{1},t}'\left(\theta\right),\ldots,e_{n,\lambda_{n},t}'\left(\theta\right)\right)'$
of dimension $d=d_{1}+\cdots+d_{n}$.

For the second partial derivatives of $l_{t}\left(\theta\right)$,
the expressions
\[
e_{i,xx,t}\left(\theta\right)=\frac{\partial e_{i,x,t}\left(\theta\right)}{\partial x}=\left(\frac{f_{i,xx}f_{i}-f_{i,x}^{2}}{f_{i}^{2}}\right)\left(\sigma_{i}^{-1}\iota_{i}'B\left(\beta\right)^{-1}u_{t}\left(\theta\right);\lambda_{i}\right),
\]
\[
e_{i,x\lambda_{i},t}\left(\theta\right)=\frac{\partial e_{i,x,t}\left(\theta\right)}{\partial\lambda_{i}}=\left(\frac{f_{i,x\lambda}f_{i}-f_{i,\lambda_{i}}f_{i,x}^{2}}{f_{i}^{2}}\right)\left(\sigma_{i}^{-1}\iota_{i}'B\left(\beta\right)^{-1}u_{t}\left(\theta\right);\lambda_{i}\right),
\]
and 
\[
e_{i,\lambda_{i}\lambda_{i},t}(\theta)=\frac{\partial e_{i,\lambda_{i},t}(\theta)}{\partial\lambda_{i}'}=\left(\frac{f_{i,\lambda_{i}\lambda_{i}}f_{i}-f_{i,\lambda_{i}}f_{i,\lambda_{i}}'}{f_{i}^{2}}\right)\left(\sigma_{i}^{-1}\iota_{i}'B\left(\beta\right)^{-1}u_{t}\left(\theta\right);\lambda_{i}\right),
\]
are important. Here, $e_{xx,t}\left(\theta\right)=diag\left(e_{1,xx,t}\left(\theta\right),\ldots,e_{n,xx,t}\left(\theta\right)\right)$
is a diagonal matrix of dimension $n$ and $e_{\lambda\lambda,t}(\theta)=diag\left(e_{1,\lambda_{1}\lambda_{1},t}(\theta),\ldots,e_{n,\lambda_{n}\lambda_{n},t}(\theta)\right)$
is a block diagonal matrix with blocks of size $d_{i}$. The notation
$\frac{\partial e_{i,x,t}\left(\theta_{0}\right)}{\partial x}:=\left.\frac{\partial e_{i,x,t}\left(\theta\right)}{\partial x}\right|_{\theta=\theta_{0}}$
is used to denote the derivative evaluated at a particular point.

\paragraph{Two different ways to express the partial derivatives of $u_{t}\left(\theta\right)$.}

The observations may be represented at one particular point in time
or as a system containing all observations $\left(y_{1},\ldots,y_{T}\right)$
as well as starting values $\left(y_{1-p},\ldots,y_{0}\right)$. The
starting values for the process $\left(u_{t}\right)$ are set to zero,
i.e. $\left(u_{1-q},\ldots,u_{0}\right)=0$. For simplicity, we also
set the starting values $\left(y_{1-p},\ldots,y_{0}\right)$ equal
to zero. If clarity of presentation is not affected, we use $x_{t-1}=\left(y'_{t-1},\ldots,y'_{t-p}\right)^{'}$
of dimension $np$ and $w_{t-1}\left(\theta\right)=\left(u'_{t-1}\left(\theta\right),\ldots,u'_{t-q}\left(\theta\right)\right)^{'}$
of dimension $nq$ as shorthand notation.

For one particular point in time, we have

\begin{equation}
u_{t}\left(\theta\right)=y_{t}-\left(a_{1},\ldots,a_{p}\right)\begin{pmatrix}y_{t-1}\\
\vdots\\
y_{t-p}
\end{pmatrix}-\left(b_{1},\ldots,b_{q}\right)\begin{pmatrix}u_{t-1}\left(\theta\right)\\
\vdots\\
u_{t-q}\left(\theta\right)
\end{pmatrix}\label{eq:ut}
\end{equation}

for $t\in\left\{ 1,\ldots,T\right\} $.

All observations can be written as

\begin{equation}
\left(y_{1}\cdots y_{T}\right)-a_{1}\left(y_{0}\cdots y_{T-1}\right)-\cdots-a_{p}\left(y_{1-p}\cdots y_{T-p}\right)=\left(u_{1}\left(\theta\right)\cdots u_{T}\left(\theta\right)\right)+b_{1}+\cdots+b_{q}\left(u_{1-q}\left(\theta\right)\cdots u_{T-q}\left(\theta\right)\right).\label{eq:ut_system_reinsel}
\end{equation}
Defining the matrix 
\[
L=\begin{pmatrix}0 & \cdots &  & \cdots & 0\\
1 & 0 & \cdots &  & \vdots\\
0 & 1 & 0\\
\vdots & \ddots & \ddots & \ddots & \vdots\\
0 & \cdots & 0 & 1 & 0
\end{pmatrix}\in\mathbb{R}^{T\times T}
\]
corresponding to the (non-invertible) lag operator such that 
\[
L\begin{pmatrix}u_{1}'\left(\theta\right)\\
u_{2}'\left(\theta\right)\\
\vdots\\
u_{T}'\left(\theta\right)
\end{pmatrix}=\begin{pmatrix}0_{1\times n}\\
u_{1}'\left(\theta\right)\\
\vdots\\
u_{T-1}'\left(\theta\right)
\end{pmatrix},
\]
equation \eqref{eq:ut_system_reinsel} can be written as 
\[
\left(y_{1}\cdots y_{T}\right)-a_{1}\left(y_{0}\cdots y_{T-1}\right)-\cdots-a_{p}\left(y_{1-p}\cdots y_{T-p}\right)=\left(u_{1}\left(\theta\right)\cdots u_{T}\left(\theta\right)\right)+b_{1}\left(u_{1}\left(\theta\right)\cdots u_{T}\left(\theta\right)\right)L'+\cdots+b_{q}\left(u_{1}\left(\theta\right)\cdots u_{T}\left(\theta\right)\right)\left(L'\right)^{q}.
\]
Vectorizing equation (\ref{eq:ut_system_reinsel}) leads to {\scriptsize{}
\begin{gather}
vec\left(y_{1}\cdots y_{T}\right)-\left[\begin{pmatrix}y_{0}'\\
y_{1}'\\
\vdots\\
y_{T-1}'
\end{pmatrix}\otimes I_{n},\begin{pmatrix}y_{-1}'\\
y_{0}'\\
\vdots\\
y_{T-2}'
\end{pmatrix}\otimes I_{n},\ldots,\begin{pmatrix}y_{1-p}'\\
y_{2-p}'\\
\vdots\\
y_{T-p}'
\end{pmatrix}\otimes I_{n}\right]\underbrace{vec\left(a_{1},\ldots,a_{p}\right)}_{=\pi_{2}}=\label{eq:ut_system_reinsel_vectorized}\\
=vec\left(u_{1}\left(\theta\right)\cdots u_{T}\left(\theta\right)\right)+\left[L\begin{pmatrix}u_{1}'\left(\theta\right)\\
u_{2}'\left(\theta\right)\\
\vdots\\
u_{T}'\left(\theta\right)
\end{pmatrix}\otimes I_{n},L^{2}\begin{pmatrix}u_{1}'\left(\theta\right)\\
u_{2}'\left(\theta\right)\\
\vdots\\
u_{T}'\left(\theta\right)
\end{pmatrix}\otimes I_{n},\ldots,L^{q}\begin{pmatrix}u_{1}'\left(\theta\right)\\
u_{2}'\left(\theta\right)\\
\vdots\\
u_{T}'\left(\theta\right)
\end{pmatrix}\otimes I_{n}\right]\underbrace{vec\left(b_{1},\ldots,b_{p}\right)}_{=\pi_{3}}\nonumber \\
=\left[I_{Tn}+\sum_{i=1}^{q}\left(L^{i}\otimes b_{i}\right)\right]vec\left(u_{1}\left(\theta\right)\cdots u_{T}\left(\theta\right)\right)\nonumber 
\end{gather}
}where the vectorization formula $vec\left(ABC\right)=\left(C'\otimes A\right)vec(B)$
has been applied to $\left\{ \left[I_{n}\right]\left[a_{j}\right]\left[\left(y_{1-j}\cdots y_{T-j}\right)\right]\right\} $
on the left-hand-side and to $\left(\left[b_{j}\right]\left[\left(u_{1}\left(\theta\right)\cdots u_{T}\left(\theta\right)\right)\right]\left[\left(L'\right)^{j}\right]\right)$
and $\left(\left[I_{n}\right]\left[b_{j}\right]\left[\left(u_{1}\left(\theta\right)\cdots u_{T}\left(\theta\right)\right)\left(L'\right)^{j}\right]\right)$
on the right-hand-side of equation \eqref{eq:ut_system_reinsel}

By using the (conditional maximum likelihood) assumption that $\left(y_{1-p},\ldots,y_{0}\right)$
be zero, we can also vectorize the left-hand-side of equation \eqref{eq:ut_system_reinsel}
as 
\[
vec\left[\left(y_{1}\cdots y_{T}\right)-a_{1}\left(y_{1}\cdots y_{T}\right)L'-\cdots-a_{p}\left(y_{1}\cdots y_{T}\right)\left(L'\right)^{p}\right]=vec\left(y_{1}\cdots y_{T}\right)-\sum_{j=1}^{p}\left(L^{j}\otimes a_{j}\right)vec\left(y_{1}\cdots y_{T}\right)
\]

in order to obtain

\[
\mathcal{B}\begin{pmatrix}u_{1}\left(\theta\right)\\
\vdots\\
u_{T}\left(\theta\right)
\end{pmatrix}=\mathcal{A}\begin{pmatrix}y_{1}\\
\vdots\\
y_{T}
\end{pmatrix}
\]

where 
\[
\mathcal{A}=\left[I_{Tn}-\sum_{i=1}^{p}\left(L^{i}\otimes a_{i}\right)\right]=\begin{pmatrix}I_{n} & 0 &  & \cdots & 0 & \cdots &  & 0\\
-a_{1} & I_{n} & 0 &  &  &  &  & \vdots\\
-a_{2} & -a_{1} & I_{n} & \ddots &  &  &  & 0\\
\vdots & -a_{2} & \ddots & \ddots & 0 &  &  & 0\\
-a_{p} &  & \ddots & -a_{1} & I_{n} & \ddots &  & \vdots\\
0 & -a_{p} &  &  & \ddots & \ddots & 0 & 0\\
\vdots & 0 & \ddots &  & -a_{2} & -a_{1} & I_{n} & 0\\
0 & \cdots & 0 & -a_{p} & \cdots & -a_{2} & -a_{1} & I_{n}
\end{pmatrix}\in\mathbb{R}^{Tn\times Tn}
\]
and 
\[
\mathcal{B}=\left[I_{Tn}+\sum_{i=1}^{q}\left(L^{i}\otimes b_{i}\right)\right]=\begin{pmatrix}I_{n} & 0 &  & \cdots & 0 & \cdots &  & 0\\
b_{1} & I_{n} & 0 &  &  &  &  & \vdots\\
b_{2} & b_{1} & I_{n} & \ddots &  &  &  & 0\\
\vdots & b_{2} & \ddots & \ddots & 0 &  &  & 0\\
b_{q} &  & \ddots & b_{1} & I_{n} & \ddots &  & \vdots\\
0 & b_{q} &  &  & \ddots & \ddots & 0 & 0\\
\vdots & 0 & \ddots &  & b_{2} & b_{1} & I_{n} & 0\\
0 & \cdots & 0 & b_{q} & \cdots & b_{2} & b_{1} & I_{n}
\end{pmatrix}\in\mathbb{R}^{Tn\times Tn}.
\]

\section{\label{sec:first_partial_deriv}Partial Derivatives of $l_{t}\left(\theta\right)$}

In the following subsections, we will derive that 
\begin{align*}
\frac{\partial l_{t}\left(\theta_{0}\right)}{\partial\pi_{2}} & =-x_{b,t-1}\left(\theta_{0}\right)B'\left(\beta_{0}\right)^{-1}\Sigma_{0}^{-1}e_{x,t}\left(\theta_{0}\right)\\
\frac{\partial l_{t}\left(\theta_{0}\right)}{\partial\pi_{3}} & =-w_{b,t-1}\left(\theta_{0}\right)\Sigma_{0}^{-1}e_{x,t}\left(\theta_{0}\right)\\
\frac{\partial l_{t}\left(\theta_{0}\right)}{\partial\beta} & =H'\sum_{i=1}^{q}\left(I_{n}\otimes b_{i}'B'\left(\beta_{0}\right)^{-1}\Sigma_{0}^{-1}\right)\left(\varepsilon_{t-i}\left(\theta_{0}\right)\otimes e_{x,t}\left(\theta_{0}\right)\right)-H'\left(I_{n}\otimes B'\left(\beta_{0}\right)^{-1}\Sigma_{0}^{-1}\right)\left(\varepsilon_{t}\left(\theta_{0}\right)\otimes e_{x,t}\left(\theta_{0}\right)\right)-H'vec\left(B'\left(\beta_{0}\right)^{-1}\right)\\
\frac{\partial}{\partial\sigma}l_{t}\left(\theta_{0}\right) & =-\Sigma_{0}^{-2}\left[e_{x,t}\left(\theta_{0}\right)\odot\varepsilon_{t}\left(\theta_{0}\right)+\sigma_{0}\right]\\
\frac{\partial}{\partial\lambda}l_{t}\left(\theta_{0}\right) & =e_{\lambda,t}\left(\theta_{0}\right)
\end{align*}

\subsection{Partial Derivative with respect to $\pi_{2}$}

\paragraph{Intermediate step for $l_{\pi_{2},t}\left(\theta\right)$.}

We obtain that

\begin{align}
\frac{\partial l_{t}\left(\theta\right)}{\partial\pi_{2}} & =\frac{\partial}{\partial\pi_{2}}\left\{ \sum_{i=1}^{n}\log\left[f_{i}\left(\sigma_{i}^{-1}\iota_{i}^{'}B\left(\beta\right)^{-1}u_{t}\left(\theta\right);\lambda_{i}\right)\right]\right\} =\frac{\partial}{\partial\pi_{2}}\left\{ \sum_{i=1}^{n}\log\left[f_{i}\left(\sigma_{i}^{-1}u_{t}\left(\theta\right)'B'\left(\beta\right)^{-1}\iota_{i};\lambda_{i}\right)\right]\right\} \nonumber \\
 & =\sum_{i=1}^{n}e_{i,x,t}\left(\theta\right)\frac{\partial u_{t}\left(\theta\right)'}{\partial\pi_{2}}\sigma_{i}^{-1}B'\left(\beta\right)^{-1}\iota_{i}\nonumber \\
 & =\frac{\partial u_{t}\left(\theta\right)'}{\partial\pi_{2}}B'\left(\beta\right)^{-1}\Sigma^{-1}e_{x,t}\label{lik_pi2}
\end{align}

\paragraph{The derivative of $u_{t}$ with respect to $\pi_{2}$ for one equation.}

We obtain from vectorizing \eqref{eq:ut} that

\begin{align*}
u_{t}\left(\theta\right) & =y_{t}-\left(a_{1},\ldots,a_{p}\right)\begin{pmatrix}y_{t-1}\\
\vdots\\
y_{t-p}
\end{pmatrix}-\left(b_{1},\ldots,b_{q}\right)\begin{pmatrix}u_{t-1}\left(\theta\right)\\
\vdots\\
u_{t-q}\left(\theta\right)
\end{pmatrix}\\
 & =y_{t}-\left(\left(y'_{t-1},\ldots,y'_{t-p}\right)\otimes I_{n}\right)vec\left(a_{1},\ldots,a_{p}\right)-\left(b_{1},\ldots,b_{q}\right)\begin{pmatrix}u_{t-1}\left(\theta\right)\\
\vdots\\
u_{t-q}\left(\theta\right)
\end{pmatrix}.
\end{align*}
Transposition and differentiation lead to 
\[
u_{t}'\left(\theta\right)=y_{t}'-\pi_{2}'\left(\begin{pmatrix}y_{t-1}\\
\vdots\\
y_{t-p}
\end{pmatrix}\otimes I_{n}\right)-\left(u_{t-1}'\left(\theta\right),\ldots,u_{t-q}'\left(\theta\right)\right)\begin{pmatrix}b_{1}'\\
\vdots\\
b_{q}'
\end{pmatrix}
\]

and

\[
\frac{\partial u_{t}'\left(\theta\right)}{\partial\pi_{2}}=-\left(x_{t-1}\otimes I_{n}\right)-\left(\frac{\partial u_{t-1}'\left(\theta\right)}{\partial\pi_{2}},\ldots,\frac{\partial u_{t-q}'\left(\theta\right)}{\partial\pi_{2}}\right)\begin{pmatrix}b_{1}'\\
\vdots\\
b_{q}'
\end{pmatrix}.
\]

Finally, we may express $\frac{\partial u_{t}\left(\theta\right)}{\partial\pi_{2}'}$
using a lag polynomial, i.e.
\begin{align*}
\left(I_{n}+b_{1}z+\cdots+b_{q}z^{q}\right)\overbrace{\frac{\partial u_{t}\left(\theta\right)}{\partial\pi_{2}'}}^{n\times n^{2}p} & =-\overbrace{\left(x_{t-1}'\otimes I_{n}\right)}^{=\left(n\times n^{2}p\right)}\\
\iff\frac{\partial u_{t}\left(\theta\right)}{\partial\pi_{2}'} & =-b(z)^{-1}\left[x_{t-1}'\otimes I_{n}\right]=-\left[x_{t-1}'\otimes b(z)^{-1}\right].
\end{align*}

\paragraph{The derivative of $u_{t}$ with respect to $\pi_{2}$  for all points
in time.}

Rewriting equation \eqref{eq:ut_system_reinsel_vectorized} as
\begin{align*}
\begin{pmatrix}y_{1}\\
\vdots\\
y_{T}
\end{pmatrix}-\left[\begin{pmatrix}x_{0}'\\
\vdots\\
x_{T-1}'
\end{pmatrix}\otimes I_{n}\right]\underbrace{vec\left(a_{1},\ldots,a_{p}\right)}_{=\pi_{2}} & =\underbrace{\left[I_{Tn}-\sum_{i=1}^{q}\left(L^{i}\otimes b_{i}\right)\right]}_{=\mathcal{B}}vec\left(u_{1}\left(\theta\right)\cdots u_{T}\left(\theta\right)\right),
\end{align*}
transposing it and taking partial derivatives leads to
\begin{align*}
\frac{\partial}{\partial\pi_{2}}\begin{pmatrix}u'_{1}\left(\theta\right) & \cdots & u'_{T}\left(\theta\right)\end{pmatrix} & =-\frac{\partial}{\partial\pi_{2}}\left[\pi'_{2}\left[\begin{pmatrix}x_{0} & \cdots & x_{T-1}\end{pmatrix}\otimes I_{n}\right]\mathcal{B}'^{-1}\right]\\
 & =-\left[\begin{pmatrix}x_{0} & \cdots & x_{T-1}\end{pmatrix}\otimes I_{n}\right]\mathcal{B}'^{-1}.
\end{align*}

It follows that for one point in time, we obtain over a compact set
centered at the true parameter value is finite
\[
\frac{\partial u'_{t}\left(\theta\right)}{\partial\pi_{2}}=-\left[\begin{pmatrix}x_{0} & \cdots & x_{T-1}\end{pmatrix}\otimes I_{n}\right]\underbrace{\left(\mathcal{B}'^{-1}\right)_{\left[\bullet,t\right]}}_{=\mathfrak{b_{t}}}
\]
where $\mathfrak{b}_{t}=\left(\mathcal{B}'^{-1}\right)_{\left[\bullet,t\right]}$
is the $t$-th $\left(Tn\times n\right)$-dimensional block of columns
of the $\left(Tn\times Tn\right)$-dimensional matrix $\mathcal{B}'^{-1}$.
Note that $\mathcal{B}'^{-1}$ is (block-) upper-triangular such that
the product of $\left[\begin{pmatrix}x_{0} & \cdots & x_{T-1}\end{pmatrix}\otimes I_{n}\right]$
with $\mathfrak{b}_{t}$ only involves terms depending on time $t-1$
and earlier. Furthermore, note that the non-zero elements of $\mathfrak{b}_{t}$
correspond to the coefficients of $b(z)^{-1}$ whose (matrix-) norms
are decreasing at an exponential rate.

\paragraph{Result for $l_{\pi_{2},t}\left(\theta\right)$ for one point in time.}

In the expression involving $\left[x_{t.-1}\otimes b'(z)^{-1}\right]$
below, it is unclear to which quantity the lag operator applies. This
is why we introduce addional notation $x_{t-1}\otimes b'(z)^{-1}=x_{b,t-1}\left(\theta\right)$
in a similar way as in \citet{meitzSaikkonen13}.

This implies for the score that 
\begin{align*}
\frac{\partial l_{t}\left(\theta\right)}{\partial\pi_{2}} & =-\left[\begin{pmatrix}x_{0} & \cdots & x_{T-1}\end{pmatrix}\otimes I_{n}\right]\mathfrak{b}_{t}B'\left(\beta\right)^{-1}\Sigma^{-1}e_{x,t}\left(\theta\right)\\
 & =-\left[x_{t.-1}\otimes b'(z)^{-1}\right]B'\left(\beta\right)^{-1}\Sigma^{-1}e_{x,t}\left(\theta\right)\\
 & =-x_{b,t-1}\left(\theta\right)B'\left(\beta\right)^{-1}\Sigma^{-1}e_{x,t}\left(\theta\right).
\end{align*}

\paragraph{Result for $\frac{\partial L_{t}\left(\theta\right)}{\partial\pi_{2}}$.}

The partial derivative of the standardized log-likelihood function
with respect to $\pi_{2}$ is
\begin{align*}
\frac{\partial L_{t}\left(\theta\right)}{\partial\pi_{2}} & =\frac{1}{T}\sum_{i=1}^{T}l_{\pi_{2},t}\left(\theta\right)\\
 & =-\frac{1}{T}\left[\begin{pmatrix}x_{0} & \cdots & x_{T-1}\end{pmatrix}\otimes I_{n}\right]\mathcal{B}'^{-1}\left[\left(I_{T}\otimes\Sigma^{-1}B'\left(\beta\right)^{-1}\right)\begin{pmatrix}e_{x,1}\left(\theta\right)\\
\vdots\\
e_{x,T}\left(\theta\right)
\end{pmatrix}\right].
\end{align*}

\subsection{Partial Derivative with respect to $\pi_{3}$}

\paragraph{Intermediate step for $l_{\pi_{3},t}\left(\theta\right)$.}

As for \eqref{lik_pi2}, we obtain
\begin{align*}
\frac{\partial l_{t}\left(\theta\right)}{\partial\pi_{3}} & =\sum_{i=1}^{n}e_{i,x,t}\left(\theta\right)\frac{\partial u_{t}\left(\theta\right)'}{\partial\pi_{3}}\sigma_{i}^{-1}B'\left(\beta\right)^{-1}\iota_{i}\\
 & =\frac{\partial u_{t}\left(\theta\right)'}{\partial\pi_{3}}B'\left(\beta\right)^{-1}\Sigma^{-1}e_{x,t}\left(\theta\right)
\end{align*}

\paragraph{Applying the product rule to $\frac{\partial u_{t}\left(\theta\right)'}{\partial\pi_{3}}$.}

Note that both $\left(b_{1},\ldots,b_{q}\right)$ and $u_{t-j}\left(\theta\right)=B\left(\beta\right)\varepsilon_{t-j}\left(\theta\right)$
depend on $\pi_{3}$ and that the last term in
\[
u_{t}\left(\theta\right)=y_{t}-\left(a_{1},\ldots,a_{p}\right)\begin{pmatrix}y_{t-1}\\
\vdots\\
y_{t-p}
\end{pmatrix}-\left(b_{1},\ldots,b_{q}\right)\begin{pmatrix}u_{t-1}\left(\theta\right)\\
\vdots\\
u_{t-q}\left(\theta\right)
\end{pmatrix},
\]
can be written as $vec\left[I_{n}\left(b_{1},\ldots,b_{q}\right)\begin{pmatrix}u_{t-1}\left(\theta\right)\\
\vdots\\
u_{t-q}\left(\theta\right)
\end{pmatrix}\right]=\left[w_{t-1}'\left(\theta\right)\otimes I_{n}\right]\pi_{3}$. Thus, we obtain that 
\[
\frac{\partial u_{t}'\left(\theta\right)}{\partial\pi_{3}}=-\left[w_{t-1}\left(\theta\right)\otimes I_{n}\right]-\left[\begin{pmatrix}\frac{\partial u_{t-1}'\left(\theta\right)}{\partial\pi_{3}} & \cdots & \frac{\partial u_{t-q}'\left(\theta\right)}{\partial\pi_{3}}\end{pmatrix}\begin{pmatrix}b_{1}'\\
\vdots\\
b_{q}'
\end{pmatrix}\right].
\]

Transposing this equation and using the lag operator to solve this
equation leads to 
\begin{align*}
\left(I_{n}+b_{1}z+\cdots+b_{q}z^{q}\right)\frac{\partial u_{t}\left(\theta\right)}{\partial\pi_{3}'} & =-\left[w_{t-1}'\left(\theta\right)\otimes I_{n}\right]\\
\iff\frac{\partial u_{t}\left(\theta\right)}{\partial\pi_{3}'} & =-b(z)^{-1}\left[w_{t-1}'\left(\theta\right)\otimes I_{n}\right].
\end{align*}

\paragraph{Applying the product rule to$\frac{\partial\left(u_{1}'\left(\theta\right),\ldots,u_{T}'\left(\theta\right)\right)}{\partial\pi_{3}}$.}

Rewriting equation \eqref{eq:ut_system_reinsel_vectorized} as 
\begin{align*}
\begin{pmatrix}u_{1}\left(\theta\right)\\
\vdots\\
u_{T}\left(\theta\right)
\end{pmatrix} & =-\left[\sum_{i=1}^{q}\left(L^{i}\otimes b_{i}\right)\right]\begin{pmatrix}u_{1}\left(\theta\right)\\
\vdots\\
u_{T}\left(\theta\right)
\end{pmatrix}+\mathcal{A}\begin{pmatrix}y_{1}\\
\vdots\\
y_{T}
\end{pmatrix}+\delta\\
 & =-\left[L\begin{pmatrix}u_{1}'\left(\theta\right)\\
u_{2}'\left(\theta\right)\\
\vdots\\
u_{T}'\left(\theta\right)
\end{pmatrix}\otimes I_{n},L^{2}\begin{pmatrix}u_{1}'\left(\theta\right)\\
u_{2}'\left(\theta\right)\\
\vdots\\
u_{T}'\left(\theta\right)
\end{pmatrix}\otimes I_{n},\ldots,L^{q}\begin{pmatrix}u_{1}'\left(\theta\right)\\
u_{2}'\left(\theta\right)\\
\vdots\\
u_{T}'\left(\theta\right)
\end{pmatrix}\otimes I_{n}\right]\pi_{2}+\mathcal{A}\begin{pmatrix}y_{1}\\
\vdots\\
y_{T}
\end{pmatrix}+\delta,\\
 & =-\left[\begin{pmatrix}u_{0}'\left(\theta\right) & \cdots & u_{1-q}'\left(\theta\right)\\
\vdots &  & \vdots\\
u_{T-1}'\left(\theta\right) & \cdots & u_{T-q}'\left(\theta\right)
\end{pmatrix}\otimes I_{n}\right]\pi_{2}+\mathcal{A}\begin{pmatrix}y_{1}\\
\vdots\\
y_{T}
\end{pmatrix}+\delta,
\end{align*}
transposing and taking derivatives leads to 
\begin{gather*}
\frac{\partial\left(u_{1}'\left(\theta\right),\ldots,u_{T}'\left(\theta\right)\right)}{\partial\pi_{3}}=-\frac{\partial\left(u_{1}'\left(\theta\right),\ldots,u_{T}'\left(\theta\right)\right)}{\partial\pi_{3}}\left[\sum_{i=1}^{q}\left(\left(L'\right)^{i}\otimes b_{i}'\right)\right]-\left[\begin{pmatrix}u_{0}\left(\theta\right) & \cdots & u_{T-1}\left(\theta\right)\\
\vdots &  & \vdots\\
u_{1-q}\left(\theta\right) & \cdots & u_{T-q}\left(\theta\right)
\end{pmatrix}\otimes I_{n}\right]\\
\iff\frac{\partial\left(u_{1}'\left(\theta\right),\ldots,u_{T}'\left(\theta\right)\right)}{\partial\pi_{3}}\underbrace{\left[I_{Tn}+\sum_{i=1}^{q}\left(\left(L'\right)^{i}\otimes b_{i}'\right)\right]}_{=\mathcal{B}'}=-\left[\begin{pmatrix}u_{0}\left(\theta\right) & \cdots & u_{T-1}\left(\theta\right)\\
\vdots &  & \vdots\\
u_{1-q}\left(\theta\right) & \cdots & u_{T-q}\left(\theta\right)
\end{pmatrix}\otimes I_{n}\right]
\end{gather*}
which in turn is equivalent to 
\begin{align*}
\frac{\partial\left(u_{1}'\left(\theta\right),\ldots,u_{T}'\left(\theta\right)\right)}{\partial\pi_{3}} & =-\begin{pmatrix}\begin{pmatrix}u_{1}\left(\theta\right) & \cdots & u_{T}\left(\theta\right)\end{pmatrix}L'\otimes I_{n}\\
\vdots\\
\begin{pmatrix}u_{1}\left(\theta\right) & \cdots & u_{T}\left(\theta\right)\end{pmatrix}\left(L'\right)^{q}\otimes I_{n}
\end{pmatrix}\mathcal{B}'^{-1}\\
 & =-\left(I_{q}\otimes\begin{pmatrix}u_{1}\left(\theta\right) & \cdots & u_{T}\left(\theta\right)\end{pmatrix}\otimes I_{n}\right)\left[\begin{pmatrix}L'\\
\vdots\\
\left(L'\right)^{q}
\end{pmatrix}\otimes I_{n}\right]\mathcal{B}'^{-1}\\
 & =-\left(\begin{pmatrix}u_{0}\left(\theta\right) & \cdots & u_{T-1}\left(\theta\right)\\
\vdots &  & \vdots\\
u_{1-q}\left(\theta\right) & \cdots & u_{T-q}\left(\theta\right)
\end{pmatrix}\otimes I_{n}\right)\mathcal{B}'^{-1}=-\left(\begin{pmatrix}w_{0}\left(\theta\right) & \cdots & w_{T-1}\left(\theta\right)\end{pmatrix}\otimes I_{n}\right)\mathcal{B}'^{-1}=
\end{align*}

With the same notation as for the partial derivative with respect
to $\pi_{2}$, it follows that for one point in time, we obtain 
\[
\frac{\partial u'_{t}\left(\theta\right)}{\partial\pi_{3}}=-\left[\begin{pmatrix}w_{0}\left(\theta\right) & \cdots & w_{T-1}\left(\theta\right)\end{pmatrix}\otimes I_{n}\right]\mathfrak{b_{t}}.
\]

\paragraph{Result for $l_{\pi_{3},t}\left(\theta\right)$ for one point in time.}

Finally, we obtain for the partial derivative (and introduce the notation
$w_{b,t-1}\left(\theta\right)=\left(w_{t-1}\left(\theta\right)\otimes b'(z)^{-1}\right)$
as in the derivation of $l_{\pi_{2},t}\left(\theta\right)$) 
\begin{align*}
\frac{\partial l_{t}\left(\theta\right)}{\partial\pi_{3}} & =-\left[\begin{pmatrix}u_{t-1}\left(\theta\right)\\
\vdots\\
u_{t-q}\left(\theta\right)
\end{pmatrix}\otimes B'\left(\beta\right)^{-1}+\begin{pmatrix}\frac{\partial u_{t-1}'\left(\theta\right)}{\partial\pi_{3}} & \cdots & \frac{\partial u_{t-q}'\left(\theta\right)}{\partial\pi_{3}}\end{pmatrix}\begin{pmatrix}b_{1}'\\
\vdots\\
b_{q}'
\end{pmatrix}B'\left(\beta\right)^{-1}\right]\Sigma^{-1}e_{x,t}\left(\theta\right).\\
 & =-\left[\begin{pmatrix}w_{0}\left(\theta\right) & \cdots & w_{T-1}\left(\theta\right)\end{pmatrix}\otimes I_{n}\right]\mathfrak{b_{t}}\Sigma^{-1}e_{x,t}\left(\theta\right)\\
 & =-\left\{ \left[w_{t-1}\left(\theta\right)\otimes b'(z)^{-1}\right]\right\} \Sigma^{-1}e_{x,t}\left(\theta\right)=-w_{b,t-1}\left(\theta\right)B'\left(\beta\right)^{-1}\Sigma^{-1}e_{x,t}\left(\theta\right).
\end{align*}

Note that $u_{t-i}\left(\theta\right)$ can be expressed as a function
of the observations $y_{t-j},\ j\geq i$.

\paragraph{Result for $\frac{\partial L_{t}\left(\theta\right)}{\partial\pi_{3}}$.}

Finally, we obtain for the partial derivative of the standardized
log-likelihood function with respect to $\pi_{3}$ that
\begin{align*}
\frac{\partial L_{t}\left(\theta\right)}{\partial\pi_{3}} & =\frac{1}{T}\sum_{i=1}^{T}l_{\pi_{3},t}\left(\theta\right)\\
 & =-\frac{1}{T}\left(I_{q}\otimes\begin{pmatrix}u_{1}\left(\theta\right) & \cdots & u_{T}\left(\theta\right)\end{pmatrix}\otimes I_{n}\right)\begin{pmatrix}L'\otimes I_{n}\\
\vdots\\
\left(L'\right)^{q}\otimes I_{n}
\end{pmatrix}\mathcal{B}'^{-1}\left[\left(I_{T}\otimes\Sigma^{-1}B'\left(\beta\right)^{-1}\right)\begin{pmatrix}e_{x,1}\left(\theta\right)\\
\vdots\\
e_{x,T}\left(\theta\right)
\end{pmatrix}\right]\\
 & =-\frac{1}{T}\left[\begin{pmatrix}w_{0}\left(\theta\right) & \cdots & w_{T-1}\left(\theta\right)\end{pmatrix}\otimes I_{n}\right]\mathcal{B}'^{-1}\left[\left(I_{T}\otimes\Sigma^{-1}B'\left(\beta\right)^{-1}\right)\begin{pmatrix}e_{x,1}\left(\theta\right)\\
\vdots\\
e_{x,T}\left(\theta\right)
\end{pmatrix}\right]
\end{align*}

\subsection{Partial Derivative with respect to $\beta$}

\paragraph{Intermediate step for $l_{\beta,t}\left(\theta\right)$.}

By taking the derivative of \eqref{eq:likelihood_individual}, we
obtain for $\beta\in\mathbb{R}^{n(n-1)}${\scriptsize{}
\begin{align*}
\frac{\partial l_{t}\left(\theta\right)}{\partial\beta} & =\frac{\partial}{\partial\beta}\left\{ \sum_{i=1}^{n}\log\left[f_{i}\left(\sigma_{i}^{-1}\iota_{i}^{'}B\left(\beta\right)^{-1}u_{t}\left(\theta\right);\lambda_{i}\right)\right]\right\} -\frac{\partial\log\left\{ \det\left[B\left(\beta\right)\right]\right\} }{\partial\beta}\\
 & =\sum_{i=1}^{n}e_{i,x,t}\left(\theta\right)\sigma_{i}^{-1}\frac{\partial}{\partial\beta}\left(\frac{1}{2}u_{t}'\left(\theta\right)B'\left(\beta\right)^{-1}\iota_{i}+\frac{1}{2}vec\left(\iota_{i}^{'}B\left(\beta\right)^{-1}u_{t}\left(\theta\right)\right)\right)-\frac{1}{\det\left(B\left(\beta\right)\right)}\frac{\partial\det\left(B\left(\beta\right)\right)}{\partial\beta}\\
 & =\sum_{i=1}^{n}e_{i,x,t}\left(\theta\right)\sigma_{i}^{-1}\frac{\partial}{\partial\beta}\left(\frac{1}{2}u_{t}'\left(\theta\right)B'\left(\beta\right)^{-1}\iota_{i}+\frac{1}{2}\underbrace{\left[\left(u_{t}'\left(\theta\right)\otimes\iota_{i}^{'}\right)vec\left(B\left(\beta\right)^{-1}\right)\right]}_{=\text{scalar}}\right)-\frac{1}{\det\left(B\left(\beta\right)\right)}\frac{\partial\det\left(B\left(\beta\right)\right)}{\partial\beta}\\
 & =\sum_{i=1}^{n}e_{i,x,t}\left(\theta\right)\sigma_{i}^{-1}\left\{ \left(\frac{\partial u_{t}'\left(\theta\right)}{\partial\beta}\right)B'\left(\beta\right)^{-1}\iota_{i}+\left[\frac{\partial vec\left(B\left(\beta\right)^{-1}\right)^{'}}{\partial\beta}\left(u_{t}\left(\theta\right)\otimes\iota_{i}\right)\right]\right\} -\frac{1}{\det\left(B\left(\beta\right)\right)}\frac{\partial\det\left(B\left(\beta\right)\right)}{\partial\beta}\\
 & =\left(\frac{\partial u_{t}'\left(\theta\right)}{\partial\beta}\right)B'\left(\beta\right)^{-1}\Sigma^{-1}e_{x,t}\left(\theta\right)+\left[\frac{\partial vec\left(B\left(\beta\right)^{-1}\right)^{'}}{\partial\beta}\left(u_{t}\left(\theta\right)\otimes\Sigma^{-1}e_{x,t}\left(\theta\right)\right)\right]-\frac{1}{\det\left(B\left(\beta\right)\right)}\frac{\partial\det\left(B\left(\beta\right)\right)}{\partial\beta}
\end{align*}
}where we used the following matrix differentiation rules.

\paragraph{Matrix differentiation rules.}

We obtain from \citet{seber08} 17.33(b), page 363, that\footnote{This result can be obtained by taking the derivative of $FF^{-1}=I\Rightarrow F\frac{\partial F^{-1}}{\partial x_{j}}+\frac{\partial F}{\partial x_{j}}F^{-1}=0$
such that $FF^{-1}=I\Rightarrow F\frac{\partial F^{-1}}{\partial x_{j}}+\frac{\partial F}{\partial x_{j}}F^{-1}=0$
results. Vectorization of $\frac{\partial F^{-1}}{\partial x_{j}}=-F^{-1}\frac{\partial F}{\partial x_{j}}F^{-1}$gives
the desired result, see \citet{Harville97} page 366.}
\[
\frac{\partial vec\left(F^{-1}\right)}{\partial x'}=-\left(F'^{-1}\otimes F^{-1}\right)\frac{\partial vec\left(F\right)}{\partial x'}\quad\text{and}\quad\frac{\partial vec\left(F^{-1}\right)'}{\partial x}=-\frac{\partial vec\left(F\right)'}{\partial x}\left(F^{-1}\otimes F'^{-1}\right)
\]
Moreover, we obtain from \citet{seber08}17.26(c), page 361, for the
derivative of the determinant that $\frac{\partial\det\left(Z\right)}{\partial x'}=vec\left[adj\left(Z\right)'\right]'\frac{\partial vec\left(Z\right)}{\partial x'}=\det\left(Z\right)vec\left(Z'^{-1}\right)'\frac{\partial vec\left(Z\right)}{\partial x'}$
or equivalently $\frac{\partial\det\left(Z\right)}{\partial x}=\det\left(Z\right)\frac{\partial vec\left(Z\right)}{\partial x}vec\left(Z'^{-1}\right)$.

\paragraph{Intermediate results for $l_{\beta,t}\left(\theta\right)$ using
matrix differentiation rules.}

Using the results above, we obtain, using additionally $vec\left(B(\beta)\right)=H\beta+vec\left(I_{n}\right)$
and thus $\frac{\partial}{\partial\beta'}vec\left(B(\beta)\right)=H$,
that {\scriptsize{}
\begin{align}
\frac{\partial l_{t}\left(\theta\right)}{\partial\beta} & =\left(\frac{\partial u_{t}'\left(\theta\right)}{\partial\beta}\right)B'\left(\beta\right)^{-1}\Sigma^{-1}e_{x,t}\left(\theta\right)+\left[\frac{\partial vec\left(B\left(\beta\right)^{-1}\right)^{'}}{\partial\beta}\left(u_{t}\left(\theta\right)\otimes\Sigma^{-1}e_{x,t}\left(\theta\right)\right)\right]-\frac{1}{\det\left(B\left(\beta\right)\right)}\frac{\partial\det\left(B\left(\beta\right)\right)}{\partial\beta}\nonumber \\
 & =\left(\frac{\partial u_{t}'\left(\theta\right)}{\partial\beta}\right)B'\left(\beta\right)^{-1}\Sigma^{-1}e_{x,t}\left(\theta\right)-H'\left(B\left(\beta\right)^{-1}\otimes B'\left(\beta\right)^{-1}\right)\left(u_{t}\left(\theta\right)\otimes\Sigma^{-1}e_{x,t}\left(\theta\right)\right)-\frac{1}{\det\left(B\left(\beta\right)\right)}\left(\det\left(B\left(\beta\right)\right)\frac{\partial vec\left(B\left(\beta\right)\right)}{\partial\beta}vec\left(B'\left(\beta\right)^{-1}\right)\right)\nonumber \\
 & =\left(\frac{\partial u_{t}'\left(\theta\right)}{\partial\beta}\right)B'\left(\beta\right)^{-1}\Sigma^{-1}e_{x,t}\left(\theta\right)-H'\left(B\left(\beta\right)^{-1}u_{t}\left(\theta\right)\otimes B'\left(\beta\right)^{-1}\Sigma^{-1}e_{x,t}\left(\theta\right)\right)-H'vec\left(B'\left(\beta\right)^{-1}\right)\label{eq:lik_beta}
\end{align}
}{\scriptsize\par}

\paragraph{The derivative of $u_{t}$ with respect to $\beta$ for one equation.}

From
\[
u_{t}\left(\theta\right)=y_{t}-\left(a_{1},\ldots,a_{p}\right)\begin{pmatrix}y_{t-1}\\
\vdots\\
y_{t-p}
\end{pmatrix}-\left(b_{1},\ldots,b_{q}\right)\begin{pmatrix}u_{t-1}\left(\theta\right)\\
\vdots\\
u_{t-q}\left(\theta\right)
\end{pmatrix}
\]
we obtain immediately
\[
\frac{\partial u_{t}'\left(\theta\right)}{\partial\beta}=-\begin{pmatrix}\frac{\partial u_{t-1}'\left(\theta\right)}{\partial\beta} & \cdots & \frac{\partial u_{t-q}'\left(\theta\right)}{\partial\beta}\end{pmatrix}\begin{pmatrix}b_{1}'\\
\vdots\\
b_{q}'
\end{pmatrix}.
\]
Additionally, an explicit expression for the derivative of $u_{t}\left(\theta\right)=B\left(\beta\right)\varepsilon_{t}\left(\theta\right)=\left(\varepsilon_{t}'\left(\theta\right)\otimes I_{n}\right)vec\left(B\left(\beta\right)\right)$
with respect to $\beta$ can be found as $\frac{\partial u_{t}'\left(\theta\right)}{\partial\beta}=H'\left(\varepsilon_{t}\left(\theta\right)\otimes I_{n}\right)$
and subsequently combined with the quantity above. We thus obtain
\begin{align*}
\frac{\partial u_{t}'\left(\theta\right)}{\partial\beta} & =-H'\left[\left(\varepsilon_{t-1}\left(\theta\right)\otimes I_{n}\right),\ldots,\left(\varepsilon_{t-q}\left(\theta\right)\otimes I_{n}\right)\right]\begin{pmatrix}b_{1}'\\
\vdots\\
b_{q}'
\end{pmatrix}\\
 & =-H'\left[\left(\varepsilon_{t-1}\left(\theta\right),\ldots,\varepsilon_{t-q}\left(\theta\right)\right)\otimes I_{n}\right]\begin{pmatrix}b_{1}'\\
\vdots\\
b_{q}'
\end{pmatrix}\\
 & =-H'\sum_{i=1}^{q}\left(\varepsilon_{t-i}\left(\theta\right)\otimes b_{i}'\right)=-H'\sum_{i=1}^{q}\left(B\left(\beta\right)^{-1}u_{t-i}\left(\theta\right)\otimes b_{i}'\right).
\end{align*}
Note that $\frac{\partial u_{t}'\left(\theta\right)}{\partial\beta}$
depends only on $\mathcal{F}_{t-1}$.

\paragraph{Result for $l_{\beta,t}\left(\theta\right)$ for one point in time.}

The above leads to {\scriptsize{}
\begin{align*}
\frac{\partial l_{t}\left(\theta\right)}{\partial\beta} & =\left(\frac{\partial u_{t}'\left(\theta\right)}{\partial\beta}\right)B'\left(\beta\right)^{-1}\Sigma^{-1}e_{x,t}\left(\theta\right)-H'\left(B\left(\beta\right)^{-1}u_{t}\left(\theta\right)\otimes B'\left(\beta\right)^{-1}\Sigma^{-1}e_{x,t}\left(\theta\right)\right)-H'vec\left(B'\left(\beta\right)^{-1}\right)\\
 & =-\begin{pmatrix}\frac{\partial u_{t-1}'\left(\theta\right)}{\partial\beta} & \cdots & \frac{\partial u_{t-q}'\left(\theta\right)}{\partial\beta}\end{pmatrix}\begin{pmatrix}b_{1}'\\
\vdots\\
b_{q}'
\end{pmatrix}B'\left(\beta\right)^{-1}\Sigma^{-1}e_{x,t}\left(\theta\right)-H'\left(B\left(\beta\right)^{-1}u_{t}\left(\theta\right)\otimes B'\left(\beta\right)^{-1}\Sigma^{-1}e_{x,t}\left(\theta\right)\right)-H'vec\left(B'\left(\beta\right)^{-1}\right)\\
 & =-H'\sum_{i=1}^{q}\left(B\left(\beta\right)^{-1}u_{t-i}\left(\theta\right)\otimes b_{i}'B'\left(\beta\right)^{-1}\Sigma^{-1}e_{x,t}\left(\theta\right)\right)-H'\left(B\left(\beta\right)^{-1}u_{t}\left(\theta\right)\otimes B'\left(\beta\right)^{-1}\Sigma^{-1}e_{x,t}\left(\theta\right)\right)-H'vec\left(B'\left(\beta\right)^{-1}\right)\\
 & =-H'\sum_{i=1}^{q}\left(B\left(\beta\right)^{-1}\otimes b_{i}'\right)\left(u_{t-i}\left(\theta\right)\otimes B'\left(\beta\right)^{-1}\Sigma^{-1}e_{x,t}\left(\theta\right)\right)-H'\left(B\left(\beta\right)^{-1}\otimes B'\left(\beta\right)^{-1}\Sigma^{-1}\right)\left(u_{t}\left(\theta\right)\otimes e_{x,t}\left(\theta\right)\right)-H'vec\left(B'\left(\beta\right)^{-1}\right)\\
 & =-H'\left[B\left(\beta\right)^{-1}\otimes\begin{pmatrix}I_{n} & b_{1}' & \cdots & b_{q}'\end{pmatrix}\right]\left[\begin{pmatrix}u_{t}\left(\theta\right)\\
u_{t-1}\left(\theta\right)\\
\vdots\\
u_{t-q}\left(\theta\right)
\end{pmatrix}\otimes\left(B'\left(\beta\right)^{-1}\Sigma^{-1}e_{x,t}\left(\theta\right)\right)\right]-H'vec\left(B'\left(\beta\right)^{-1}\right)
\end{align*}
}{\scriptsize\par}

\paragraph{Result for $\frac{\partial L_{t}\left(\theta\right)}{\partial\beta}$.}

Finally, we obtain for the partial derivative of the standardized
log-likelihood function with respect to $\beta$ that
\begin{align*}
\frac{\partial L_{t}\left(\theta\right)}{\partial\beta} & =\frac{1}{T}\sum_{t=1}^{T}l_{\beta,t}\left(\theta\right)\\
 & =-\frac{1}{T}H'\left[B\left(\beta\right)^{-1}\otimes\begin{pmatrix}I_{n} & b_{1}' & \cdots & b_{q}'\end{pmatrix}\right]\sum_{t=1}^{T}\left[\begin{pmatrix}u_{t}\left(\theta\right)\\
u_{t-1}\left(\theta\right)\\
\vdots\\
u_{t-q}\left(\theta\right)
\end{pmatrix}\otimes\left(B'\left(\beta\right)^{-1}\Sigma^{-1}e_{x,t}\left(\theta\right)\right)\right]-H'vec\left(B'\left(\beta\right)^{-1}\right)
\end{align*}

{\scriptsize{}}{\scriptsize\par}

\subsection{Partial Derivative with respect to $\sigma$}

Since the individual contribution to the (standardized) log-likelihood
function is

\[
l_{t}\left(\theta\right)=\sum_{i=1}^{n}\log\left[f_{i}\left(\sigma_{i}^{-1}\iota_{i}^{'}B\left(\beta\right)^{-1}u_{t}\left(\theta\right);\lambda_{i}\right)\right]-\log\left\{ \det\left[B\left(\beta\right)\right]\right\} -\sum_{i=1}^{n}\log\left(\sigma_{i}\right),
\]
we obtain that 
\begin{align*}
\frac{\partial}{\partial\sigma}l_{t}\left(\theta\right) & =\sum_{i=1}^{n}e_{i,x,t}\left(\theta\right)\left(-\iota_{i}\sigma_{i}^{-2}\right)\iota_{i}^{'}\underbrace{B\left(\beta\right)^{-1}u_{t}\left(\theta\right)}_{=\varepsilon_{t}\left(\theta\right)}-\underbrace{\sum_{i=1}^{n}\iota_{i}\sigma_{i}^{-1}}_{=\Sigma^{-2}\sigma}\\
 & =-\sum_{i=1}^{n}\sigma_{i}^{-2}\left(\iota_{i}\iota_{i}^{'}\right)e_{i,x,t}\left(\theta\right)\varepsilon_{t}\left(\theta\right)-\Sigma^{-2}\sigma\\
 & =-\Sigma^{-2}\left[e_{x,t}\left(\theta\right)\odot\varepsilon_{t}\left(\theta\right)+\sigma\right]
\end{align*}
where $\odot$ denotes element-wise multiplication. 

The partial derivative of $l_{t}\left(\theta\right)$ with respect
to $\sigma$ is thus identical to the one derived in \citet{LMS_svarIdent16}.

\paragraph{Result for $\frac{\partial L_{t}\left(\theta\right)}{\partial\sigma}$.}

Finally, we obtain for the partial derivative of the standardized
log-likelihood function with respect to $\beta$ that
\begin{align*}
\frac{\partial L_{t}\left(\theta\right)}{\partial\sigma} & =\frac{1}{T}\sum_{t=1}^{T}l_{\sigma,t}\left(\theta\right)\\
 & =-\frac{1}{T}\Sigma^{-2}\left(\sum_{t=1}^{T}e_{x,t}\left(\theta\right)\odot\varepsilon_{t}\left(\theta\right)\right)-\begin{pmatrix}\sigma_{1}^{-1}\\
\vdots\\
\sigma_{n}^{-1}
\end{pmatrix}
\end{align*}

\subsection{Partial Derivative with respect to $\lambda$}

Analogous to $l_{\sigma,t}\left(\theta\right)$, the partial derivative
of $l_{t}\left(\theta\right)$ with respect to $\lambda$ is identical
to the one derived in \citet{LMS_svarIdent16}, i.e. $\frac{\partial}{\partial\lambda_{i}}l_{t}\left(\theta\right)=e_{i,\lambda_{i},t}$
for all $i$.

\pagebreak{}

\section{\label{sec:score_MDS}The Score is a Martingale Difference Sequence
at the true $\theta_{0}$}

Here we show that the score is indeed a martingale difference sequence
such that we may apply the central limit theorem (CLT) for martingale
difference sequences. For the reader's convenience, we repeat here
that {\scriptsize{}
\begin{align}
\frac{\partial l_{t}\left(\theta\right)}{\partial\pi_{2}} & =\frac{\partial u_{t}\left(\theta\right)'}{\partial\pi_{2}}B'\left(\beta\right)^{-1}\Sigma^{-1}e_{x,t}\nonumber \\
 & =-x_{b,t-1}\left(\theta\right)B'\left(\beta\right)^{-1}\Sigma^{-1}e_{x,t}\left(\theta\right)\label{eq:mds_lik_pi2_poly}\\
\frac{\partial l_{t}\left(\theta\right)}{\partial\pi_{3}} & =\frac{\partial u_{t}\left(\theta\right)'}{\partial\pi_{3}}B'\left(\beta\right)^{-1}\Sigma^{-1}e_{x,t}\left(\theta\right)\nonumber \\
 & =-w_{b,t-1}\left(\theta\right)B'\left(\beta\right)^{-1}\Sigma^{-1}e_{x,t}\left(\theta\right)\label{eq:mds_lik_pi3_poly}\\
\frac{\partial l_{t}\left(\theta\right)}{\partial\beta} & =\left(\frac{\partial u_{t}'\left(\theta\right)}{\partial\beta}\right)B'\left(\beta\right)^{-1}\Sigma^{-1}e_{x,t}\left(\theta\right)-H\left(B\left(\beta\right)^{-1}u_{t}\left(\theta\right)\otimes B'\left(\beta\right)^{-1}\Sigma^{-1}e_{x,t}\left(\theta\right)\right)-H'vec\left(B'\left(\beta\right)^{-1}\right)\nonumber \\
 & =-H'\sum_{i=1}^{q}\left(B\left(\beta\right)^{-1}u_{t-i}\left(\theta\right)\otimes b_{i}'B'\left(\beta\right)^{-1}\Sigma^{-1}e_{x,t}\left(\theta\right)\right)-H'\left(B\left(\beta\right)^{-1}u_{t}\left(\theta\right)\otimes B'\left(\beta\right)^{-1}\Sigma^{-1}e_{x,t}\left(\theta\right)\right)-H'vec\left(B'\left(\beta\right)^{-1}\right)\label{eq:mds_lik_beta}\\
\frac{\partial}{\partial\sigma}l_{t}\left(\theta\right) & =-\Sigma^{-2}\left[e_{x,t}\left(\theta\right)\odot\varepsilon_{t}\left(\theta\right)+\sigma\right]\label{eq:mds_lik_sigma}\\
\frac{\partial}{\partial\lambda}l_{t}\left(\theta\right) & =e_{\lambda,t}\left(\theta\right).\nonumber 
\end{align}
}{\scriptsize\par}

\paragraph{Derivative with respect to $\pi_{2}$ and $\pi_{3}$ is an MDS.}

The first term in \eqref{eq:mds_lik_pi2_poly} and \eqref{eq:mds_lik_pi3_poly}
depends only on $\mathcal{F}_{t-1}$ from which
\[
\mathbb{E}_{t-1}\left(l_{\pi_{2},t}\left(\theta_{0}\right)\right)=-x_{b,t-1}\left(\theta\right)B'\left(\beta_{0}\right)^{-1}\Sigma_{0}^{-1}\underbrace{\mathbb{E}_{t-1}\left(e_{x,t}\left(\theta_{0}\right)\right)}_{=0}=0
\]
and 
\[
\mathbb{E}_{t-1}\left(l_{\pi_{3},t}\left(\theta_{0}\right)\right)=-w_{b,t-1}\left(\theta_{0}\right)\Sigma_{0}^{-1}\underbrace{\mathbb{E}_{t-1}\left(e_{x,t}\left(\theta_{0}\right)\right)}_{=0}=0
\]

\paragraph{Derivative with respect to $\beta$ is an MDS.}

We analyze the summands in \eqref{eq:mds_lik_beta} separately. For
the first one, i.e. 
\[
H'\sum_{i=1}^{q}\left(B\left(\beta\right)^{-1}u_{t-i}\left(\theta\right)\otimes b_{i}'B'\left(\beta\right)^{-1}\Sigma^{-1}e_{x,t}\left(\theta\right)\right),
\]
it follows in the same way as for $\pi_{2}$ and $\pi_{3}$ that 
\[
\mathbb{E}_{t-1}\left(H'\sum_{i=1}^{q}\left(B\left(\beta_{0}\right)^{-1}u_{t-i}\left(\theta_{0}\right)\otimes b_{i}'B'\left(\beta_{0}\right)^{-1}\Sigma_{0}^{-1}e_{x,t}\left(\theta\right)\right)\right)=H'\sum_{i=1}^{q}\left(B\left(\beta\right)^{-1}u_{t-i}\left(\theta\right)\otimes b_{i}'B'\left(\beta\right)^{-1}\Sigma^{-1}\right)\underbrace{\mathbb{E}_{t-1}\left(e_{x,t}\left(\theta_{0}\right)\right)}_{=0}=0.
\]

For the second term evaluated at the true parameter value $\theta_{0}$,
i.e. $H\left(B\left(\beta_{0}\right)^{-1}u_{t}\left(\theta_{0}\right)\otimes B'\left(\beta_{0}\right)^{-1}\Sigma_{0}^{-1}e_{x,t}\left(\theta_{0}\right)\right)$,
we obtain from the fact that $B\left(\beta_{0}\right)^{-1}u_{t}\left(\theta_{0}\right)=\varepsilon_{t}\left(\theta_{0}\right)$
and from Lemma B1.(v) in \citet{LMS_svarIdent16}, i.e. $\mathbb{E}\left(\varepsilon_{i,t}\left(\theta_{0}\right)e_{i,x,t}\left(\theta_{0}\right)\right)=-\sigma_{i,0}$,
that

\begin{align}
\mathbb{E}\left(H'\left(\varepsilon_{t}\left(\theta_{0}\right)\otimes B'\left(\beta_{0}\right)^{-1}\Sigma^{-1}e_{x,t}\left(\theta_{0}\right)\right)\right) & =H'\mathbb{E}\left(vec\left(B'\left(\beta_{0}\right)^{-1}\Sigma^{-1}e_{x,t}\left(\theta_{0}\right)\varepsilon_{t}'\left(\theta_{0}\right)\right)\right)\nonumber \\
 & =H'vec\left[B'\left(\beta_{0}\right)^{-1}\Sigma^{-1}\underbrace{\mathbb{E}\left(e_{x,t}\left(\theta_{0}\right)\varepsilon_{t}'\left(\theta_{0}\right)\right)}_{=-\Sigma}\right]\nonumber \\
 & =-H'vec\left(B'\left(\beta_{0}\right)^{-1}\right).\label{eq:mds_exp_beta}
\end{align}

To sum up, we obtain that $\mathbb{E}_{t-1}\left(l_{\beta,t}\left(\theta_{0}\right)\right)=0+H'vec\left(B'\left(\beta_{0}\right)^{-1}\right)-H'vec\left(B'\left(\beta_{0}\right)^{-1}\right)=0$.

\paragraph{Derivative with respect to $\sigma$ is an MDS.}

Starting from \eqref{eq:mds_lik_sigma}, we obtain from Lemma B1.(v)
in \citet{LMS_svarIdent16}, i.e. $\mathbb{E}\left(\varepsilon_{i,t}\left(\theta_{0}\right)e_{i,x,t}\left(\theta_{0}\right)\right)=-\sigma_{i,0}$,
that $\mathbb{E}_{t-1}\left(l_{\sigma,t}\left(\theta_{0}\right)\right)=\mathbb{E}\left(l_{\sigma,t}\left(\theta_{0}\right)\right)=0.$

\paragraph{Derivative with respect to $\lambda$ is an MDS.}

Identitically to \citet{LMS_svarIdent16}, we obtain from Lemma B1.(iii)
in \citet{LMS_svarIdent16}, i.e. $\mathbb{E}\left(e_{i,\lambda_{i},t}\left(\theta_{0}\right)\right)=0$,
that $\mathbb{E}_{t-1}\left(l_{\lambda_{i},t}\left(\theta_{0}\right)\right)=\mathbb{E}\left(e_{i,\lambda_{i},t}\left(\theta_{0}\right)\right)=0.$

\pagebreak{}

\section{\label{sec:opg_expression}Covariance of the Score evaluated at $\theta_{0}$}

Here we will calculate the covariance matrix of the score, evaluated
at the true parameter value $\theta_{0}$, i.e.
\[
\mathbb{E}\left(\begin{pmatrix}l_{\pi_{2},t}\left(\theta_{0}\right)\\
l_{\pi_{3},t}\left(\theta_{0}\right)\\
l_{\beta,t}\left(\theta_{0}\right)\\
l_{\sigma,t}\left(\theta_{0}\right)\\
l_{\lambda,t}\left(\theta_{0}\right)
\end{pmatrix}\begin{pmatrix}l_{\pi_{2},t}'\left(\theta_{0}\right) & l_{\pi_{3},t}'\left(\theta_{0}\right) & l_{\beta,t}'\left(\theta_{0}\right) & l_{\sigma,t}'\left(\theta_{0}\right) & l_{\lambda,t}'\left(\theta_{0}\right)\end{pmatrix}\right).
\]
We will start with the first block of rows in this matrix and derive
all terms involving $\pi_{2}.$ Subsequently, we will do the same
for all remaining terms involving $\pi_{3}$ and the other variables,
taking the symmetry of the covariance matrix into account.

To repeat, we will work here with 
\begin{align}
\frac{\partial l_{t}\left(\theta_{0}\right)}{\partial\pi_{2}} & =-x_{b,t-1}\left(\theta_{0}\right)B'\left(\beta_{0}\right)^{-1}\Sigma_{0}^{-1}e_{x,t}\left(\theta_{0}\right)\nonumber \\
\frac{\partial l_{t}\left(\theta_{0}\right)}{\partial\pi_{3}} & =-w_{b,t-1}\left(\theta_{0}\right)\Sigma_{0}^{-1}e_{x,t}\left(\theta_{0}\right)\nonumber \\
\frac{\partial l_{t}\left(\theta_{0}\right)}{\partial\beta} & =H'\sum_{i=1}^{q}\left(I_{n}\otimes b_{i}'B'\left(\beta_{0}\right)^{-1}\Sigma_{0}^{-1}\right)\left(\varepsilon_{t-i}\left(\theta_{0}\right)\otimes e_{x,t}\left(\theta_{0}\right)\right)\nonumber \\
 & \quad-H'\left(I_{n}\otimes B'\left(\beta_{0}\right)^{-1}\Sigma_{0}^{-1}\right)\left(\varepsilon_{t}\left(\theta_{0}\right)\otimes e_{x,t}\left(\theta_{0}\right)\right)-H'vec\left(B'\left(\beta_{0}\right)^{-1}\right)\label{eq:cov_lik_beta}\\
\frac{\partial}{\partial\sigma}l_{t}\left(\theta_{0}\right) & =-\Sigma_{0}^{-2}\left[e_{x,t}\left(\theta_{0}\right)\odot\varepsilon_{t}\left(\theta_{0}\right)+\sigma_{0}\right]\nonumber \\
\frac{\partial}{\partial\lambda}l_{t}\left(\theta_{0}\right) & =e_{\lambda,t}\left(\theta_{0}\right).\nonumber 
\end{align}

\subparagraph{Summary of results.}

In the following subsections, we will derive
\begin{align*}
\mathbb{E}\left(l_{\pi_{2},t}\left(\theta_{0}\right)l_{\pi_{2},t}'\left(\theta_{0}\right)\right) & =\mathbb{E}\left(x_{b,t-1}\left(\theta_{0}\right)\right)B'\left(\beta_{0}\right)^{-1}\Sigma_{0}^{-1}\mathbb{E}\left(e_{x,t}\left(\theta_{0}\right)e_{x,t}'\left(\theta_{0}\right)\right)\Sigma_{0}^{-1}B\left(\beta_{0}\right)^{-1}\mathbb{E}\left(x_{b,t-1}'\left(\theta_{0}\right)\right)\\
\mathbb{E}\left(l_{\pi_{2},t}\left(\theta_{0}\right)l_{\pi_{3},t}'\left(\theta_{0}\right)\right) & =\mathbb{E}\left(x_{b,t-1}\left(\theta_{0}\right)\right)B'\left(\beta_{0}\right)^{-1}\Sigma_{0}^{-1}V_{e_{x}}\Sigma_{0}^{-1}B\left(\beta_{0}\right)^{-1}w_{b,t-1}'\left(\theta_{0}\right)\\
\mathbb{E}\left(l_{\pi_{2},t}\left(\theta_{0}\right)l_{\beta,t}'\left(\theta_{0}\right)\right) & =\mathbb{E}\left(x_{b,t-1}\left(\theta_{0}\right)\right)\mathfrak{b}_{t}B'\left(\beta_{0}\right)^{-1}\Sigma_{0}^{-1}\mathbb{E}\left[\varepsilon_{t}\left(\theta_{0}\right)\otimes\left(e_{x,t}\left(\theta_{0}\right)e_{x,t}'\left(\theta_{0}\right)\right)\right]\left(I_{n}\otimes\Sigma_{0}^{-1}B\left(\beta_{0}\right)^{-1}\right)H\\
\mathbb{E}\left(l_{\pi_{2},t}\left(\theta_{0}\right)\left(l_{\sigma,t}\left(\theta_{0}\right)\right)^{'}\right) & =-\mathbb{E}\left(x_{b,t-1}\left(\theta_{0}\right)\right)B'\left(\beta\right)^{-1}\Sigma^{-1}\mathbb{E}\left[e_{x,t}\left(\theta_{0}\right)\left(e_{x,t}'\left(\theta_{0}\right)\odot\varepsilon_{t}'\left(\theta_{0}\right)\right)\right]\Sigma_{0}^{-2}\\
\mathbb{E}\left(l_{\pi_{2},t}\left(\theta_{0}\right)l_{\lambda,t}'\left(\theta_{0}\right)\right) & =-x_{b,t-1}\left(\theta_{0}\right)B'\left(\beta\right)^{-1}\Sigma^{-1}\mathbb{E}\left(e_{x,t}\left(\theta_{0}\right)e_{\lambda,t}'\left(\theta_{0}\right)\right)\\
\mathbb{E}\left(l_{\pi_{3},t}\left(\theta_{0}\right)l_{\pi_{3},t}'\left(\theta_{0}\right)\right) & =\mathbb{E}\left(w_{b,t-1}\left(\theta_{0}\right)\right)B'\left(\beta_{0}\right)^{-1}\Sigma_{0}^{-1}V_{e_{x}}\Sigma_{0}^{-1}B\left(\beta_{0}\right)^{-1}\mathbb{E}\left(w_{b,t-1}'\left(\theta_{0}\right)\right)\\
\mathbb{E}\left(l_{\pi_{2},t}\left(\theta_{0}\right)l_{\beta,t}'\left(\theta_{0}\right)\right) & =\mathbb{E}\left(w_{b,t-1}\left(\theta_{0}\right)\right)B'\left(\beta_{0}\right)^{-1}\Sigma_{0}^{-1}\mathbb{E}\left[\varepsilon_{t}\left(\theta_{0}\right)\otimes\left(e_{x,t}\left(\theta_{0}\right)e_{x,t}'\left(\theta_{0}\right)\right)\right]\left(I_{n}\otimes\Sigma_{0}^{-1}B\left(\beta_{0}\right)^{-1}\right)H.\\
\mathbb{E}\left(l_{\pi_{3},t}\left(\theta_{0}\right)l_{\sigma,t}'\left(\theta_{0}\right)\right) & =\mathbb{E}\left(w_{b,t-1}\left(\theta_{0}\right)\right)B'\left(\beta\right)^{-1}\Sigma^{-1}\left\{ \mathbb{E}\left[e_{x,t}\left(\theta_{0}\right)\left(e_{x,t}'\left(\theta_{0}\right)\odot\varepsilon_{t}'\left(\theta_{0}\right)\right)\right]+\underbrace{\mathbb{E}\left[e_{x,t}\left(\theta_{0}\right)\sigma_{0}'\right]}_{=0}\right\} \Sigma_{0}^{-2}.\\
\mathbb{E}\left(l_{\pi_{3},t}\left(\theta_{0}\right)l_{\lambda,t}'\left(\theta_{0}\right)\right) & =-\text{\ensuremath{\mathbb{E}}}\left(w_{b,t-1}\left(\theta_{0}\right)\right)B'\left(\beta\right)^{-1}\Sigma^{-1}\mathbb{E}\left(e_{x,t}\left(\theta_{0}\right)e_{\lambda,t}'\left(\theta_{0}\right)\right)\\
\mathbb{E}\left(l_{\beta,t}\left(\theta_{0}\right)l_{\beta,t}'\left(\theta_{0}\right)\right) & =H'\left[\sum_{i=1}^{q}\left(\Sigma\otimes b_{i}'B'\left(\beta_{0}\right)^{-1}\Sigma^{-1}V_{e_{x}}\Sigma^{-1}B\left(\beta_{0}\right)^{-1}b_{i}\right)\right]H\\
 & \quad+H'\left(I_{n}\otimes B'(\beta_{0})^{-1}\Sigma_{0}^{-1}\right)\mathbb{E}\left(\varepsilon_{t}\left(\theta_{0}\right)\varepsilon_{t}'\left(\theta_{0}\right)\otimes e_{x,t}\left(\theta_{0}\right)e_{x,t}'\left(\theta_{0}\right)\right)\left(I_{n}\otimes\Sigma_{0}'^{-1}B\left(\beta_{0}\right)^{-1}\right)H\\
 & \quad-H'vec\left(B'\left(\beta_{0}\right)^{-1}\right)vec\left(B'\left(\beta_{0}\right)^{-1}\right)'H.\\
\mathbb{E}\left(l_{\beta,t}\left(\theta_{0}\right)l_{\sigma,t}'\left(\theta_{0}\right)\right) & =H'\left(I_{n}\otimes B'\left(\beta_{0}\right)^{-1}\Sigma_{0}^{-1}\right)\mathbb{E}\left\{ \left(\varepsilon_{t}\left(\theta\right)\otimes e_{x,t}\left(\theta\right)\left[e_{x,t}'\left(\theta_{0}\right)\odot\varepsilon_{t}'\left(\theta_{0}\right)\right]\right)\right\} \Sigma_{0}^{-2}+H'vec\left(B'\left(\beta\right)^{-1}\right)\sigma_{0}'\Sigma_{0}^{-2}\\
\mathbb{E}\left(l_{\beta,t}\left(\theta_{0}\right)l_{\lambda,t}'\left(\theta_{0}\right)\right) & =-H'\left(I_{n}\otimes B'(\beta)^{-1}\Sigma^{-1}\right)\mathbb{E}\left(\varepsilon_{t}\left(\theta_{0}\right)\otimes e_{x,t}\left(\theta_{0}\right)e_{i,\lambda,t}'\right)\\
\mathbb{E}\left(l_{\sigma,t}\left(\theta_{0}\right)l_{\sigma,t}'\left(\theta_{0}\right)\right) & =\Sigma_{0}^{-2}\mathbb{E}\left[\left(e_{x,t}\left(\theta\right)\odot\varepsilon_{t}\left(\theta\right)+\sigma\right)\left(e_{x,t}\left(\theta\right)\odot\varepsilon_{t}\left(\theta\right)+\sigma\right)^{'}\right]\Sigma_{0}^{-2}\\
\mathbb{E}\left(l_{\sigma,t}\left(\theta_{0}\right)l_{\lambda,t}'\left(\theta_{0}\right)\right) & =-\Sigma_{0}^{-2}\mathbb{E}\left(\left(e_{x,t}\left(\theta_{0}\right)\odot\varepsilon_{t}\left(\theta_{0}\right)\right)e_{\lambda,t}'\left(\theta_{0}\right)\right).
\end{align*}

\pagebreak{}

\subsection{Elements of the Covariance Matrix involving $\pi_{2}$}

\paragraph{Diagonal term $\left(\pi_{2},\pi_{2}\right)$.}

We obtain
\begin{align*}
\mathbb{E}\left(l_{\pi_{2},t}\left(\theta_{0}\right)l_{\pi_{2},t}'\left(\theta_{0}\right)\right) & =\mathbb{E}\left\{ x_{b,t-1}\left(\theta_{0}\right)B'\left(\beta_{0}\right)^{-1}\Sigma_{0}^{-1}e_{x,t}\left(\theta_{0}\right)e_{x,t}'\left(\theta_{0}\right)\Sigma_{0}^{-1}B\left(\beta_{0}\right)^{-1}x_{b,t-1}'\left(\theta_{0}\right)\right\} \\
 & =\mathbb{E}\left(x_{b,t-1}\left(\theta_{0}\right)\right)B'\left(\beta_{0}\right)^{-1}\Sigma_{0}^{-1}V_{e_{x}}\Sigma_{0}^{-1}B\left(\beta_{0}\right)^{-1}\mathbb{E}\left(x_{b,t-1}'\left(\theta_{0}\right)\right)
\end{align*}
where $V_{e_{x}}=\mathbb{E}\left(e_{x,t}\left(\theta_{0}\right)e_{x,t}'\left(\theta_{0}\right)\right)$.

\paragraph{Term $\left(\pi_{2},\pi_{3}\right)$.}

Similarly, we obtain 
\begin{align*}
\mathbb{E}\left(l_{\pi_{2},t}\left(\theta_{0}\right)l_{\pi_{3},t}'\left(\theta_{0}\right)\right) & =\mathbb{E}\left\{ x_{b,t-1}\left(\theta_{0}\right)B'\left(\beta_{0}\right)^{-1}\Sigma_{0}^{-1}e_{x,t}\left(\theta_{0}\right)e_{x,t}'\left(\theta_{0}\right)\Sigma_{0}^{-1}B\left(\beta_{0}\right)^{-1}w_{b,t-1}'\left(\theta_{0}\right)\right\} \\
 & =\mathbb{E}\left(x_{b,t-1}\left(\theta_{0}\right)\right)B'\left(\beta_{0}\right)^{-1}\Sigma_{0}^{-1}V_{e_{x}}\Sigma_{0}^{-1}B\left(\beta_{0}\right)^{-1}w_{b,t-1}'\left(\theta_{0}\right).
\end{align*}

\paragraph{Term $\left(\pi_{2},\beta\right)$.}

We consider the three summands of \eqref{eq:cov_lik_beta} separately.
The expectation of $l_{\pi_{2},t}\left(\theta_{0}\right)$ with the
last summand is zero as a consequence of the expectation of $l_{\pi_{2},t}\left(\theta_{0}\right)$
being zero. For the first summand, i.e. 
\[
-\left(\sum_{i=1}^{q}\left(\varepsilon_{t-i}'\left(\theta_{0}\right)\otimes e_{x,t}'\left(\theta_{0}\right)\right)\left(I_{n}\otimes\Sigma_{0}^{-1}B\left(\beta_{0}\right)^{-1}b_{i}\right)\right)H
\]
, we have 
\begin{align*}
\mathbb{E}\left(l_{\pi_{2},t}\left(\theta_{0}\right)l_{\beta,t,\text{first}}'\left(\theta_{0}\right)\right) & =-\mathbb{E}\left\{ x_{b,t-1}\left(\theta_{0}\right)B'\left(\beta_{0}\right)^{-1}\Sigma_{0}^{-1}e_{x,t}\left(\theta_{0}\right)\left(\sum_{i=1}^{q}\left(\varepsilon_{t-i}'\left(\theta_{0}\right)\otimes e_{x,t}'\left(\theta_{0}\right)\right)\left(I_{n}\otimes\Sigma_{0}^{-1}B\left(\beta_{0}\right)^{-1}b_{i}\right)\right)H\right\} \\
 & =-\mathbb{E}\left(x_{b,t-1}\left(\theta_{0}\right)\right)B'\left(\beta_{0}\right)^{-1}\Sigma_{0}^{-1}\left(\sum_{i=1}^{q}\underbrace{\mathbb{E}\left[\varepsilon_{t-i}'\left(\theta_{0}\right)\otimes\left(e_{x,t}\left(\theta_{0}\right)e_{x,t}'\left(\theta_{0}\right)\right)\right]}_{=0}\left(I_{n}\otimes\Sigma_{0}^{-1}B\left(\beta_{0}\right)^{-1}b_{i}\right)\right)H=0
\end{align*}
due to independence of (functions of) $\varepsilon_{t}$ of $\varepsilon_{t-i},\ i>0$
and $\mathbb{E}\left(\varepsilon_{t}\left(\theta_{0}\right)\right)=0$.
What remains is thus the covariance with the second summand, i.e.
\begin{align*}
\mathbb{E}\left(l_{\pi_{2},t}\left(\theta_{0}\right)l_{\beta,t}'\left(\theta_{0}\right)\right) & =\mathbb{E}\left\{ x_{b,t-1}\left(\theta_{0}\right)B'\left(\beta_{0}\right)^{-1}\Sigma_{0}^{-1}e_{x,t}\left(\theta_{0}\right)\left(\varepsilon_{t}\left(\theta_{0}\right)\otimes e_{x,t}'\left(\theta_{0}\right)\right)\left(I_{n}\otimes\Sigma_{0}^{-1}B\left(\beta_{0}\right)^{-1}\right)H\right\} \\
 & =\mathbb{E}\left(x_{b,t-1}\left(\theta_{0}\right)\right)\mathfrak{b}_{t}B'\left(\beta_{0}\right)^{-1}\Sigma_{0}^{-1}\mathbb{E}\left[\varepsilon_{t}\left(\theta_{0}\right)\otimes\left(e_{x,t}\left(\theta_{0}\right)e_{x,t}'\left(\theta_{0}\right)\right)\right]\left(I_{n}\otimes\Sigma_{0}^{-1}B\left(\beta_{0}\right)^{-1}\right)H.
\end{align*}

\paragraph{Term $\left(\pi_{2},\sigma\right)$.}

The independence of $x_{b,t-1}\left(\theta_{0}\right)$ from $\varepsilon_{t}\left(\theta_{0}\right)$
and $e_{x,t}\left(\theta_{0}\right)$ entails that
\begin{align*}
\mathbb{E}\left(l_{\pi_{2},t}\left(\theta_{0}\right)\left(l_{\sigma,t}\left(\theta_{0}\right)\right)^{'}\right) & =-\mathbb{E}\left(\left\{ x_{b,t-1}\left(\theta_{0}\right)B'\left(\beta_{0}\right)^{-1}\Sigma_{0}^{-1}e_{x,t}\left(\theta_{0}\right)\right\} \left[e_{x,t}'\left(\theta_{0}\right)\odot\varepsilon_{t}'\left(\theta_{0}\right)+\sigma_{0}'\right]\right)\Sigma_{0}^{-2}\\
 & =-\mathbb{E}\left(x_{b,t-1}\left(\theta_{0}\right)\right)B'\left(\beta\right)^{-1}\Sigma^{-1}\left\{ \mathbb{E}\left[e_{x,t}\left(\theta_{0}\right)\left(e_{x,t}'\left(\theta_{0}\right)\odot\varepsilon_{t}'\left(\theta_{0}\right)\right)\right]+\underbrace{\mathbb{E}\left[e_{x,t}\left(\theta_{0}\right)\sigma_{0}'\right]}_{=0}\right\} \Sigma_{0}^{-2}.
\end{align*}

\paragraph{Term $\left(\pi_{2},\lambda\right)$.}

We obtain 
\[
\mathbb{E}\left(l_{\pi_{2},t}\left(\theta_{0}\right)l_{\lambda,t}'\left(\theta_{0}\right)\right)=-x_{b,t-1}\left(\theta_{0}\right)B'\left(\beta\right)^{-1}\Sigma^{-1}\mathbb{E}\left(e_{x,t}\left(\theta_{0}\right)e_{\lambda,t}'\left(\theta_{0}\right)\right).
\]

\subsection{Elements of the Covariance Matrix involving $\pi_{3}$}

\paragraph{Diagonal term $\left(\pi_{3},\pi_{3}\right)$.}

Similar to the diagonal term $\left(\pi_{2},\pi_{2}\right)$, we obtain
\[
\mathbb{E}\left(l_{\pi_{3},t}\left(\theta_{0}\right)l_{\pi_{3},t}'\left(\theta_{0}\right)\right)=\mathbb{E}\left(w_{b,t-1}\left(\theta_{0}\right)\right)B'\left(\beta_{0}\right)^{-1}\Sigma_{0}^{-1}V_{e_{x}}\Sigma_{0}^{-1}B\left(\beta_{0}\right)^{-1}\mathbb{E}\left(w_{b,t-1}'\left(\theta_{0}\right)\right).
\]
Note that $u_{t}\left(\theta_{0}\right)$ may be expressed in terms
of the observations $y_{t}$ and that, due to the structure of $\mathfrak{b}_{t}$,
all elements of $w_{b,t-1}\left(\theta_{0}\right)$ are contained
in $\mathcal{F}_{t-1}$.

\paragraph{Term $\left(\pi_{3},\beta\right)$.}

Similar to the term $\left(\pi_{2},\beta\right)$, we consider the
three summands of \eqref{eq:cov_lik_beta} separately. The expectation
of $l_{\pi_{3},t}\left(\theta_{0}\right)$ with the last summand is
zero as a consequence of the expectation of $l_{\pi_{3},t}\left(\theta_{0}\right)$
being zero. For the first summand, i.e. $-\left(\sum_{i=1}^{q}\left(\varepsilon_{t-i}'\left(\theta_{0}\right)\otimes e_{x,t}'\left(\theta_{0}\right)\right)\left(I_{n}\otimes\Sigma_{0}^{-1}B\left(\beta_{0}\right)^{-1}b_{i}\right)\right)H$,
we have 
\[
\mathbb{E}\left[l_{\pi_{3},t}\left(\theta_{0}\right)\left(l_{\beta,t,\text{first}}'\left(\theta_{0}\right)\right)\right]=\mathbb{E}\left(w_{b,t-1}\left(\theta_{0}\right)\right)B'\left(\beta_{0}\right)^{-1}\Sigma_{0}^{-1}\left(\sum_{i=1}^{q}\underbrace{\mathbb{E}\left[\varepsilon_{t-i}'\left(\theta_{0}\right)\otimes\left(e_{x,t}\left(\theta_{0}\right)e_{x,t}'\left(\theta_{0}\right)\right)\right]}_{=0}\left(I_{n}\otimes\Sigma_{0}^{-1}B\left(\beta_{0}\right)^{-1}b_{i}\right)\right)H=0
\]
due to independence of (functions of) $\varepsilon_{t}$ of $\varepsilon_{t-i},\ i>0$
and $\mathbb{E}\left(\varepsilon_{t}\left(\theta_{0}\right)\right)=0$.
What remains is thus the covariance of $l_{\pi_{3},t}\left(\theta_{0}\right)$
with the second summand of $l_{\beta,t}'\left(\theta_{0}\right)$,
i.e.
\begin{align*}
\mathbb{E}\left(l_{\pi_{2},t}\left(\theta_{0}\right)l_{\beta,t}'\left(\theta_{0}\right)\right) & =\mathbb{E}\left\{ w_{b,t-1}\left(\theta_{0}\right)B'\left(\beta_{0}\right)^{-1}\Sigma_{0}^{-1}e_{x,t}\left(\theta_{0}\right)\left(\varepsilon_{t}\left(\theta_{0}\right)\otimes e_{x,t}'\left(\theta_{0}\right)\right)\left(I_{n}\otimes\Sigma_{0}^{-1}B\left(\beta_{0}\right)^{-1}\right)H\right\} \\
 & =\mathbb{E}\left(w_{b,t-1}\left(\theta_{0}\right)\right)B'\left(\beta_{0}\right)^{-1}\Sigma_{0}^{-1}\mathbb{E}\left[\varepsilon_{t}\left(\theta_{0}\right)\otimes\left(e_{x,t}\left(\theta_{0}\right)e_{x,t}'\left(\theta_{0}\right)\right)\right]\left(I_{n}\otimes\Sigma_{0}^{-1}B\left(\beta_{0}\right)^{-1}\right)H.
\end{align*}

\paragraph{Term $\left(\pi_{3},\sigma\right)$.}

Similar to the term $\left(\pi_{2},\sigma\right)$, we obtain
\begin{align*}
\mathbb{E}\left(l_{\pi_{3},t}\left(\theta_{0}\right)l_{\sigma,t}'\left(\theta_{0}\right)\right) & =\mathbb{E}\left(\left\{ w_{b,t-1}\left(\theta_{0}\right)B'\left(\beta_{0}\right)^{-1}\Sigma_{0}^{-1}e_{x,t}\left(\theta_{0}\right)\right\} \left[e_{x,t}'\left(\theta_{0}\right)\odot\varepsilon_{t}'\left(\theta_{0}\right)+\sigma_{0}'\right]\right)\Sigma_{0}^{-2}\\
 & =\mathbb{E}\left(w_{b,t-1}\left(\theta_{0}\right)\right)B'\left(\beta\right)^{-1}\Sigma^{-1}\left\{ \mathbb{E}\left[e_{x,t}\left(\theta_{0}\right)\left(e_{x,t}'\left(\theta_{0}\right)\odot\varepsilon_{t}'\left(\theta_{0}\right)\right)\right]+\underbrace{\mathbb{E}\left[e_{x,t}\left(\theta_{0}\right)\sigma_{0}'\right]}_{=0}\right\} \Sigma_{0}^{-2}.
\end{align*}

\paragraph{Term $\left(\pi_{3},\lambda\right)$.}

Similar to the term $\left(\pi_{2},\sigma\right)$, we obtain
\[
\mathbb{E}\left(l_{\pi_{3},t}\left(\theta_{0}\right)l_{\lambda,t}'\left(\theta_{0}\right)\right)=-\text{\ensuremath{\mathbb{E}}}\left(w_{b,t-1}\left(\theta_{0}\right)\right)B'\left(\beta\right)^{-1}\Sigma^{-1}\mathbb{E}\left(e_{x,t}\left(\theta_{0}\right)e_{\lambda,t}'\left(\theta_{0}\right)\right).
\]

\subsection{Elements of the Covariance Matrix involving $\beta$}

\paragraph{Diagonal term $\left(\beta,\beta\right)$.}

We will analyze sequentially the expectations of the cross- and square-terms
in
\begin{align*}
\frac{\partial l_{t}\left(\theta_{0}\right)}{\partial\beta} & =-\underbrace{H'\sum_{i=1}^{q}\left(I_{n}\otimes b_{i}'B'\left(\beta_{0}\right)^{-1}\Sigma_{0}^{-1}\right)\left(\varepsilon_{t-i}\left(\theta_{0}\right)\otimes e_{x,t}\left(\theta_{0}\right)\right)}_{=(I)}\\
 & \quad-\underbrace{H'\left(I_{n}\otimes B'\left(\beta_{0}\right)^{-1}\Sigma_{0}^{-1}\right)\left(\varepsilon_{t}\left(\theta_{0}\right)\otimes e_{x,t}\left(\theta_{0}\right)\right)}_{=(II)}-\underbrace{H'vec\left(B'\left(\beta_{0}\right)^{-1}\right)}_{=(III)}.
\end{align*}
First, note that $\mathbb{E}\left[(I)(II)\right]=\mathbb{E}\left[(I)(III)\right]=0$
because 
\begin{align*}
\mathbb{E}\left[\left(\varepsilon_{t-i}\left(\theta_{0}\right)\otimes e_{x,t}\left(\theta_{0}\right)\right)\left(\varepsilon_{t}\left(\theta_{0}\right)\otimes e_{x,t}\left(\theta_{0}\right)\right)'\right] & =\mathbb{E}\left(\varepsilon_{t-i}\left(\theta_{0}\right)\varepsilon_{t}'\left(\theta_{0}\right)\otimes e_{x,t}\left(\theta_{0}\right)e_{x,t}'\left(\theta_{0}\right)\right)=0
\end{align*}
and $\mathbb{E}\left(\varepsilon_{t-i}\left(\theta_{0}\right)\otimes e_{x,t}\left(\theta_{0}\right)\right)=0$
due to the fact that $e_{x,t}\left(\theta_{0}\right)$ is a function
of $\varepsilon_{t}\left(\theta_{0}\right)$ and $\varepsilon_{t}\left(\theta_{0}\right)$
is independent of $\left(\varepsilon_{t-1}'\left(\theta_{0}\right),\ldots,\varepsilon_{t-q}'\left(\theta_{0}\right)\right)$.
For the last cross-term, we obtain in the same way as in \eqref{eq:mds_exp_beta}
that 
\begin{align*}
\mathbb{E}\left(\left(II\right)\left(III\right)\right) & =\mathbb{E}\left[H'\left(I_{n}\otimes B'\left(\beta_{0}\right)^{-1}\Sigma_{0}^{-1}\right)\left(\varepsilon_{t}\left(\theta_{0}\right)\otimes e_{x,t}\left(\theta_{0}\right)\right)vec\left(B'^{-1}\right)'H\right]\\
 & =H'vec\left[B'\left(\beta_{0}\right)^{-1}\Sigma^{-1}\underbrace{\mathbb{E}\left(e_{x,t}\left(\theta_{0}\right)\varepsilon_{t}'\left(\theta_{0}\right)\right)}_{=-\Sigma}\right]vec\left(B'^{-1}\right)'H\\
 & =H'\left(-vec\left(B'\left(\beta_{0}\right)^{-1}\right)\right)vec\left(B'\left(\beta_{0}\right)^{-1}\right)'H\\
 & =-H'vec\left(B'\left(\beta_{0}\right)^{-1}\right)vec\left(B'\left(\beta_{0}\right)^{-1}\right)'H.
\end{align*}

Regarding the first square term, consider
\begin{align*}
\mathbb{E}\left(\left(I\right)^{2}\right) & =H'\left[\left(I_{n}\otimes b_{1}'B'\left(\beta_{0}\right)^{-1}\Sigma^{-1}\right),\ldots,\left(I_{n}\otimes b_{q}'B'\left(\beta_{0}\right)^{-1}\Sigma^{-1}\right)\right]\cdot\\
 & \quad\cdot\mathbb{E}\left[\begin{pmatrix}\left(\varepsilon_{t-1}\left(\theta_{0}\right)\otimes e_{x,t}\left(\theta_{0}\right)\right)\\
\vdots\\
\left(\varepsilon_{t-q}\left(\theta_{0}\right)\otimes e_{x,t}\left(\theta_{0}\right)\right)
\end{pmatrix}\left(\varepsilon_{t-1}'\left(\theta_{0}\right)\otimes e_{x,t}'\left(\theta_{0}\right)\right),\ldots,\left(\varepsilon_{t-q}'\left(\theta_{0}\right)\otimes e_{x,t}'\left(\theta_{0}\right)\right)\right]\left[\begin{pmatrix}\left(I_{n}\otimes\Sigma^{-1}B\left(\beta_{0}\right)^{-1}b_{1}\right)\\
\vdots\\
\left(I_{n}\otimes\Sigma^{-1}B\left(\beta_{0}\right)^{-1}b_{q}\right)
\end{pmatrix}\right]H
\end{align*}
and note that the term in the middle is equal to 
\begin{gather*}
\mathbb{E}\left\{ \left[\begin{pmatrix}\left(\varepsilon_{t-1}\left(\theta_{0}\right)\otimes e_{x,t}\left(\theta_{0}\right)\right)\\
\vdots\\
\left(\varepsilon_{t-q}\left(\theta_{0}\right)\otimes e_{x,t}\left(\theta_{0}\right)\right)
\end{pmatrix}\right]\left[\left(\varepsilon_{t-1}'\left(\theta_{0}\right)\otimes e_{x,t}'\left(\theta_{0}\right)\right),\ldots,\left(\varepsilon_{t-q}'\left(\theta_{0}\right)\otimes e_{x,t}'\left(\theta_{0}\right)\right)\right]\right\} =\\
=\mathbb{E}\left\{ \left[\begin{pmatrix}\varepsilon_{t-1}\left(\theta_{0}\right)\\
\vdots\\
\varepsilon_{t-q}\left(\theta_{0}\right)
\end{pmatrix}\otimes e_{x,t}\left(\theta_{0}\right)\right]\left[\left(\varepsilon_{t-1}'\left(\theta_{0}\right),\ldots,\varepsilon_{t-q}'\left(\theta_{0}\right)\right)\otimes e_{x,t}'\left(\theta_{0}\right)\right]\right\} \\
=\mathbb{E}\left\{ \left[\begin{pmatrix}\varepsilon_{t-1}\left(\theta_{0}\right)\\
\vdots\\
\varepsilon_{t-q}\left(\theta_{0}\right)
\end{pmatrix}\left(\varepsilon_{t-1}'\left(\theta_{0}\right),\ldots,\varepsilon_{t-q}'\left(\theta_{0}\right)\right)\right]\otimes\left[e_{x,t}\left(\theta_{0}\right)e_{x,t}'\left(\theta_{0}\right)\right]\right\} \\
=\mathbb{E}\left[\begin{pmatrix}\varepsilon_{t-1}\left(\theta_{0}\right)\\
\vdots\\
\varepsilon_{t-q}\left(\theta_{0}\right)
\end{pmatrix}\left(\varepsilon_{t-1}'\left(\theta_{0}\right),\ldots,\varepsilon_{t-q}'\left(\theta_{0}\right)\right)\right]\otimes\mathbb{E}\left[e_{x,t}\left(\theta_{0}\right)e_{x,t}'\left(\theta_{0}\right)\right]=\left(I_{q}\otimes\Sigma\right)\otimes V_{e_{x}}.
\end{gather*}
Thus, we obtain 
\begin{align*}
\mathbb{E}\left(\left(I\right)^{2}\right) & =H'\left[\left(I_{n}\otimes b_{1}'B'\left(\beta_{0}\right)^{-1}\Sigma^{-1}\right),\ldots,\left(I_{n}\otimes b_{q}'B'\left(\beta_{0}\right)^{-1}\Sigma^{-1}\right)\right]\left[I_{q}\otimes\left(\Sigma\otimes V_{e_{x}}\right)\right]\left[\begin{pmatrix}\left(I_{n}\otimes\Sigma^{-1}B\left(\beta_{0}\right)^{-1}b_{1}\right)\\
\vdots\\
\left(I_{n}\otimes\Sigma^{-1}B\left(\beta_{0}\right)^{-1}b_{q}\right)
\end{pmatrix}\right]H\\
 & =H'\left[\sum_{i=1}^{q}\left(\Sigma\otimes b_{i}'B'\left(\beta_{0}\right)^{-1}\Sigma^{-1}V_{e_{x}}\Sigma^{-1}B\left(\beta_{0}\right)^{-1}b_{i}\right)\right]H.
\end{align*}
For the second square term, we have
\begin{align*}
\mathbb{E}\left(\left(II\right)^{2}\right) & =H'\left(I_{n}\otimes B'(\beta_{0})^{-1}\Sigma_{0}^{-1}\right)\mathbb{E}\left(\varepsilon_{t}\left(\theta_{0}\right)\varepsilon_{t}'\left(\theta_{0}\right)\otimes e_{x,t}\left(\theta_{0}\right)e_{x,t}'\left(\theta_{0}\right)\right)\left(I_{n}\otimes\Sigma_{0}'^{-1}B\left(\beta_{0}\right)^{-1}\right)H
\end{align*}
where the expectation $\mathbb{E}\left(\varepsilon_{t}\left(\theta_{0}\right)\varepsilon_{t}'\left(\theta_{0}\right)\otimes e_{x,t}\left(\theta_{0}\right)e_{x,t}'\left(\theta_{0}\right)\right)$
can be obtained from Lemma C2.(v) in \citet{LMS_svarIdent16}. The
third term is $\mathbb{E}\left(\left(III\right)^{2}\right)=H'vec\left(B'\left(\beta\right)^{-1}\right)vec\left(B'\left(\beta\right)^{-1}\right)'H.$

To sum up, we have that
\begin{align*}
\mathbb{E}\left(l_{\beta,t}\left(\theta_{0}\right)l_{\beta,t}'\left(\theta_{0}\right)\right) & =H'\left[\sum_{i=1}^{q}\left(\Sigma\otimes b_{i}'B'\left(\beta_{0}\right)^{-1}\Sigma^{-1}V_{e_{x}}\Sigma^{-1}B\left(\beta_{0}\right)^{-1}b_{i}\right)\right]H\\
 & \quad+H'\left(I_{n}\otimes B'(\beta_{0})^{-1}\Sigma_{0}^{-1}\right)\mathbb{E}\left(\varepsilon_{t}\left(\theta_{0}\right)\varepsilon_{t}'\left(\theta_{0}\right)\otimes e_{x,t}\left(\theta_{0}\right)e_{x,t}'\left(\theta_{0}\right)\right)\left(I_{n}\otimes\Sigma_{0}'^{-1}B\left(\beta_{0}\right)^{-1}\right)H\\
 & \quad-H'vec\left(B'\left(\beta_{0}\right)^{-1}\right)vec\left(B'\left(\beta_{0}\right)^{-1}\right)'H.
\end{align*}

\paragraph{Term $\left(\beta,\sigma\right)$.}

We consider the terms in {\scriptsize{}
\begin{gather*}
\mathbb{E}\left(l_{\beta,t}\left(\theta_{0}\right)l_{\sigma,t}'\left(\theta_{0}\right)\right)=\\
=\mathbb{E}\left\{ \left[\underbrace{H'\sum_{i=1}^{q}\left(\varepsilon_{t-i}\left(\theta_{0}\right)\otimes b_{i}'B'\left(\beta_{0}\right)^{-1}\Sigma_{0}^{-1}e_{x,t}\left(\theta_{0}\right)\right)}_{=(I)}+\underbrace{H'\left(\varepsilon_{t}\left(\theta_{0}\right)\otimes B'\left(\beta_{0}\right)^{-1}\Sigma_{0}^{-1}e_{x,t}\left(\theta_{0}\right)\right)}_{=(II)}+\underbrace{H'vec\left(B'\left(\beta_{0}\right)^{-1}\right)}_{=(III)}\right]\left[e_{x,t}'\left(\theta_{0}\right)\odot\varepsilon_{t}'\left(\theta_{0}\right)+\sigma_{0}'\right]\right\} \Sigma_{0}^{-2}
\end{gather*}
}separately. For the covariance involving the first term $(I)$, we
have due to the independence of $\varepsilon_{t-i}\left(\theta_{0}\right),\ i\geq1$,
on the one hand and $e_{x,t}\left(\theta_{0}\right)$ and $\varepsilon_{t}\left(\theta_{0}\right)$
on the other hand that {\scriptsize{}
\begin{align*}
\mathbb{E}\left(\left(I\right)l_{\sigma,t}'\left(\theta_{0}\right)\right) & =\mathbb{E}\left\{ \left[H'\sum_{i=1}^{q}\left(\varepsilon_{t-i}\left(\theta_{0}\right)\otimes b_{i}'B'\left(\beta_{0}\right)^{-1}\Sigma_{0}^{-1}e_{x,t}\left(\theta_{0}\right)\right)\right]\left[e_{x,t}'\left(\theta_{0}\right)\odot\varepsilon_{t}'\left(\theta_{0}\right)+\sigma_{0}'\right]\right\} \Sigma_{0}^{-2}\\
 & =H'\sum_{i=1}^{q}\left(I_{n}\otimes b_{i}'B'\left(\beta_{0}\right)^{-1}\Sigma_{0}^{-1}\right)\left\{ \mathbb{E}\left[\left(\varepsilon_{t-i}\left(\theta_{0}\right)\otimes e_{x,t}\left(\theta_{0}\right)\left[e_{x,t}'\left(\theta_{0}\right)\odot\varepsilon_{t}'\left(\theta_{0}\right)\right]\right)\right]+\underbrace{\mathbb{E}\left[\left(\varepsilon_{t-i}\left(\theta_{0}\right)\otimes e_{x,t}\left(\theta_{0}\right)\sigma_{0}'\right)\right]}_{=0}\right\} \Sigma_{0}^{-2}\\
 & =H'\sum_{i=1}^{q}\left(I_{n}\otimes b_{i}'B'\left(\beta_{0}\right)^{-1}\Sigma_{0}^{-1}\right)\left\{ \underbrace{\mathbb{E}\left(\varepsilon_{t-i}\left(\theta_{0}\right)\right)}_{=0}\otimes\mathbb{E}\left\{ e_{x,t}\left(\theta_{0}\right)\left[e_{x,t}'\left(\theta_{0}\right)\odot\varepsilon_{t}'\left(\theta_{0}\right)\right]\right\} \right\} \Sigma_{0}^{-2}=0.
\end{align*}
}For the covariance involving the second term $(II)$, we have {\scriptsize{}
\begin{align*}
\mathbb{E}\left(\left(II\right)l_{\sigma,t}'\left(\theta_{0}\right)\right) & =\mathbb{E}\left\{ \left[H'\left(\varepsilon_{t}\left(\theta\right)\otimes B'\left(\beta\right)^{-1}\Sigma_{0}^{-1}e_{x,t}\left(\theta\right)\right)\right]\left[e_{x,t}'\left(\theta_{0}\right)\odot\varepsilon_{t}'\left(\theta_{0}\right)+\sigma_{0}'\right]\right\} \Sigma_{0}^{-2}\\
 & =H'\left(I_{n}\otimes b_{i}'B'\left(\beta_{0}\right)^{-1}\Sigma_{0}^{-1}\right)\mathbb{E}\left\{ \left(\varepsilon_{t}\left(\theta\right)\otimes e_{x,t}\left(\theta\right)\left[e_{x,t}'\left(\theta_{0}\right)\odot\varepsilon_{t}'\left(\theta_{0}\right)\right]\right)\right\} \Sigma_{0}^{-2}+H'\underbrace{\mathbb{E}\left\{ \left[\left(\varepsilon_{t}\left(\theta\right)\otimes B'\left(\beta\right)^{-1}\Sigma_{0}^{-1}e_{x,t}\left(\theta\right)\right)\right]\right\} }_{=vec\left(B'\left(\beta\right)^{-1}\right)}\sigma_{0}'\Sigma_{0}^{-2}\\
 & =H'\left(I_{n}\otimes b_{i}'B'\left(\beta_{0}\right)^{-1}\Sigma_{0}^{-1}\right)\mathbb{E}\left\{ \left(\varepsilon_{t}\left(\theta\right)\otimes e_{x,t}\left(\theta\right)\left[e_{x,t}'\left(\theta_{0}\right)\odot\varepsilon_{t}'\left(\theta_{0}\right)\right]\right)\right\} \Sigma_{0}^{-2}+H'vec\left(B'\left(\beta\right)^{-1}\right)\sigma_{0}'\Sigma_{0}^{-2}
\end{align*}
}where the last line follows from \eqref{eq:mds_exp_beta}.The covariance
involving the third term $(III)$ is zero as a consequence of $\mathbb{E}\left(l_{\sigma,t}\left(\theta_{0}\right)\right)=0$. 

Finally, we obtain 
\[
\mathbb{E}\left(l_{\beta,t}\left(\theta_{0}\right)l_{\sigma,t}'\left(\theta_{0}\right)\right)=H'\left(I_{n}\otimes B'\left(\beta_{0}\right)^{-1}\Sigma_{0}^{-1}\right)\mathbb{E}\left\{ \left(\varepsilon_{t}\left(\theta\right)\otimes e_{x,t}\left(\theta\right)\left[e_{x,t}'\left(\theta_{0}\right)\odot\varepsilon_{t}'\left(\theta_{0}\right)\right]\right)\right\} \Sigma_{0}^{-2}+H'vec\left(B'\left(\beta\right)^{-1}\right)\sigma_{0}'\Sigma_{0}^{-2}
\]

\paragraph{Term $\left(\beta,\lambda\right)$.}

For the same reasons as above, the covariance with terms $\left(I\right)$
and $(III)$ are zero. Thus, we have $\mathbb{E}\left(l_{\beta,t}\left(\theta_{0}\right)l_{\lambda,t}'\left(\theta_{0}\right)\right)=-H'\left(I_{n}\otimes B'(\beta)^{-1}\Sigma^{-1}\right)\mathbb{E}\left(\varepsilon_{t}\left(\theta_{0}\right)\otimes e_{x,t}\left(\theta_{0}\right)e_{i,\lambda,t}'\right)$.

\subsection{Elements of the Covariance Matrix involving $\sigma$}

\paragraph{Diagonal term $\left(\sigma,\sigma\right)$.}

For completeness, we mention that
\[
\mathbb{E}\left(l_{\sigma,t}\left(\theta_{0}\right)l_{\sigma,t}'\left(\theta_{0}\right)\right)=\Sigma_{0}^{-2}\mathbb{E}\left[\left(e_{x,t}\left(\theta\right)\odot\varepsilon_{t}\left(\theta\right)+\sigma\right)\left(e_{x,t}\left(\theta\right)\odot\varepsilon_{t}\left(\theta\right)+\sigma\right)^{'}\right]\Sigma_{0}^{-2}.
\]

\paragraph{Term $\left(\sigma,\lambda\right)$. }

We have
\begin{align*}
\mathbb{E}\left(l_{\sigma,t}\left(\theta_{0}\right)l_{\lambda,t}'\left(\theta_{0}\right)\right) & =-\Sigma_{0}^{-2}\mathbb{E}\left(\left[e_{x,t}\left(\theta_{0}\right)\odot\varepsilon_{t}\left(\theta_{0}\right)+\sigma_{0}\right]e_{\lambda,t}'\left(\theta_{0}\right)\right)\\
 & =-\Sigma_{0}^{-2}\mathbb{E}\left(\left(e_{x,t}\left(\theta_{0}\right)\odot\varepsilon_{t}\left(\theta_{0}\right)\right)e_{\lambda,t}'\left(\theta_{0}\right)\right)
\end{align*}
because $\mathbb{E}\left(e_{\lambda,t}\left(\theta_{0}\right)\right)=0$.

\subsection{Elements of the Covariance Matrix $\lambda$}

Likewise, we mention for completeness that
\[
\mathbb{E}\left(l_{\lambda,t}\left(\theta_{0}\right)l_{\lambda,t}'\left(\theta_{0}\right)\right)=V_{e_{\lambda}}.
\]

\pagebreak{}

\section{\label{sec:opg_finite}Finiteness of Covariance Matrix of Score}

Thanks to the Cauchy-Schwarz inequality, we only need to consider
the (block-) diagonal elements of the covariance matrix of the score.

\subsection{Diagonal block pertaining to $\pi_{2}$ }

We consider
\begin{align*}
l_{\pi_{2},t}\left(\theta_{0}\right) & =-\left\{ \left[\begin{pmatrix}x_{0} & \cdots & x_{T-1}\end{pmatrix}\otimes I_{n}\right]\mathfrak{b}_{t}\right\} B'\left(\beta_{0}\right)^{-1}\Sigma_{0}^{-1}e_{x,t}\left(\theta_{0}\right).
\end{align*}
and remind the reader that $\mathcal{B}'^{-1}$ is upper-triangular
such that the product of $\left[\begin{pmatrix}x_{0} & \cdots & x_{T-1}\end{pmatrix}\otimes I_{n}\right]$
with $\mathfrak{b}_{t}$ only involves depending on time $t-1$ and
earlier. Furthermore, note that the non-zero elements of $\mathfrak{b}_{t}$
correspond to the coefficients of $b(z)^{-1}$ whose norms are decreasing
at an exponential rate. We denote the $j$-th $\left(n\times n\right)$-dimensional
block of the $\left(Tn\times n\right)$-dimensional matrix $\mathfrak{b}_{t}$
by $\mathfrak{b}_{t}^{j}$ so that $\mathfrak{b}_{t}^{1}$ corresponds
to the $t$-th coefficient of the power series expansion of $b(z)^{-1}$
around zero and that $\mathfrak{b}_{t}^{t-1}$ corresponds to the
first coefficient. We thus obtain
\begin{align*}
\mathbb{E}\left(l_{\pi_{2},t}\left(\theta_{0}\right)l_{\pi_{2},t}'\left(\theta_{0}\right)\right) & =\mathbb{E}\left(\left[\begin{pmatrix}x_{0} & \cdots & x_{T-1}\end{pmatrix}\otimes I_{n}\right]\mathfrak{b}_{t}\Sigma_{0}^{-1}B'\left(\beta_{0}\right)^{-1}e_{x,t}\left(\theta_{0}\right)e_{x,t}'\left(\theta_{0}\right)B\left(\beta_{0}\right)^{-1}\Sigma_{0}^{-1}\mathfrak{b}_{t}'\left[\begin{pmatrix}x_{0}'\\
\vdots\\
x_{T-1}'
\end{pmatrix}\otimes I_{n}\right]\right)\\
 & =\mathbb{E}\left(\left[\sum_{i=1}^{t}\left(x_{i-1}\otimes\mathfrak{b}_{t}^{i}\right)\right]\left[\Sigma_{0}^{-1}B'\left(\beta_{0}\right)^{-1}e_{x,t}\left(\theta_{0}\right)e_{x,t}'\left(\theta_{0}\right)B\left(\beta_{0}\right)^{-1}\Sigma_{0}^{-1}\right]\left[\sum_{j=1}^{t}\left(x_{j-1}'\otimes\left(\mathfrak{b}_{t}^{j}\right)'\right)\right]\right)\\
 & =\sum_{i,j=1}^{t}\mathbb{E}\left(x_{i-1}x_{j-1}'\otimes\mathfrak{b}_{t}^{i}\Sigma_{0}^{-1}B'\left(\beta_{0}\right)^{-1}e_{x,t}\left(\theta_{0}\right)e_{x,t}'\left(\theta_{0}\right)B\left(\beta_{0}\right)^{-1}\Sigma_{0}^{-1}\left(\mathfrak{b}_{t}^{j}\right)'\right)\\
 & =\sum_{i,j=1}^{t}\mathbb{E}\left(x_{i-1}x_{j-1}'\right)\otimes\left[\mathfrak{b}_{t}^{i}\Sigma_{0}^{-1}B'\left(\beta_{0}\right)^{-1}\mathbb{E}\left(e_{x,t}\left(\theta_{0}\right)e_{x,t}'\left(\theta_{0}\right)\right)B\left(\beta_{0}\right)^{-1}\Sigma_{0}^{-1}\left(\mathfrak{b}_{t}^{j}\right)'\right].
\end{align*}
where the last equation is obtained from independence of (functions
of) $\varepsilon_{t}$ of (functions of) $x_{t-k},\ k>1$. We note
that 
\[
\mathbb{E}\left(x_{i-1}x_{j-1}'\right)=\begin{pmatrix}\gamma\left(i-j\right) & \gamma\left(i-j+1\right) & \cdots & \gamma\left(i-j+p-1\right)\\
\gamma(i-j-1) & \gamma\left(i-j\right)\\
\vdots &  & \ddots & \vdots\\
\gamma\left(i-j-p+1\right) &  & \cdots & \gamma\left(i-j\right)
\end{pmatrix}
\]
whose largest eigenvalue converges to zero whenever $j$ or $i$ tend
to infinity. The same convergence property holds true for the term
$\mathfrak{b}_{t}^{i}\Sigma_{0}^{-1}B'\left(\beta_{0}\right)^{-1}\mathbb{E}\left(e_{x,t}\left(\theta_{0}\right)e_{x,t}'\left(\theta_{0}\right)\right)B\left(\beta_{0}\right)^{-1}\Sigma_{0}^{-1}\left(\mathfrak{b}_{t}^{j}\right)'$
since by Lemma B.1.(i)-(ii) in \citet{LMS_svarIdent16} $\mathbb{E}\left[e_{x,t}\left(\theta_{0}\right)e_{x,t}'\left(\theta_{0}\right)\right]$
is finite. From these facts, it follows that, for $T\rightarrow\infty$,
the double sum converges.

\subsection{Diagonal block pertaining to $\pi_{3}$ }

Similar (even though a bit easier) to the block pertaining to $\pi_{2}$,
we consider
\begin{align*}
l_{\pi_{3},t}\left(\theta_{0}\right) & =\left\{ \left[\begin{pmatrix}w_{0}\left(\theta_{0}\right) & \cdots & w_{T-1}\left(\theta_{0}\right)\end{pmatrix}\otimes I_{n}\right]\mathfrak{b}\right\} B'\left(\beta_{0}\right)^{-1}\Sigma_{0}^{-1}e_{x,t}\left(\theta_{0}\right)
\end{align*}
and obtain{\scriptsize{}
\begin{align*}
\mathbb{E}\left(l_{\pi_{3},t}\left(\theta_{0}\right)l_{\pi_{3},t}'\left(\theta_{0}\right)\right) & =\mathbb{E}\left(\left[\begin{pmatrix}w_{0}\left(\theta_{0}\right) & \cdots & w_{T-1}\left(\theta_{0}\right)\end{pmatrix}\otimes I_{n}\right]\mathfrak{b}_{t}\Sigma_{0}^{-1}B'\left(\beta_{0}\right)^{-1}e_{x,t}\left(\theta_{0}\right)e_{x,t}'\left(\theta_{0}\right)B\left(\beta_{0}\right)^{-1}\Sigma_{0}^{-1}\mathfrak{b}_{t}'\left[\begin{pmatrix}w_{0}'\left(\theta_{0}\right)\\
\vdots\\
w_{T-1}'\left(\theta_{0}\right)
\end{pmatrix}\otimes I_{n}\right]\right)\\
 & =\sum_{i,j=1}^{t}\mathbb{E}\left(w_{i-1}\left(\theta_{0}\right)w_{j-1}'\left(\theta_{0}\right)\right)\otimes\left[\mathfrak{b}_{t}^{i}\Sigma_{0}^{-1}B'\left(\beta_{0}\right)^{-1}\mathbb{E}\left(e_{x,t}\left(\theta_{0}\right)e_{x,t}'\left(\theta_{0}\right)\right)B\left(\beta_{0}\right)^{-1}\Sigma_{0}^{-1}\left(\mathfrak{b}_{t}^{j}\right)'\right].
\end{align*}
}Note that $\mathbb{E}\left(w_{i-1}\left(\theta_{0}\right)w_{j-1}'\left(\theta_{0}\right)\right)=0$
for $\left|i-j\right|\geq q$ and thus finiteness of $\mathbb{E}\left(l_{\pi_{3},t}\left(\theta_{0}\right)l_{\pi_{3},t}'\left(\theta_{0}\right)\right)$
follows from finiteness of $\mathbb{E}\left(e_{x,t}\left(\theta_{0}\right)e_{x,t}'\left(\theta_{0}\right)\right)$.

\subsection{Diagonal block pertaining to $\beta$}

The finiteness of
\begin{align*}
\mathbb{E}\left(l_{\beta,t}\left(\theta_{0}\right)l_{\beta,t}'\left(\theta_{0}\right)\right) & =H'\left[\sum_{i=1}^{q}\left(\Sigma\otimes b_{i}'B'\left(\beta_{0}\right)^{-1}\Sigma^{-1}V_{e_{x}}\Sigma^{-1}B\left(\beta_{0}\right)^{-1}b_{i}\right)\right]H\\
 & \quad+H'\left(I_{n}\otimes B'(\beta_{0})^{-1}\Sigma_{0}^{-1}\right)\mathbb{E}\left(\varepsilon_{t}\left(\theta_{0}\right)\varepsilon_{t}'\left(\theta_{0}\right)\otimes e_{x,t}\left(\theta_{0}\right)e_{x,t}'\left(\theta_{0}\right)\right)\left(I_{n}\otimes\Sigma_{0}'^{-1}B\left(\beta_{0}\right)^{-1}\right)H\\
 & \quad-H'vec\left(B'\left(\beta_{0}\right)^{-1}\right)vec\left(B'\left(\beta_{0}\right)^{-1}\right)'H.
\end{align*}
follows from the finiteness of $V_{e_{x}}$ and the finiteness of
$\mathbb{E}\left(\varepsilon_{t}\left(\theta_{0}\right)\varepsilon_{t}'\left(\theta_{0}\right)\otimes e_{x,t}\left(\theta_{0}\right)e_{x,t}'\left(\theta_{0}\right)\right)$.
The latter follows from considering 
\[
\mathbb{E}\left(\varepsilon_{i,t}\left(\theta_{0}\right)\varepsilon_{j,t}\left(\theta_{0}\right)e_{k,t}\left(\theta_{0}\right)e_{l,t}\left(\theta_{0}\right)\right)=\left\{ \begin{array}{llc}
\mathbb{E}\left(\varepsilon_{i,t}^{2}\left(\theta_{0}\right)e_{i,t}^{2}\left(\theta_{0}\right)\right) & <\infty, & i=j=k=l\\
\mathbb{E}\left[\left(\varepsilon_{i,t}\left(\theta_{0}\right)e_{i,t}\left(\theta_{0}\right)\right)\left(\varepsilon_{j,t}\left(\theta_{0}\right)e_{j,t}\left(\theta_{0}\right)\right)\right] & =\sigma_{0,i}\sigma_{0,j}, & i=k,j=l,i\neq j\\
\mathbb{E}\left(\varepsilon_{i,t}^{2}\left(\theta_{0}\right)e_{k,t}^{2}\left(\theta_{0}\right)\right) & =\sigma_{0,i}\mathbb{E}\left(e_{k,t}^{2}\left(\theta_{0}\right)\right), & i=j,k=l,i\neq k\\
0, &  & otherwise
\end{array}\right.
\]
which in turn follows from independence of component processes and
Lemma B1.(vi), i.e. $\mathbb{E}\left(\varepsilon_{i,t}^{2}\left(\theta_{0}\right)e_{i,t}^{2}\left(\theta_{0}\right)\right)<\infty$,
Lemma B1.(v), i.e. $\mathbb{E}\left(\varepsilon_{i,t}\left(\theta_{0}\right)e_{i,t}\left(\theta_{0}\right)\right)=-\sigma_{0,i}$,
Lemma B1.(ii), i.e. $\mathbb{E}\left(e_{k,t}^{2}\left(\theta_{0}\right)\right)<\infty$,
in \citet{LMS_svarIdent16}.

\subsection{Diagonal block pertaining to $\sigma$}

Remember that $\mathbb{E}\left(l_{\sigma,t}\left(\theta_{0}\right)l_{\sigma,t}'\left(\theta_{0}\right)\right)=\Sigma_{0}^{-2}\mathbb{E}\left[\left(e_{x,t}\left(\theta_{0}\right)\odot\varepsilon_{t}\left(\theta_{0}\right)+\sigma_{0}\right)\left(e_{x,t}\left(\theta_{0}\right)\odot\varepsilon_{t}\left(\theta_{0}\right)+\sigma_{0}\right)^{'}\right]\Sigma_{0}^{-2}$
is a diagonal matrix. It is a diagonal matrix because the $\left(i,j\right)$
element of $\mathbb{E}\left[\left(e_{x,t}\left(\theta_{0}\right)\odot\varepsilon_{t}\left(\theta_{0}\right)+\sigma_{0}\right)\left(e_{x,t}\left(\theta_{0}\right)\odot\varepsilon_{t}\left(\theta_{0}\right)+\sigma_{0}\right)^{'}\right]$
is of the form 
\[
\mathbb{E}\left[\left(e_{i,x,t}\left(\theta_{0}\right)\varepsilon_{i,t}\left(\theta_{0}\right)+\sigma_{0,i}\right)\left(e_{j,x,t}\left(\theta_{0}\right)\varepsilon_{j,t}\left(\theta_{0}\right)+\sigma_{0,j}\right)\right]=\mathbb{E}\left(e_{i,x,t}\left(\theta_{0}\right)\varepsilon_{i,t}\left(\theta_{0}\right)+\sigma_{0,i}\right)\mathbb{E}\left(e_{j,x,t}\left(\theta_{0}\right)\varepsilon_{j,t}\left(\theta_{0}\right)+\sigma_{0,j}\right)
\]
when $i\neq j$. Both terms in this product are zero because of Lemma
B1.(v) in \citet{LMS_svarIdent16}, i.e. $\mathbb{E}\left(\varepsilon_{i,t}\left(\theta_{0}\right)e_{i,t}\left(\theta_{0}\right)\right)=-\sigma_{0,i}$.
Thus the finiteness of the diagonal matrix 
\[
\mathbb{E}\left[\left(e_{x,t}\left(\theta_{0}\right)\odot\varepsilon_{t}\left(\theta_{0}\right)+\sigma_{0}\right)\left(e_{x,t}\left(\theta_{0}\right)\odot\varepsilon_{t}\left(\theta_{0}\right)+\sigma_{0}\right)^{'}\right]
\]
 follows because of Lemma B1.(vi) in \citet{LMS_svarIdent16}, i.e.
$\mathbb{E}\left(\varepsilon_{i,t}^{2}\left(\theta_{0}\right)e_{i,t}^{2}\left(\theta_{0}\right)\right)<\infty$.

\subsection{Diagonal block pertaining to $\lambda$}

The finiteness of $\mathbb{E}\left(l_{\lambda,t}\left(\theta_{0}\right)l_{\lambda,t}'\left(\theta_{0}\right)\right)=V_{e_{\lambda}}$
follows from the independence of the component processes and Lemma
B1.(iii) in \citet{LMS_svarIdent16}, i.e. $\mathbb{E}\left(e_{i,\lambda_{i},t}\left(\theta_{0}\right)\right)<\infty$.

\pagebreak{}

\section{\label{sec:hessian_expr}Hessian}

We start by analyzing the diagonal terms $l_{\pi_{2}\pi_{2},t}\left(\theta\right)$,
$l_{\pi_{3}\pi_{3},t}\left(\theta\right)$, and then go on with $l_{\pi_{2}\pi_{3},t}\left(\theta\right)$.
Subsequently, we analyze the remaining terms involving $\beta$, starting
with the diagonal block. Finally, the terms involving $\sigma$ and
$\lambda$, which are simpler than the other ones and similar to the
ones derived in \citet{LMS_svarIdent16}, are analyzed.

The expression of the Hessian is needed in order to show that $\mathbb{E}\left(\sup_{\Theta_{0}}\left\Vert l_{\theta\theta,t}\left(\theta\right)\right\Vert \right)$
is finite. This fact, in turn, is needed as input for the Theorem
stating that $\sup_{\Theta_{0}}\left\Vert \frac{1}{T}\sum_{t=1}^{T}l_{\theta\theta,t}\left(\theta\right)-\mathbb{E}\left(l_{\theta\theta,t}\left(\theta\right)\right)\right\Vert $
converges almost surely to zero which is necessary to show that the
MLE is asymptotically normally distributed. 

The following derivations sometimes contain arrays with more than
two dimensions, e.g. as a consequence of taking derivatives of a vector
with respect to a matrix. While tensor index notation might prove
useful, we opted for sequential vectorization, leading to high-dimensional
matrices. 

\subsection{Elements of the Hessian involving $\pi_{2}$ }

\subsubsection{Diagonal term $\left(\pi_{2},\pi_{2}\right)$}

\paragraph{Intermediate step for $l_{\pi_{2}\pi_{2},t}\left(\theta\right)$.}

We need to calculate 
\begin{align*}
\frac{\partial}{\partial\pi_{2}'}\left(\frac{\partial l_{t}\left(\theta\right)}{\partial\pi_{2}}\right) & =\frac{\partial}{\partial\pi_{2}'}\left(\frac{\partial u_{t}\left(\theta\right)'}{\partial\pi_{2}}B'\left(\beta\right)^{-1}\Sigma^{-1}e_{x,t}\left(\theta\right)\right)\\
 & =\sum_{i=1}^{n}\underbrace{\left[\frac{\partial}{\partial\pi_{2}'}\left(\frac{\partial u_{t}\left(\theta\right)'\iota_{i}}{\partial\pi_{2}}\right)\right]}_{=(A)}\iota_{i}'B'\left(\beta\right)^{-1}\Sigma^{-1}e_{x,t}\left(\theta\right)+\frac{\partial u_{t}\left(\theta\right)'}{\partial\pi_{2}}\frac{\partial}{\partial\pi_{2}'}\left(B'\left(\beta\right)^{-1}\Sigma^{-1}e_{x,t}\left(\theta\right)\right).
\end{align*}

\paragraph{The derivative of $\frac{\partial u_{t}\left(\theta\right)'}{\partial\pi_{2}}$
with respect to $\pi_{2}$.}

Note that term (A) is zero because taking the derivative of 
\[
\frac{\partial u_{t}'\left(\theta\right)}{\partial\pi_{2}}=-\left(x_{t-1}\otimes I_{n}\right)-\left(\frac{\partial u_{t-1}'\left(\theta\right)}{\partial\pi_{2}},\ldots,\frac{\partial u_{t-q}'\left(\theta\right)}{\partial\pi_{2}}\right)\begin{pmatrix}b_{1}'\\
\vdots\\
b_{q}'
\end{pmatrix}.
\]
with respect to $\pi_{2}'$ leads to
\begin{align*}
\left(I_{n}+b_{1}z+\cdots+b_{q}z^{q}\right)\frac{\partial u_{t}'\left(\theta\right)}{\partial\pi_{2}'\partial\pi_{2}} & =-\frac{\partial}{\partial\pi_{2}'}\left(x_{t-1}\otimes I_{n}\right)\\
 & =0.
\end{align*}

Thus, we are left with $\frac{\partial}{\partial\pi_{2}'}\left(\frac{\partial l_{t}\left(\theta\right)}{\partial\pi_{2}}\right)=\left(\frac{\partial u_{t}\left(\theta\right)'}{\partial\pi_{2}}\right)B'\left(\beta\right)^{-1}\Sigma^{-1}\left(\frac{\partial e_{x,t}\left(\theta\right)}{\partial\pi_{2}'}\right).$

\paragraph{Calculating $\left(\frac{\partial e_{x,t}\left(\theta\right)}{\partial\pi_{2}'}\right)$.}

For the univariate term, we obtain $\frac{\partial e_{i,x,t}\left(\theta\right)}{\partial\pi_{2}'}=e_{i,xx,t}\left(\theta\right)\sigma_{i}^{-1}\iota_{i}'B(\beta)^{-1}\frac{\partial u_{t}\left(\theta\right)}{\partial\pi_{2}'}.$

\paragraph{Result.}

Combining the results above, and using $\frac{\partial u_{t}\left(\theta\right)'}{\partial\pi_{2}}=-\left[\begin{pmatrix}x_{0} & \cdots & x_{T-1}\end{pmatrix}\otimes I_{n}\right]\mathfrak{b}_{t}=-x_{b,t-1}\left(\theta\right),$
we result in \footnote{Remember that $e_{xx,t}\left(\theta\right)=diag\left(e_{1,xx,t}\left(\theta\right),\ldots,e_{n,xx,t}\left(\theta\right)\right)$}
\begin{align*}
l_{\pi_{2}\pi_{2},t}\left(\theta\right) & =\frac{\partial u_{t}\left(\theta\right)'}{\partial\pi_{2}}\left(B'\left(\beta\right)^{-1}\Sigma^{-1}\right)\begin{pmatrix}e_{1,xx,t}\left(\theta\right)\sigma_{1}^{-1} & 0 & 0\\
0 & \ddots & 0\\
0 & 0 & e_{n,xx,t}\left(\theta\right)\sigma_{n}^{-1}
\end{pmatrix}B(\beta)^{-1}\frac{\partial u_{t}\left(\theta\right)}{\partial\pi_{2}'}\\
 & =\frac{\partial u_{t}\left(\theta\right)'}{\partial\pi_{2}}B'\left(\beta\right)^{-1}\Sigma^{-1}e_{xx,t}\left(\theta\right)\Sigma^{-1}B(\beta)^{-1}\frac{\partial u_{t}\left(\theta\right)}{\partial\pi_{2}'}\\
 & =x_{b,t-1}\left(\theta\right)B'\left(\beta\right)^{-1}\Sigma^{-1}e_{xx,t}\left(\theta\right)\Sigma^{-1}B(\beta)^{-1}x_{b,t-1}'\left(\theta\right).
\end{align*}

\subsection{Elements of the Hessian involving $\pi_{3}$ }

\subsubsection{Diagonal term $\left(\pi_{3},\pi_{3}\right)$}

\paragraph{Intermediate step for $l_{\pi_{3}\pi_{3},t}\left(\theta\right)$.}

We need to calculate
\[
\frac{\partial}{\partial\pi_{3}'}\left(\frac{\partial l_{t}\left(\theta\right)}{\partial\pi_{3}}\right)=\sum_{i=1}^{n}\underbrace{\left[\frac{\partial}{\partial\pi_{3}'}\left(\frac{\partial u_{t}\left(\theta\right)'\iota_{i}}{\partial\pi_{3}}\right)\right]}_{=(A)}\iota_{i}'B'\left(\beta\right)^{-1}\Sigma^{-1}e_{x,t}\left(\theta\right)+\left(\frac{\partial u_{t}\left(\theta\right)'}{\partial\pi_{3}}\right)B'\left(\beta\right)^{-1}\Sigma^{-1}\left(\frac{\partial e_{x,t}\left(\theta\right)}{\partial\pi_{3}'}\right).
\]

\paragraph{Derivative of $\frac{\partial u_{t}'\left(\theta\right)}{\partial\pi_{3}}$
with respect to $\pi_{3}$.}

Term (A) is more complicated than in the case of $\left(\pi_{2},\pi_{2}\right)$.
The derivative of 
\[
\frac{\partial u_{t}'\left(\theta\right)}{\partial\pi_{3}}=-\left(w_{t-1}\left(\theta\right)\otimes I_{n}\right)-\left[\begin{pmatrix}\frac{\partial u_{t-1}'\left(\theta\right)}{\partial\pi_{3}} & \cdots & \frac{\partial u_{t-q}'\left(\theta\right)}{\partial\pi_{3}}\end{pmatrix}\begin{pmatrix}b_{1}'\\
\vdots\\
b_{q}'
\end{pmatrix}\right].
\]
with respect to $\pi_{3}'$ is calculated in the following steps.
First, we vectorize $\frac{\partial u_{t}'\left(\theta\right)}{\partial\pi_{3}}$,
second we calculate $\frac{\partial}{\partial\pi_{3}'}\left(vec\left(\frac{\partial u_{t}'\left(\theta\right)}{\partial\pi_{3}}\right)\right)$
for one point in time, third we obtain the system involving all points
in time and solve it in order to obtain $\frac{\partial}{\partial\pi_{3}'}\left(vec\left(\frac{\partial u_{t}'\left(\theta\right)}{\partial\pi_{3}}\right)\right)$
in terms of observables $y_{1},\ldots,y_{T}$.

We apply the following differentiation rules to the vectorized equation
\[
vec\left(\frac{\partial u_{t}'\left(\theta\right)}{\partial\pi_{3}}\right)=-vec\left(w_{t-1}\left(\theta\right)\otimes I_{n}\right)-vec\left[I_{n^{2}q}\begin{pmatrix}\frac{\partial u_{t-1}'\left(\theta\right)}{\partial\pi_{3}} & \cdots & \frac{\partial u_{t-q}'\left(\theta\right)}{\partial\pi_{3}}\end{pmatrix}\begin{pmatrix}b_{1}'\\
\vdots\\
b_{q}'
\end{pmatrix}I_{n}\right].
\]
 According to \citet{MagnusNeudecker07} (Chapter 3 Section 7, page
55 , Theorem 10), we have 
\[
vec\left(A_{m\times n}\otimes B_{p\times q}\right)=\left(I_{n}\otimes K_{q,m}\otimes I_{p}\right)\left[vec(A_{m\times n})\otimes vec\left(B_{p\times q}\right)\right]
\]
 and its differential (see Chapter 9, Section 14, page 209, formula
(11)) is 
\begin{align*}
dvec\left(A_{m\times n}\otimes B_{p\times q}\right) & =\left(I_{n}\otimes K_{q,m}\otimes I_{p}\right)d\left[vec(A_{m\times n})\otimes vec\left(B_{p\times q}\right)\right]\\
 & =\left(I_{n}\otimes K_{q,m}\otimes I_{p}\right)\left\{ \left[I_{mn}\otimes vec\left(B_{p\times q}\right)\right]dvec(A_{m\times n})+\left[vec\left(A_{m\times n}\right)\otimes I_{pq}\right]dvec\left(B_{p\times q}\right)\right\} .
\end{align*}
where we use the commutation matrix defined by $K_{m,n}vec\left(A_{m\times n}\right)=vec\left(A'\right)$.

We thus obtain
\begin{align*}
\frac{\partial}{\partial\pi_{3}'}vec\left(\frac{\partial u_{t}'\left(\theta\right)}{\partial\pi_{3}}\right) & =-\left(K_{n,nq}\otimes I_{n}\right)\left[\frac{\partial w_{t-1}\left(\theta\right)}{\partial\pi_{3}'}\otimes vec\left(I_{n}\right)\right]+\\
 & \quad-vec\left[\begin{pmatrix}b_{1} & \cdots & b_{q}\end{pmatrix}\otimes I_{n^{2}q}\right]\frac{\partial}{\partial\pi_{3}'}vec\begin{pmatrix}\frac{\partial u_{t-1}'\left(\theta\right)}{\partial\pi_{3}} & \cdots & \frac{\partial u_{t-q}'\left(\theta\right)}{\partial\pi_{3}}\end{pmatrix}\\
 & \quad-\left[I_{n}\otimes\begin{pmatrix}\frac{\partial u_{t-1}'\left(\theta\right)}{\partial\pi_{3}} & \cdots & \frac{\partial u_{t-q}'\left(\theta\right)}{\partial\pi_{3}}\end{pmatrix}\right]\underbrace{\frac{\partial}{\partial\pi_{3}'}vec\left[\begin{pmatrix}b_{1}'\\
\vdots\\
b_{q}'
\end{pmatrix}\right]}_{=K_{n,nq}\frac{\partial}{\partial\pi_{3}'}vec\begin{pmatrix}b_{1} & \cdots & b_{q}\end{pmatrix}}
\end{align*}

\paragraph{Result for $\frac{\partial}{\partial\pi_{3}'}\left(vec\left(\frac{\partial u_{t}'\left(\theta\right)}{\partial\pi_{3}}\right)\right)$
for one point in time.}

Reordering the equation above, we obtain
\[
\left[\begin{pmatrix}I_{n} & b_{1} & \cdots & b_{q}\end{pmatrix}\otimes I_{n^{2}q}\right]\begin{pmatrix}\frac{\partial}{\partial\pi_{3}'}\left[vec\left(\frac{\partial u_{t}'\left(\theta\right)}{\partial\pi_{3}}\right)\right]\\
\vdots\\
\frac{\partial}{\partial\pi_{3}'}\left[vec\left(\frac{\partial u_{t-q}'\left(\theta\right)}{\partial\pi_{3}}\right)\right]
\end{pmatrix}=-\left(K_{n,nq}\otimes I_{n}\right)\left[\frac{\partial w_{t-1}\left(\theta\right)}{\partial\pi_{3}'}\otimes vec\left(I_{n}\right)\right]-\left[I_{n}\otimes\frac{\partial w_{t-1}'\left(\theta\right)}{\partial\pi_{3}}\right]K_{n,nq},
\]

or equivalently
\begin{equation}
\left[\begin{pmatrix}b_{q} & \cdots & b_{1} & I_{n}\end{pmatrix}\otimes I_{n^{2}q}\right]\begin{pmatrix}\frac{\partial}{\partial\pi_{3}'}\left[vec\left(\frac{\partial u_{t-q}'\left(\theta\right)}{\partial\pi_{3}}\right)\right]\\
\vdots\\
\frac{\partial}{\partial\pi_{3}'}\left[vec\left(\frac{\partial u_{t}'\left(\theta\right)}{\partial\pi_{3}}\right)\right]
\end{pmatrix}=-\left(K_{n,nq}\otimes I_{n}\right)\left[\frac{\partial w_{t-1}\left(\theta\right)}{\partial\pi_{3}'}\left(\theta\right)\otimes vec\left(I_{n}\right)\right]-\left[I_{n}\otimes\frac{\partial w_{t-1}'\left(\theta\right)}{\partial\pi_{3}}\right]K_{n,nq}.\label{eq:h_pi3pi3_one}
\end{equation}

\paragraph{The terms involved in the system for $\frac{\partial}{\partial\pi_{3}'}vec\left(\frac{\partial u_{t}'\left(\theta\right)}{\partial\pi_{3}}\right)$
for all points in time.}

The left-hand-side of \eqref{eq:h_pi3pi3_one} for all points in time
is

{\tiny{}
\[
\underbrace{\begin{pmatrix}\left(I_{n}\otimes I_{n^{2}q}\right) & 0_{n^{3}q} & \cdots &  &  &  &  & \cdots & 0_{n^{3}q}\\
\left(b_{1}\otimes I_{n^{2}q}\right) & \left(I_{n}\otimes I_{n^{2}q}\right) & 0_{n^{3}q} &  &  &  &  &  & \vdots\\
\vdots &  & \ddots & \ddots\\
\left(b_{q}\otimes I_{n^{2}q}\right) &  & \left(b_{1}\otimes I_{n^{2}q}\right) & \left(I_{n}\otimes I_{n^{2}q}\right)\\
0_{n^{3}q} & \left(b_{q}\otimes I_{n^{2}q}\right) &  & \left(b_{1}\otimes I_{n^{2}q}\right) & \left(I_{n}\otimes I_{n^{2}q}\right)\\
\vdots & 0_{n^{3}q}\\
 &  & \ddots &  &  &  & \ddots & \ddots & \vdots\\
\vdots &  &  & 0_{n^{3}q} & \left(b_{q}\otimes I_{n^{2}q}\right) &  & \left(b_{1}\otimes I_{n^{2}q}\right) & \left(I_{n}\otimes I_{n^{2}q}\right) & 0_{n^{3}q}\\
0_{n^{3}q} & \cdots &  & \cdots & 0_{n^{3}q} & \left(b_{q}\otimes I_{n^{2}q}\right) &  & \left(b_{1}\otimes I_{n^{2}q}\right) & \left(I_{n}\otimes I_{n^{2}q}\right)
\end{pmatrix}}_{=\mathcal{C}}\begin{pmatrix}\frac{\partial}{\partial\pi_{3}'}\left[vec\left(\frac{\partial u_{1}'\left(\theta\right)}{\partial\pi_{3}}\right)\right]\\
\vdots\\
\frac{\partial}{\partial\pi_{3}'}\left[vec\left(\frac{\partial u_{T}'\left(\theta\right)}{\partial\pi_{3}}\right)\right]
\end{pmatrix}
\]
}where the matrix $\mathcal{C}$ is $\left(Tn^{3}q\times Tn^{3}q\right)$-dimensional.

Using the rule, $\left[\left(A_{1},A_{2}\right)\otimes B\right]=\left[\left(A_{1}\otimes B\right),\left(A_{2}\otimes B\right)\right]$,
see \citet{seber08} (Chapter 11.3, page 235), and the rule 
\begin{align*}
K_{p,m}\left(A_{m\times n}\otimes B_{p\times q}\right) & =\left(B_{p\times q}\otimes A_{m\times n}\right)K_{q,n}\\
\left(A_{m\times n}\otimes B_{p\times q}\right) & =K_{m,p}\left(B_{p\times q}\otimes A_{m\times n}\right)K_{q,n}
\end{align*}
for commutation matrices in \citet{MagnusNeudecker07} (Chapter 3,
Section 7, page 55, Theorem 9), we obtain that the first term on the
right-hand-side of \eqref{eq:h_pi3pi3_one} for all points in time
is 
\[
-\begin{pmatrix}\left(K_{n,nq}\otimes I_{n}\right)\\
 & \ddots\\
 &  & \left(K_{n,nq}\otimes I_{n}\right)
\end{pmatrix}\begin{pmatrix}\left[\left(\frac{\partial w_{0}\left(\theta\right)}{\partial\pi_{3}'}\right)\otimes vec\left(I_{n}\right)\right]\\
\vdots\\
\left[\left(\frac{\partial w_{T-1}\left(\theta\right)}{\partial\pi_{3}'}\right)\otimes vec\left(I_{n}\right)\right]
\end{pmatrix}=\underbrace{-\left(I_{T}\otimes K_{n,nq}\otimes I_{n}\right)\left[\begin{pmatrix}\frac{\partial w_{0}\left(\theta\right)}{\partial\pi_{3}'}\\
\vdots\\
\frac{\partial w_{T-1}\left(\theta\right)}{\partial\pi_{3}'}
\end{pmatrix}\otimes vec\left(I_{n}\right)\right]}_{=\mathcal{D}}
\]
and the second term is 
\[
-\begin{pmatrix}\left[I_{n}\otimes\underbrace{\frac{\partial w_{0}'\left(\theta\right)}{\partial\pi_{3}}}_{=\left(n^{2}p\times n\right)}\right]\\
\vdots\\
\left[I_{n}\otimes\frac{\partial w_{T-1}'\left(\theta\right)}{\partial\pi_{3}}\right]
\end{pmatrix}K_{n,nq}=\begin{pmatrix}K_{n,n^{2}q}\left[\frac{\partial w_{0}'\left(\theta\right)}{\partial\pi_{3}}\otimes I_{n}\right]K_{n,n}\\
\vdots\\
K_{n,n^{2}q}\left[\frac{\partial w_{T}'\left(\theta\right)}{\partial\pi_{3}}\otimes I_{n}\right]K_{n,n}
\end{pmatrix}K_{n,nq}=\underbrace{-\left(I_{T}\otimes K_{n,n^{2}q}\right)\begin{pmatrix}\left[\frac{\partial w_{0}'\left(\theta\right)}{\partial\pi_{3}}\otimes I_{n}\right]\\
\vdots\\
\left[\frac{\partial w_{T}'\left(\theta\right)}{\partial\pi_{3}}\otimes I_{n}\right]
\end{pmatrix}K_{n,n}K_{n,nq}}_{=\mathcal{E}}
\]

\paragraph{Result for $\frac{\partial}{\partial\pi_{3}'}\left[vec\left(\frac{\partial u_{t}'\left(\theta\right)}{\partial\pi_{3}}\right)\right]$
.}

Finally, we have $\frac{\partial}{\partial\pi_{3}'}\left[vec\left(\frac{\partial u_{t}'\left(\theta\right)}{\partial\pi_{3}}\right)\right]=\left(\mathcal{C}^{-1}\right)_{\left[t,\bullet\right]}\left(\mathcal{D}+\mathcal{E}\right)$
where the subscript $\left[t,\bullet\right]$ corresponds to choosing
the $t$-th $\left(n^{3}q\times Tn^{3}q\right)$-dimensional block
of rows in $\mathcal{C}^{-1}$. 

\paragraph{Result for $\frac{\partial}{\partial\pi_{3}'}\left[\frac{\partial l_{t}\left(\theta\right)}{\partial\pi_{3}}\right]$.}

We have 
\begin{align*}
\frac{\partial}{\partial\pi_{3}'}\left[vec\left(\frac{\partial l_{t}\left(\theta\right)}{\partial\pi_{3}}\right)\right] & =\frac{\partial}{\partial\pi_{3}'}\left[vec\left(I_{n^{2}p}\frac{\partial u_{t}'\left(\theta\right)}{\partial\pi_{3}}B'\left(\beta\right)^{-1}\Sigma^{-1}e_{x,t}\left(\theta\right)\right)\right]\\
 & =\frac{\partial}{\partial\pi_{3}'}\left[\frac{1}{2}\left(e_{x,t}'\left(\theta\right)\Sigma^{-1}B'\left(\beta\right)^{-1}\otimes I_{n^{2}p}\right)vec\left(\frac{\partial u_{t}'\left(\theta\right)}{\partial\pi_{3}}\right)+\frac{1}{2}\left(\frac{\partial u_{t}'\left(\theta\right)}{\partial\pi_{3}}B'\left(\beta\right)^{-1}\Sigma^{-1}\right)e_{x,t}\left(\theta\right)\right]
\end{align*}
such that 
\[
\frac{\partial}{\partial\pi_{3}'}\left[\frac{\partial l_{t}\left(\theta\right)}{\partial\pi_{3}}\right]=\left(e_{x,t}'\left(\theta\right)\Sigma^{-1}B\left(\beta\right)^{-1}\otimes I_{n^{2}p}\right)\begin{pmatrix}\frac{\partial^{2}u_{t,1}\left(\theta\right)}{\partial\pi_{3}'\partial\pi_{3}}\\
\vdots\\
\frac{\partial^{2}u_{t,n}\left(\theta\right)}{\partial\pi_{3}'\partial\pi_{3}}
\end{pmatrix}+\left(\frac{\partial u_{t}'\left(\theta\right)}{\partial\pi_{3}}\right)B'\left(\beta\right)^{-1}\Sigma^{-1}\left(\frac{\partial e_{x,t}\left(\theta\right)}{\partial\pi_{3}'}\right).
\]

Above, we calculated $\begin{pmatrix}\frac{\partial^{2}u_{t,1}\left(\theta\right)}{\partial\pi_{3}'\partial\pi_{3}}\\
\vdots\\
\frac{\partial^{2}u_{t,n}\left(\theta\right)}{\partial\pi_{3}'\partial\pi_{3}}
\end{pmatrix}$ in terms of $\frac{\partial u_{t}'\left(\theta\right)}{\partial\pi_{3}}=\left[\begin{pmatrix}w_{0}\left(\theta\right) & \cdots & w_{T-1}\left(\theta\right)\end{pmatrix}\otimes I_{n}\right]\mathfrak{b}_{t}=w_{b,t-1}\left(\theta\right)$
and\footnote{Note that 
\begin{align*}
\frac{\partial e_{x,t}\left(\theta\right)}{\partial\pi_{3}'} & =\sum_{i=1}^{n}e_{i,xx,t}\left(\theta\right)\frac{\partial}{\partial\pi_{3}'}\left(\sigma_{i}^{-1}\iota_{i}'B(\beta)^{-1}u_{t}\left(\theta\right)\right)=\sum_{i=1}^{n}e_{i,xx,t}\left(\theta\right)\sigma_{i}^{-1}\iota_{i}'B(\beta)^{-1}\frac{\partial u_{t}\left(\theta\right)}{\partial\pi_{3}'}\\
 & =e_{xx,t}\left(\theta\right)\Sigma^{-1}B\left(\beta\right)^{-1}\frac{\partial u_{t}\left(\theta\right)}{\partial\pi_{3}'}
\end{align*}
} $\frac{\partial e_{x,t}\left(\theta\right)}{\partial\pi_{3}'}=e_{xx,t}\left(\theta\right)\Sigma^{-1}B\left(\beta\right)^{-1}\frac{\partial u_{t}\left(\theta\right)}{\partial\pi_{3}'}$
such that we obtain as result
\begin{align*}
\frac{\partial}{\partial\pi_{3}'}\left[\frac{\partial l_{t}\left(\theta\right)}{\partial\pi_{3}}\right] & =\left(e_{x,t}'\left(\theta\right)\Sigma^{-1}B'\left(\beta\right)^{-1}\otimes I_{n^{2}p}\right)\left(\mathcal{C}^{-1}\right)_{\left[t,\bullet\right]}\left(\mathcal{D}+\mathcal{E}\right)+w_{b,t-1}\left(\theta\right)B'\left(\beta\right)^{-1}\Sigma^{-1}e_{xx,t}\left(\theta\right)\Sigma^{-1}B\left(\beta\right)^{-1}w_{b,t-1}'\left(\theta\right).
\end{align*}

\pagebreak{}

\subsubsection{Term $\left(\pi_{2},\pi_{3}\right)$}

\paragraph{Intermediate step for $l_{\pi_{2}\pi_{3},t}\left(\theta\right)$.}

We calculate 
\begin{align*}
\frac{\partial}{\partial\pi_{3}'}\left(\frac{\partial l_{t}\left(\theta\right)}{\partial\pi_{2}}\right) & =\frac{\partial}{\partial\pi_{3}'}\left[vec\left(\frac{\partial l_{t}\left(\theta\right)}{\partial\pi_{2}}\right)\right]\\
 & =\frac{\partial}{\partial\pi_{3}'}\left[vec\left(I_{n^{2}p}\frac{\partial u_{t}\left(\theta\right)'}{\partial\pi_{2}}B'\left(\beta\right)^{-1}\Sigma^{-1}e_{x,t}\right)\right]\\
 & =\left(e_{x,t}'\left(\theta\right)\Sigma^{-1}B\left(\beta\right)^{-1}\otimes I_{n^{2}p}\right)\left\{ \frac{\partial}{\partial\pi_{3}'}\left[vec\left(\frac{\partial u_{t}\left(\theta\right)'}{\partial\pi_{2}}\right)\right]\right\} +\left(\frac{\partial u_{t}\left(\theta\right)'}{\partial\pi_{2}}\right)B'\left(\beta\right)^{-1}\Sigma^{-1}\left(\frac{\partial e_{x,t}\left(\theta\right)}{\partial\pi_{3}'}\right).
\end{align*}

\paragraph{Show that $\frac{\partial}{\partial\pi_{3}'}\left[vec\left(\frac{\partial u_{t}\left(\theta\right)'}{\partial\pi_{2}}\right)\right]$
is zero.}

Vectorizing $\frac{\partial u_{t}\left(\theta\right)'}{\partial\pi_{2}}$
leads to 
\begin{align*}
vec\left(\frac{\partial u_{t}\left(\theta\right)'}{\partial\pi_{2}}\right) & =-vec\left(x_{t-1}\otimes I_{n}\right)-vec\left[I_{n^{2}p}\left(\frac{\partial u_{t-1}'\left(\theta\right)}{\partial\pi_{2}},\ldots,\frac{\partial u_{t-q}'\left(\theta\right)}{\partial\pi_{2}}\right)\begin{pmatrix}b_{1}'\\
\vdots\\
b_{q}'
\end{pmatrix}\right]\\
 & =-vec\left(x_{t-1}\otimes I_{n}\right)-vec\left[\begin{pmatrix}b_{1} & \cdots & b_{q}\end{pmatrix}\otimes I_{n^{2}p}\right]vec\left(\frac{\partial u_{t-1}'\left(\theta\right)}{\partial\pi_{2}},\ldots,\frac{\partial u_{t-q}'\left(\theta\right)}{\partial\pi_{2}}\right)\\
 & =-vec\left(x_{t-1}\otimes I_{n}\right)-vec\left[\begin{pmatrix}b_{1} & \cdots & b_{q}\end{pmatrix}\otimes I_{n^{2}p}\right]\begin{pmatrix}vec\left(\frac{\partial u_{t-1}'\left(\theta\right)}{\partial\pi_{2}}\right)\\
\vdots\\
vec\left(\frac{\partial u_{t-q}'\left(\theta\right)}{\partial\pi_{2}}\right)
\end{pmatrix}
\end{align*}
Since the derivative of $-vec\left(x_{t-1}\otimes I_{n}\right)$ with
respect to $\pi_{3}'$ is zero, we obtain that the only solution of
\[
\frac{\partial}{\partial\pi_{3}'}\left[vec\left(\frac{\partial u_{t}\left(\theta\right)'}{\partial\pi_{2}}\right)\right]=-vec\left[\begin{pmatrix}b_{1} & \cdots & b_{q}\end{pmatrix}\otimes I_{n^{2}p}\right]\begin{pmatrix}\frac{\partial}{\partial\pi_{3}'}\left[vec\left(\frac{\partial u_{t-1}'\left(\theta\right)}{\partial\pi_{2}}\right)\right]\\
\vdots\\
\frac{\partial}{\partial\pi_{3}'}\left[vec\left(\frac{\partial u_{t-q}'\left(\theta\right)}{\partial\pi_{2}}\right)\right]
\end{pmatrix}
\]
is the trivial one.

\paragraph{Calculating $\left(\frac{\partial}{\partial\pi_{2}'}e_{x,t}\right)$.}

We have

\begin{align*}
\frac{\partial e_{x,t}\left(\theta\right)}{\partial\pi_{3}'} & =\sum_{i=1}^{n}e_{i,xx,t}\left(\theta\right)\frac{\partial}{\partial\pi_{3}'}\left(\sigma_{i}^{-1}\iota_{i}'B(\beta)^{-1}u_{t}\left(\theta\right)\right)\\
 & =\sum_{i=1}^{n}e_{i,xx,t}\left(\theta\right)\sigma_{i}^{-1}\iota_{i}'B(\beta)^{-1}\frac{\partial u_{t}\left(\theta\right)}{\partial\pi_{3}'}\\
 & =e_{xx,t}\left(\theta\right)\Sigma^{-1}B\left(\beta\right)^{-1}\frac{\partial u_{t}\left(\theta\right)}{\partial\pi_{3}'}
\end{align*}

\paragraph{Result.}

Collecting all terms above, we obtain that 
\begin{align*}
l_{\pi_{2}\pi_{3},t}\left(\theta\right) & =\left(\frac{\partial u_{t}\left(\theta\right)'}{\partial\pi_{2}}\right)B'\left(\beta\right)^{-1}\Sigma^{-1}e_{xx,t}\left(\theta\right)\Sigma^{-1}B\left(\beta\right)^{-1}\frac{\partial u_{t}\left(\theta\right)}{\partial\pi_{3}'}.\\
 & =x_{b,t-1}\left(\theta\right)B'\left(\beta\right)^{-1}\Sigma^{-1}e_{xx,t}\left(\theta\right)\Sigma^{-1}B\left(\beta\right)^{-1}w_{b,t-1}'\left(\theta\right).
\end{align*}

\subsection{Elements of the Hessian involving $\beta$ }

\subsubsection{Diagonal term $\left(\beta,\beta\right)$}

We start from \eqref{eq:lik_beta} and calculate 
\begin{align*}
\frac{\partial}{\partial\beta'}\left(\frac{\partial l_{t}\left(\theta\right)}{\partial\beta}\right) & =-H'\underbrace{\frac{\partial}{\partial\beta'}\left[\left(B\left(\beta\right)^{-1}u_{t}\left(\theta\right)\otimes B'\left(\beta\right)^{-1}\Sigma^{-1}e_{x,t}\left(\theta\right)\right)\right]}_{=(A)}+\cdots\\
 & \quad\cdots+\underbrace{\frac{\partial}{\partial\beta'}\left(\frac{\partial u_{t}'\left(\theta\right)}{\partial\beta}B'\left(\beta\right)^{-1}\Sigma^{-1}e_{x,t}\left(\theta\right)\right)}_{=(B)}-H'\underbrace{\frac{\partial}{\partial\beta'}vec\left(B(\beta)^{-1'}\right)}_{=(C)}.
\end{align*}

\paragraph{Calculate derivative of $e_{x,t}\left(\theta\right)$ with respect
to $\beta'$.}

The univariate version is

\begin{align*}
\frac{\partial e_{i,x,t}\left(\theta\right)}{\partial\beta'} & =e_{i,xx,t}\left(\theta\right)\sigma_{i}^{-1}\frac{\partial}{\partial\beta'}\left(\iota_{i}'B\left(\beta\right)^{-1}u_{t}\left(\theta\right)\right)\\
 & =e_{i,xx,t}\left(\theta\right)\sigma_{i}^{-1}\left[\left(u_{t}'\left(\theta\right)\otimes\iota_{i}'\right)\frac{\partial}{\partial\beta'}\left[vec\left(B\left(\beta\right)^{-1}\right)\right]+\iota_{i}'B\left(\beta\right)^{-1}\left(\frac{\partial u_{t}\left(\theta\right)}{\partial\beta'}\right)\right]\\
 & =e_{i,xx,t}\left(\theta\right)\sigma_{i}^{-1}\left[-\left(u_{t}'\left(\theta\right)\otimes\iota_{i}'\right)\left(B'\left(\beta\right)^{-1}\otimes B\left(\beta\right)^{-1}\right)H+\iota_{i}'B\left(\beta\right)^{-1}\left(\frac{\partial u_{t}\left(\theta\right)}{\partial\beta'}\right)\right]\\
 & =-e_{i,xx,t}\left(\theta\right)\sigma_{i}^{-1}\left(u_{t}'B'\left(\beta\right)^{-1}\left(\theta\right)\otimes\iota_{i}'B\left(\beta\right)^{-1}\right)H+e_{i,xx,t}\left(\theta\right)\sigma_{i}^{-1}\iota_{i}'B\left(\beta\right)^{-1}\left(\frac{\partial u_{t}\left(\theta\right)}{\partial\beta'}\right)
\end{align*}

The multivariate version is 
\begin{align*}
\frac{\partial e_{x,t}}{\partial\beta'} & =-e_{xx,t}\left(\theta\right)\Sigma^{-1}\left(u_{t}'B'\left(\beta\right)^{-1}\left(\theta\right)\otimes B\left(\beta\right)^{-1}\right)H+e_{xx,t}\left(\theta\right)\Sigma^{-1}B\left(\beta\right)^{-1}\left(\frac{\partial u_{t}\left(\theta\right)}{\partial\beta'}\right)\\
 & =-\left(u_{t}'B'\left(\beta\right)^{-1}\left(\theta\right)\otimes e_{xx,t}\left(\theta\right)\Sigma^{-1}B\left(\beta\right)^{-1}\right)H+e_{xx,t}\left(\theta\right)\Sigma^{-1}B\left(\beta\right)^{-1}\left(\frac{\partial u_{t}\left(\theta\right)}{\partial\beta'}\right)
\end{align*}

\paragraph{Intermediate result for term (A).}

We consider 
\begin{align*}
\frac{\partial}{\partial\beta'}\left(B\left(\beta\right)^{-1}u_{t}\left(\theta\right)\otimes B'\left(\beta\right)^{-1}\Sigma^{-1}e_{x,t}\left(\theta\right)\right) & =\left[\underbrace{\frac{\partial}{\partial\beta'}\left(B\left(\beta\right)^{-1}u_{t}\left(\theta\right)\right)}_{=(D)}\otimes B'\left(\beta\right)^{-1}\Sigma^{-1}e_{x,t}\left(\theta\right)\right]+\cdots\\
 & \quad\cdots+\left[B\left(\beta\right)^{-1}u_{t}\left(\theta\right)\otimes\underbrace{\frac{\partial}{\partial\beta'}\left(B'\left(\beta\right)^{-1}\Sigma^{-1}e_{x,t}\left(\theta\right)\right)}_{=(E)}\right]
\end{align*}

\subparagraph{Term (D).}

We have 
\begin{align*}
\frac{\partial}{\partial\beta'}\left(B\left(\beta\right)^{-1}u_{t}\left(\theta\right)\right) & =\left(u_{t}'\left(\theta\right)\otimes I_{n}\right)\left(\frac{\partial}{\partial\beta'}vec\left(B\left(\beta\right)^{-1}\right)\right)+B\left(\beta\right)^{-1}\left(\frac{\partial u_{t}\left(\theta\right)}{\partial\beta'}\right)\\
 & =-\left(u_{t}'\left(\theta\right)\otimes I_{n}\right)\left(B'\left(\beta\right)^{-1}\otimes B\left(\beta\right)^{-1}\right)H+B\left(\beta\right)^{-1}\left(\frac{\partial u_{t}\left(\theta\right)}{\partial\beta'}\right)\\
 & =-\left(u_{t}'\left(\theta\right)B'\left(\beta\right)^{-1}\otimes B\left(\beta\right)^{-1}\right)H+B\left(\beta\right)^{-1}\left(\frac{\partial u_{t}\left(\theta\right)}{\partial\beta'}\right)
\end{align*}

\subparagraph{Term (E).}

We have {\scriptsize{}
\begin{align*}
\frac{\partial}{\partial\beta'}\left(B'\left(\beta\right)^{-1}\Sigma^{-1}e_{x,t}\left(\theta\right)\right) & =\left(e_{x,t}'\left(\theta\right)\Sigma^{-1}\otimes I_{n}\right)\frac{\partial}{\partial\beta'}vec\left(B'\left(\beta\right)^{-1}\right)+B'\left(\beta\right)^{-1}\Sigma^{-1}\frac{\partial e_{x,t}\left(\theta\right)}{\partial\beta'}\\
 & =\left(e_{x,t}'\left(\theta\right)\Sigma^{-1}\otimes I_{n}\right)K_{nn}\frac{\partial}{\partial\beta'}vec\left(B\left(\beta\right)^{-1}\right)+\cdots\\
 & \quad\cdots+B'\left(\beta\right)^{-1}\Sigma^{-1}\left[-\left(u_{t}'B'\left(\beta\right)^{-1}\left(\theta\right)\otimes e_{xx,t}\left(\theta\right)\Sigma^{-1}B\left(\beta\right)^{-1}\right)H+e_{xx,t}\left(\theta\right)\Sigma^{-1}B\left(\beta\right)^{-1}\left(\frac{\partial u_{t}\left(\theta\right)}{\partial\beta'}\right)\right]\\
 & =-\left(e_{x,t}'\left(\theta\right)\Sigma^{-1}\otimes I_{n}\right)K_{nn}\left(B'\left(\beta\right)^{-1}\otimes B\left(\beta\right)^{-1}\right)H-\left(u_{t}'B'\left(\beta\right)^{-1}\left(\theta\right)\otimes B'\left(\beta\right)^{-1}\Sigma^{-1}e_{xx,t}\left(\theta\right)\Sigma^{-1}B\left(\beta\right)^{-1}\right)H+\cdots\\
 & \quad\cdots+B'\left(\beta\right)^{-1}\Sigma^{-1}e_{xx,t}\left(\theta\right)\Sigma^{-1}B\left(\beta\right)^{-1}\left(\frac{\partial u_{t}\left(\theta\right)}{\partial\beta'}\right)\\
 & =-\left(e_{x,t}'\left(\theta\right)\Sigma^{-1}\otimes I_{n}\right)\left(B\left(\beta\right)^{-1}\otimes B'\left(\beta\right)^{-1}\right)K_{nn}H-\left(u_{t}'B'\left(\beta\right)^{-1}\left(\theta\right)\otimes B'\left(\beta\right)^{-1}\Sigma^{-1}e_{xx,t}\left(\theta\right)\Sigma^{-1}B\left(\beta\right)^{-1}\right)H+\cdots\\
 & \quad\cdots+B'\left(\beta\right)^{-1}\Sigma^{-1}e_{xx,t}\left(\theta\right)\Sigma^{-1}B\left(\beta\right)^{-1}\left(\frac{\partial u_{t}\left(\theta\right)}{\partial\beta'}\right)\\
 & =-\left(e_{x,t}'\left(\theta\right)\Sigma^{-1}B\left(\beta\right)^{-1}\otimes B'\left(\beta\right)^{-1}\right)K_{nn}H-\left(u_{t}'B'\left(\beta\right)^{-1}\left(\theta\right)\otimes B'\left(\beta\right)^{-1}\Sigma^{-1}e_{xx,t}\left(\theta\right)\Sigma^{-1}B\left(\beta\right)^{-1}\right)H+\cdots\\
 & \quad\cdots+B'\left(\beta\right)^{-1}\Sigma^{-1}e_{xx,t}\left(\theta\right)\Sigma^{-1}B\left(\beta\right)^{-1}\left(\frac{\partial u_{t}\left(\theta\right)}{\partial\beta'}\right)
\end{align*}
}{\scriptsize\par}

\paragraph{Final result for term (A).}

We have 
\begin{align*}
(A) & =-\left[\left(u_{t}'\left(\theta\right)B'\left(\beta\right)^{-1}\otimes B\left(\beta\right)^{-1}\right)H\otimes B'\left(\beta\right)^{-1}\Sigma^{-1}e_{x,t}\left(\theta\right)\right]\\
 & \quad+\left[B\left(\beta\right)^{-1}\left(\frac{\partial u_{t}\left(\theta\right)}{\partial\beta'}\right)\otimes B'\left(\beta\right)^{-1}\Sigma^{-1}e_{x,t}\left(\theta\right)\right]\\
 & \quad-\left[B\left(\beta\right)^{-1}u_{t}\left(\theta\right)\otimes\left(e_{x,t}'\left(\theta\right)\Sigma^{-1}B\left(\beta\right)^{-1}\otimes B'\left(\beta\right)^{-1}\right)K_{nn}H\right]\\
 & \quad-\left[B\left(\beta\right)^{-1}u_{t}\left(\theta\right)\otimes\left(u_{t}'B'\left(\beta\right)^{-1}\left(\theta\right)\otimes B'\left(\beta\right)^{-1}\Sigma^{-1}e_{xx,t}\left(\theta\right)\Sigma^{-1}B\left(\beta\right)^{-1}\right)H\right]\\
 & \quad+\left[B\left(\beta\right)^{-1}u_{t}\left(\theta\right)\otimes B'\left(\beta\right)^{-1}\Sigma^{-1}e_{xx,t}\left(\theta\right)\Sigma^{-1}B\left(\beta\right)^{-1}\left(\frac{\partial u_{t}\left(\theta\right)}{\partial\beta'}\right)\right]
\end{align*}

\paragraph{Final result for term (B).}

We have {\scriptsize{}
\begin{align*}
(B) & =\frac{\partial}{\partial\beta'}\left(\frac{\partial u_{t}'\left(\theta\right)}{\partial\beta}B'\left(\beta\right)^{-1}\Sigma^{-1}e_{x,t}\left(\theta\right)\right)\\
 & =\left[\frac{\partial}{\partial\beta'}\left(\frac{\partial u_{t}'\left(\theta\right)}{\partial\beta}\right)\right]B'\left(\beta\right)^{-1}\Sigma^{-1}e_{x,t}\left(\theta\right)+\left(\frac{\partial u_{t}'\left(\theta\right)}{\partial\beta}\right)\left(B'\left(\beta\right)^{-1}\Sigma^{-1}\right)\left(\frac{\partial e_{x,t}\left(\theta\right)}{\partial\beta'}\right)\\
 & =\left[\frac{\partial}{\partial\beta'}\left(\frac{\partial u_{t}'\left(\theta\right)}{\partial\beta}\right)\right]B'\left(\beta\right)^{-1}\Sigma^{-1}e_{x,t}\left(\theta\right)+\cdots\\
 & \quad\cdots+\left(\frac{\partial u_{t}'\left(\theta\right)}{\partial\beta}\right)\left(B'\left(\beta\right)^{-1}\Sigma^{-1}\right)\left[-e_{xx,t}\left(\theta\right)\Sigma^{-1}\left(u_{t}'\left(\theta\right)B'\left(\beta\right)^{-1}\otimes B\left(\beta\right)^{-1}\right)H+e_{xx,t}\left(\theta\right)\Sigma^{-1}B\left(\beta\right)^{-1}\left(\frac{\partial u_{t}\left(\theta\right)}{\partial\beta'}\right)\right]\\
 & =\left[\frac{\partial}{\partial\beta'}\left(\frac{\partial u_{t}'\left(\theta\right)}{\partial\beta}\right)\right]B'\left(\beta\right)^{-1}\Sigma^{-1}e_{x,t}\left(\theta\right)-\cdots\\
 & \quad\cdots-\left(\frac{\partial u_{t}'\left(\theta\right)}{\partial\beta}\right)\left(u_{t}'\left(\theta\right)B'\left(\beta\right)^{-1}\otimes B'\left(\beta\right)^{-1}\Sigma^{-1}e_{xx,t}\left(\theta\right)\Sigma^{-1}B\left(\beta\right)^{-1}\right)H+\left(\frac{\partial u_{t}'\left(\theta\right)}{\partial\beta}\right)\left[B'\left(\beta\right)^{-1}\Sigma^{-1}e_{xx,t}\left(\theta\right)\Sigma^{-1}B\left(\beta\right)^{-1}\right]\left(\frac{\partial u_{t}\left(\theta\right)}{\partial\beta'}\right)
\end{align*}
}{\scriptsize\par}

\paragraph{Final result for term (C).}

We have

\begin{align*}
\frac{\partial}{\partial\beta'}\left(H'vec\left(B'\left(\beta\right)^{-1}\right)\right) & =H'\frac{\partial vec\left(B'\left(\beta\right)^{-1}\right)}{\partial\beta'}\\
 & =H'K_{nn}\frac{\partial vec\left(B\left(\beta\right)^{-1}\right)}{\partial\beta'}\\
 & =-H'K_{nn}\left(B'\left(\beta\right)^{-1}\otimes B\left(\beta\right)^{-1}\right)\frac{\partial vec\left(B\left(\beta\right)\right)}{\partial\beta'}\\
 & =-H'K_{nn}\left(B'\left(\beta\right)^{-1}\otimes B\left(\beta\right)^{-1}\right)H
\end{align*}

\paragraph{Result for $l_{\beta\beta,t}\left(\theta\right)$.}

Collecting all terms, we obtain
\begin{align*}
l_{\beta\beta,t}\left(\theta\right) & =-\left[\left(u_{t}'\left(\theta\right)B'\left(\beta\right)^{-1}\otimes B\left(\beta\right)^{-1}\right)H\otimes B'\left(\beta\right)^{-1}\Sigma^{-1}e_{x,t}\left(\theta\right)\right] & (A_{1})\\
 & \quad+\left[B\left(\beta\right)^{-1}\left(\frac{\partial u_{t}\left(\theta\right)}{\partial\beta'}\right)\otimes B'\left(\beta\right)^{-1}\Sigma^{-1}e_{x,t}\left(\theta\right)\right] & (A_{2})\\
 & \quad-\left[B\left(\beta\right)^{-1}u_{t}\left(\theta\right)\otimes\left(e_{x,t}'\left(\theta\right)\Sigma^{-1}B\left(\beta\right)^{-1}\otimes B'\left(\beta\right)^{-1}\right)K_{nn}H\right] & (A_{3})\\
 & \quad-\left[B\left(\beta\right)^{-1}u_{t}\left(\theta\right)\otimes\left(u_{t}'B'\left(\beta\right)^{-1}\left(\theta\right)\otimes B'\left(\beta\right)^{-1}\Sigma^{-1}e_{xx,t}\left(\theta\right)\Sigma^{-1}B\left(\beta\right)^{-1}\right)H\right] & (A_{4})\\
 & \quad+\left[B\left(\beta\right)^{-1}u_{t}\left(\theta\right)\otimes B'\left(\beta\right)^{-1}\Sigma^{-1}e_{xx,t}\left(\theta\right)\Sigma^{-1}B\left(\beta\right)^{-1}\left(\frac{\partial u_{t}\left(\theta\right)}{\partial\beta'}\right)\right] & (A_{5})\\
 & \quad+\left[\frac{\partial}{\partial\beta'}\left(\frac{\partial u_{t}'\left(\theta\right)}{\partial\beta}\right)\right]B'\left(\beta\right)^{-1}\Sigma^{-1}e_{x,t}\left(\theta\right) & (B_{1})\\
 & \quad-\left(\frac{\partial u_{t}'\left(\theta\right)}{\partial\beta}\right)\left(u_{t}'\left(\theta\right)B'\left(\beta\right)^{-1}\otimes B'\left(\beta\right)^{-1}\Sigma^{-1}e_{xx,t}\left(\theta\right)\Sigma^{-1}B\left(\beta\right)^{-1}\right)H & (B_{2})\\
 & \quad+\left(\frac{\partial u_{t}'\left(\theta\right)}{\partial\beta}\right)\left[B'\left(\beta\right)^{-1}\Sigma^{-1}e_{xx,t}\left(\theta\right)\Sigma^{-1}B\left(\beta\right)^{-1}\right]\left(\frac{\partial u_{t}\left(\theta\right)}{\partial\beta'}\right) & (B_{3})\\
 & \quad-H'K_{nn}\left(B'\left(\beta\right)^{-1}\otimes B\left(\beta\right)^{-1}\right)H & (C)
\end{align*}
Applying formula 9.22(3), i.e. $K_{mp}\left(a_{p\times1}\otimes A_{m\times n}\right)=\left(A_{m\times n}\otimes a_{p\times1}\right)$,
from \citet{luet_mat96}, page 117, to term $\left(A_{1}\right)$,
we obtain{\scriptsize{}
\begin{align*}
\left\{ \underbrace{\left[\left(u_{t}'\left(\theta\right)B'\left(\beta\right)^{-1}\otimes B\left(\beta\right)^{-1}\right)H\right]}_{=\left(n\times n(n-1)\right)}\otimes\underbrace{\left[B'\left(\beta\right)^{-1}\Sigma^{-1}e_{x,t}\left(\theta\right)\right]}_{=\left(n\times1\right)}\right\}  & =K_{n,n}\left\{ \left[B'\left(\beta\right)^{-1}\Sigma^{-1}e_{x,t}\left(\theta\right)\right]\otimes\left[\left(u_{t}'\left(\theta\right)B'\left(\beta\right)^{-1}\otimes B\left(\beta\right)^{-1}\right)H\right]\right\} \\
 & =K_{n,n}\left\{ \left[B'\left(\beta\right)^{-1}\Sigma^{-1}e_{x,t}\left(\theta\right)\right]\otimes\left[u_{t}'\left(\theta\right)B'\left(\beta\right)^{-1}\right]\otimes B\left(\beta\right)^{-1}\right\} H\\
 & =K_{n,n}\left\{ \left[B'\left(\beta\right)^{-1}\Sigma^{-1}e_{x,t}\left(\theta\right)u_{t}'\left(\theta\right)B'\left(\beta\right)^{-1}\right]\otimes B\left(\beta\right)^{-1}\right\} H
\end{align*}
}{\scriptsize\par}

Rearranging all terms results in
\begin{align*}
l_{\beta\beta,t}\left(\theta\right) & =H'K_{n,n}\left\{ \left[B'\left(\beta\right)^{-1}\Sigma^{-1}e_{x,t}\left(\theta\right)u_{t}'\left(\theta\right)B'\left(\beta\right)^{-1}\right]\otimes B\left(\beta\right)^{-1}\right\} H & \left(A_{1}\right)\\
 & \quad+H'\left\{ \left[B\left(\beta\right)^{-1}u_{t}\left(\theta\right)e_{x,t}'\left(\theta\right)\Sigma^{-1}B\left(\beta\right)^{-1}\right]\otimes B'\left(\beta\right)^{-1}\right\} K_{nn}H & \left(A_{3}\right)\\
\\
 & \quad-H'\left\{ B\left(\beta\right)^{-1}u_{t}\left(\theta\right)\otimes\left[B'\left(\beta\right)^{-1}\Sigma^{-1}e_{xx,t}\left(\theta\right)\Sigma^{-1}B\left(\beta\right)^{-1}\left(\frac{\partial u_{t}\left(\theta\right)}{\partial\beta'}\right)\right]\right\}  & \left(A_{5}\right)\\
 & \quad-H'\left\{ u_{t}'\left(\theta\right)B'\left(\beta\right)^{-1}\otimes\left[\left(\frac{\partial u_{t}'\left(\theta\right)}{\partial\beta}\right)B'\left(\beta\right)^{-1}\Sigma^{-1}e_{xx,t}\left(\theta\right)\Sigma^{-1}B\left(\beta\right)^{-1}\right]\right\} H & \left(B_{2}\right)\\
\\
 & \quad-H'\left[B\left(\beta\right)^{-1}u_{t}\left(\theta\right)u_{t}'B'\left(\beta\right)^{-1}\left(\theta\right)\otimes B'\left(\beta\right)^{-1}\Sigma^{-1}e_{xx,t}\left(\theta\right)\Sigma^{-1}B\left(\beta\right)^{-1}\right]H & \left(A_{4}\right)\\
 & \quad+\left(\frac{\partial u_{t}'\left(\theta\right)}{\partial\beta}\right)\left[B'\left(\beta\right)^{-1}\Sigma^{-1}e_{xx,t}\left(\theta\right)\Sigma^{-1}B\left(\beta\right)^{-1}\right]\left(\frac{\partial u_{t}\left(\theta\right)}{\partial\beta'}\right) & \left(B_{3}\right)\\
\\
 & \quad+H\left[B\left(\beta\right)^{-1}\left(\frac{\partial u_{t}\left(\theta\right)}{\partial\beta'}\right)\otimes B'\left(\beta\right)^{-1}\Sigma^{-1}e_{x,t}\left(\theta\right)\right] & \left(A_{2}\right)\\
 & \quad+\left[\frac{\partial}{\partial\beta'}\left(\frac{\partial u_{t}'\left(\theta\right)}{\partial\beta}\right)\right]B'\left(\beta\right)^{-1}\Sigma^{-1}e_{x,t}\left(\theta\right) & \left(B_{1}\right)\\
\\
 & \quad-H'K_{nn}\left(B'\left(\beta\right)^{-1}\otimes B\left(\beta\right)^{-1}\right)H & \left(C\right)
\end{align*}

\subsubsection{Term $\left(\beta,\pi_{3}\right)$}

\paragraph{Intermediate result for $l_{\beta\pi_{3},t}\left(\theta\right)$.}

Taking the derivative with respect to $\pi_{3}'$ of \eqref{eq:mds_lik_beta},
we obtain
\begin{align*}
\frac{\partial}{\partial\pi_{3}'}\left(\frac{\partial l_{t}\left(\theta\right)}{\partial\beta}\right) & =\left[\frac{\partial}{\partial\pi_{3}'}\left(\frac{\partial u_{t}'\left(\theta\right)}{\partial\beta}\right)\right]B'\left(\beta\right)^{-1}\Sigma^{-1}e_{x,t}\left(\theta\right)+\left(\frac{\partial u_{t}'\left(\theta\right)}{\partial\beta}\right)B'\left(\beta\right)^{-1}\Sigma^{-1}\left(\frac{\partial e_{x,t}\left(\theta\right)}{\partial\pi_{3}'}\right)\\
 & \quad-H\left(B\left(\beta\right)^{-1}\frac{\partial u_{t}\left(\theta\right)}{\partial\pi_{3}'}\otimes B'\left(\beta\right)^{-1}\Sigma^{-1}e_{x,t}\left(\theta\right)\right)-H\left(B\left(\beta\right)^{-1}u_{t}\left(\theta\right)\otimes B'\left(\beta\right)^{-1}\Sigma^{-1}\left(\frac{\partial e_{x,t}\left(\theta\right)}{\partial\pi_{3}'}\right)\right)
\end{align*}

\paragraph{Expressions in the intermediate result for $l_{\beta\pi_{3},t}\left(\theta\right)$.}

As usual, we obtain 
\begin{align*}
\frac{\partial e_{x,t}\left(\theta\right)}{\partial\pi_{3}'} & =\sum_{i=1}^{n}e_{i,xx,t}\left(\theta\right)\sigma_{i}^{-1}\iota_{i}'B'\left(\beta\right)^{-1}\frac{\partial u_{t}\left(\theta\right)}{\partial\pi_{3}'}\\
 & =e_{xx,t}\left(\theta\right)\Sigma^{-1}B'\left(\beta\right)^{-1}\frac{\partial u_{t}\left(\theta\right)}{\partial\pi_{3}'},
\end{align*}
where $\frac{\partial u_{t}'\left(\theta\right)}{\partial\pi_{3}}=-\left[\begin{pmatrix}w_{0}\left(\theta\right) & \cdots & w_{T-1}\left(\theta\right)\end{pmatrix}\otimes I_{n}\right]\mathfrak{b_{t}}=w_{b,t-1}\left(\theta\right)$
and $\frac{\partial u_{t}'\left(\theta\right)}{\partial\beta}=-H'\sum_{i=1}^{q}\left(B\left(\beta\right)^{-1}u_{t-i}\left(\theta\right)\otimes b_{i}'\right)$,
such that
\begin{align*}
\frac{\partial}{\partial\pi_{3}'}\left(\frac{\partial u_{t}'\left(\theta\right)}{\partial\beta}\right) & =-\frac{\partial}{\partial\pi_{3}'}\left(H'\sum_{i=1}^{q}\left(B\left(\beta\right)^{-1}u_{t-i}\left(\theta\right)\otimes b_{i}'\right)\right)\\
 & =-\left(H'\sum_{i=1}^{q}\left(B\left(\beta\right)^{-1}\frac{\partial u_{t-i}\left(\theta\right)}{\partial\pi_{3}'}\otimes b_{i}'\right)\right)\\
 & =H'\sum_{i=1}^{q}\left(\left\{ B\left(\beta\right)^{-1}w_{b,t-i}'\left(\theta\right)\right\} \otimes b_{i}'\right).
\end{align*}

\paragraph{Result.}

We have that {\scriptsize{}
\begin{align*}
l_{\beta\pi_{3},t}\left(\theta\right) & =\left[\frac{\partial}{\partial\pi_{3}'}\left(\frac{\partial u_{t}'\left(\theta\right)}{\partial\beta}\right)\right]B'\left(\beta\right)^{-1}\Sigma^{-1}e_{x,t}\left(\theta\right)+\left(H'\sum_{i=1}^{q}\left(B\left(\beta\right)^{-1}u_{t-i}\left(\theta\right)\otimes b_{i}'\right)\right)B'\left(\beta\right)^{-1}\Sigma^{-1}\left(e_{xx,t}\left(\theta\right)\Sigma^{-1}B'\left(\beta\right)^{-1}\frac{\partial u_{t}\left(\theta\right)}{\partial\pi_{3}'}\right)\\
 & \quad-H\left(\left\{ B\left(\beta\right)^{-1}w_{b,t-i}'\left(\theta\right)\right\} \otimes B'\left(\beta\right)^{-1}\Sigma^{-1}e_{x,t}\left(\theta\right)\right)-H\left(B\left(\beta\right)^{-1}u_{t}\left(\theta\right)\otimes B'\left(\beta\right)^{-1}\Sigma^{-1}\left(e_{xx,t}\left(\theta\right)\Sigma^{-1}B'\left(\beta\right)^{-1}\frac{\partial u_{t}\left(\theta\right)}{\partial\pi_{3}'}\right)\right)\\
 & =H'\sum_{i=1}^{q}\left(\left\{ B\left(\beta\right)^{-1}w_{b,t-i}'\left(\theta\right)\right\} \otimes b_{i}'\right)B'\left(\beta\right)^{-1}\Sigma^{-1}e_{x,t}\left(\theta\right)\\
 & \quad+\left(H'\sum_{i=1}^{q}\left(B\left(\beta\right)^{-1}u_{t-i}\left(\theta\right)\otimes b_{i}'\right)\right)B'\left(\beta\right)^{-1}\Sigma^{-1}e_{xx,t}\left(\theta\right)\Sigma^{-1}B'\left(\beta\right)^{-1}w_{b,t-i}'\left(\theta\right)\\
 & \quad-H\left(\left\{ B\left(\beta\right)^{-1}w_{b,t-i}'\left(\theta\right)\right\} \otimes B'\left(\beta\right)^{-1}\Sigma^{-1}e_{x,t}\left(\theta\right)\right)\\
 & \quad-H\left(B\left(\beta\right)^{-1}u_{t}\left(\theta\right)\otimes B'\left(\beta\right)^{-1}\Sigma^{-1}e_{xx,t}\left(\theta\right)\Sigma^{-1}B'\left(\beta\right)^{-1}w_{b,t-i}'\left(\theta\right)\right)
\end{align*}
}{\scriptsize\par}

\subsubsection{Term $\left(\beta,\pi_{2}\right)$}

Taking the derivative of \eqref{eq:mds_lik_beta} with respect to
$\pi_{2}'$, we obtain
\begin{align*}
\frac{\partial}{\partial\pi_{2}'}\left(\frac{\partial l_{t}\left(\theta\right)}{\partial\beta}\right) & =\left[\frac{\partial}{\partial\pi_{2}'}\left(\frac{\partial u_{t}'\left(\theta\right)}{\partial\beta}\right)\right]B'\left(\beta\right)^{-1}\Sigma^{-1}e_{x,t}\left(\theta\right)+\left(\frac{\partial u_{t}'\left(\theta\right)}{\partial\beta}\right)B'\left(\beta\right)^{-1}\Sigma^{-1}\left(\frac{\partial e_{x,t}\left(\theta\right)}{\partial\pi_{2}'}\right)\\
 & \quad-H\left(B\left(\beta\right)^{-1}\frac{\partial u_{t}\left(\theta\right)}{\partial\pi_{2}'}\otimes B'\left(\beta\right)^{-1}\Sigma^{-1}e_{x,t}\left(\theta\right)\right)-H\left(B\left(\beta\right)^{-1}u_{t}\left(\theta\right)\otimes B'\left(\beta\right)^{-1}\Sigma^{-1}\left(\frac{\partial e_{x,t}\left(\theta\right)}{\partial\pi_{2}'}\right)\right)
\end{align*}

\paragraph{Expressions in the intermediate result for $l_{\beta\pi_{2},t}\left(\theta\right)$.}

We need the expressions $\frac{\partial e_{x,t}\left(\theta\right)}{\partial\pi_{2}'}=e_{xx,t}\left(\theta\right)\Sigma^{-1}B'\left(\beta\right)^{-1}\frac{\partial u_{t}\left(\theta\right)}{\partial\pi_{2}'}$
and $\frac{\partial u_{t}'\left(\theta\right)}{\partial\pi_{2}}=-\left[\begin{pmatrix}x_{0} & \cdots & x_{T-1}\end{pmatrix}\otimes I_{n}\right]\mathfrak{b_{t}}=-x_{b,t-1}\left(\theta\right)$
as well as $\frac{\partial u_{t}'\left(\theta\right)}{\partial\beta}=H'\sum_{i=1}^{q}\left(B\left(\beta\right)^{-1}u_{t-i}\left(\theta\right)\otimes b_{i}'\right)$.

\paragraph{Result.}

We have that {\scriptsize{}
\begin{align*}
l_{\beta\pi_{2},t}\left(\theta\right) & =\underbrace{\left[\frac{\partial}{\partial\pi_{2}'}\left(\frac{\partial u_{t}'\left(\theta\right)}{\partial\beta}\right)\right]}_{=0}B'\left(\beta\right)^{-1}\Sigma^{-1}e_{x,t}\left(\theta\right)+\left(H'\sum_{i=1}^{q}\left(B\left(\beta\right)^{-1}u_{t-i}\left(\theta\right)\otimes b_{i}'\right)\right)B'\left(\beta\right)^{-1}\Sigma^{-1}\left(e_{xx,t}\left(\theta\right)\Sigma^{-1}B'\left(\beta\right)^{-1}\frac{\partial u_{t}\left(\theta\right)}{\partial\pi_{2}'}\right)\\
 & \qquad-H\left(\left\{ B\left(\beta\right)^{-1}x_{b,t-1}'\left(\theta\right)\right\} \otimes B'\left(\beta\right)^{-1}\Sigma^{-1}e_{x,t}\left(\theta\right)\right)-H\left(B\left(\beta\right)^{-1}u_{t}\left(\theta\right)\otimes B'\left(\beta\right)^{-1}\Sigma^{-1}\left(e_{xx,t}\left(\theta\right)\Sigma^{-1}B'\left(\beta\right)^{-1}\frac{\partial u_{t}\left(\theta\right)}{\partial\pi_{2}'}\right)\right)\\
 & =\left(H'\sum_{i=1}^{q}\left(B\left(\beta\right)^{-1}u_{t-i}\left(\theta\right)\otimes b_{i}'\right)\right)B'\left(\beta\right)^{-1}\Sigma^{-1}e_{xx,t}\left(\theta\right)\Sigma^{-1}B'\left(\beta\right)^{-1}x_{b,t-1}'\left(\theta\right)\\
 & \qquad-H\left(\left\{ B\left(\beta\right)^{-1}x_{b,t-1}'\left(\theta\right)\right\} \otimes B'\left(\beta\right)^{-1}\Sigma^{-1}e_{x,t}\left(\theta\right)\right)-H\left(B\left(\beta\right)^{-1}u_{t}\left(\theta\right)\otimes B'\left(\beta\right)^{-1}\Sigma^{-1}e_{xx,t}\left(\theta\right)\Sigma^{-1}B'\left(\beta\right)^{-1}x_{b,t-1}'\left(\theta\right)\right).
\end{align*}
}{\scriptsize\par}

\subsection{Elements of the Hessian involving $\sigma$}

\subsubsection{Diagonal term $\left(\sigma,\sigma\right)$}

Remember that $\frac{\partial}{\partial\sigma_{i}}l_{t}\left(\theta\right)=-e_{i,x,t}\left(\theta\right)\sigma_{i}^{-2}\iota_{i}^{'}B\left(\beta\right)^{-1}u_{t}\left(\theta\right)-\sigma_{i}^{-1}$
such that the second derivative is 
\begin{align*}
\frac{\partial}{\partial\sigma_{i}}\left(\frac{\partial l_{t}\left(\theta\right)}{\partial\sigma_{i}}\right) & =-\left(\frac{\partial e_{i,x,t}\left(\theta\right)}{\partial\sigma_{i}}\right)\sigma_{i}^{-2}\iota_{i}^{'}B\left(\beta\right)^{-1}u_{t}\left(\theta\right)-e_{i,x,t}\left(\theta\right)\left(\frac{\partial\sigma_{i}^{-2}}{\partial\sigma_{i}}\right)\iota_{i}^{'}B\left(\beta\right)^{-1}u_{t}\left(\theta\right)-\left(\frac{\partial\sigma_{i}^{-1}}{\partial\sigma_{i}}\right)\\
 & =-e_{i,xx,t}\left(\theta\right)\left(-\sigma_{i}^{-1}\right)\sigma_{i}^{-2}\varepsilon_{i,t}\left(\theta\right)-e_{i,x,t}\left(\theta\right)\left(-2\sigma_{i}^{-3}\right)\varepsilon_{i,t}\left(\theta\right)-\left(-\sigma_{i}^{-2}\right)\\
 & =e_{i,xx,t}\left(\theta\right)\sigma_{i}^{-4}\varepsilon_{i,t}\left(\theta\right)+2e_{i,x,t}\left(\theta\right)\sigma_{i}^{-3}\varepsilon_{i,t}\left(\theta\right)+\sigma_{i}^{-2}.
\end{align*}

\paragraph{Result.}

We obtain that 
\[
l_{\sigma\sigma,t}\left(\theta\right)=\Sigma^{-4}e_{xx,t}\left(\theta\right)\mathcal{E}_{t}^{2}\left(\theta\right)+2\Sigma^{-3}E_{x,t}\left(\theta\right)\mathcal{E}_{t}\left(\theta\right)+\Sigma^{-2}
\]
where $E_{x,t}\left(\theta\right)=diag\left(e_{1,x,t}\left(\theta\right),\ldots,e_{n,x,t}\left(\theta\right)\right)$
and $\mathcal{E}_{t}\left(\theta\right)=diag\left(\varepsilon_{1,t}\left(\theta\right),\ldots,\varepsilon_{n,t}\left(\theta\right)\right)$.

\subsubsection{Diagonal term $\left(\sigma,\beta\right)$}

\paragraph{Intermediate result for $l_{\sigma\beta,t}\left(\theta\right)$.}

We have that 
\begin{align*}
\frac{\partial}{\partial\beta'}\left(\frac{\partial l_{t}\left(\theta\right)}{\partial\sigma_{i}}\right) & =\frac{\partial}{\partial\beta'}\left(-e_{i,x,t}\left(\theta\right)\sigma_{i}^{-2}\iota_{i}^{'}B\left(\beta\right)^{-1}u_{t}\left(\theta\right)-\sigma_{i}^{-1}\right)\\
 & =-\left(\frac{\partial e_{i,x,t}\left(\theta\right)}{\partial\beta'}\right)\sigma_{i}^{-2}\varepsilon_{i,t}-e_{i,x,t}\left(\theta\right)\sigma_{i}^{-2}\left(u_{t}'\left(\theta\right)\otimes\iota_{i}^{'}\right)\left(\frac{\partial B\left(\beta\right)^{-1}}{\partial\beta'}\right)-e_{i,x,t}\left(\theta\right)\sigma_{i}^{-2}\iota_{i}^{'}B\left(\beta\right)^{-1}\left(\frac{\partial u_{t}\left(\theta\right)}{\partial\beta'}\right)\\
 & =-\left\{ e_{i,xx,t}\left(\theta\right)\sigma_{i}^{-1}\left[\left(u_{t}'\left(\theta\right)\otimes\iota_{i}^{'}\right)\left(\frac{\partial B\left(\beta\right)^{-1}}{\partial\beta'}\right)+\iota_{i}'B\left(\beta\right)^{-1}\frac{\partial u_{t}\left(\theta\right)}{\partial\beta'}\right]\right\} \sigma_{i}^{-2}\varepsilon_{i,t}\\
 & \qquad-e_{i,x,t}\left(\theta\right)\sigma_{i}^{-2}\left(u_{t}'\left(\theta\right)\otimes\iota_{i}^{'}\right)\left(\frac{\partial B\left(\beta\right)^{-1}}{\partial\beta'}\right)-e_{i,x,t}\left(\theta\right)\sigma_{i}^{-2}\iota_{i}^{'}B\left(\beta\right)^{-1}\left(\frac{\partial u_{t}\left(\theta\right)}{\partial\beta'}\right)\\
 & =-\left\{ e_{i,xx,t}\left(\theta\right)\sigma_{i}^{-1}\left[-\left(u_{t}'\left(\theta\right)\otimes\iota_{i}^{'}\right)\left(B'\left(\beta\right)^{-1}\otimes B\left(\beta\right)^{-1}\right)H+\iota_{i}'B\left(\beta\right)^{-1}\frac{\partial u_{t}\left(\theta\right)}{\partial\beta'}\right]\right\} \sigma_{i}^{-2}\varepsilon_{i,t}\\
 & \qquad+e_{i,x,t}\left(\theta\right)\sigma_{i}^{-2}\left(u_{t}'\left(\theta\right)\otimes\iota_{i}^{'}\right)\left(B'\left(\beta\right)^{-1}\otimes B\left(\beta\right)^{-1}\right)H-e_{i,x,t}\left(\theta\right)\sigma_{i}^{-2}\iota_{i}^{'}B\left(\beta\right)^{-1}\left(\frac{\partial u_{t}\left(\theta\right)}{\partial\beta'}\right)\\
 & =-\left\{ e_{i,xx,t}\left(\theta\right)\sigma_{i}^{-1}\left[-\left(u_{t}'\left(\theta\right)B'\left(\beta\right)^{-1}\otimes\iota_{i}^{'}B\left(\beta\right)^{-1}\right)H+\iota_{i}'B\left(\beta\right)^{-1}\frac{\partial u_{t}\left(\theta\right)}{\partial\beta'}\right]\right\} \sigma_{i}^{-2}\varepsilon_{i,t}\\
 & \qquad+e_{i,x,t}\left(\theta\right)\sigma_{i}^{-2}\left(u_{t}'\left(\theta\right)B'\left(\beta\right)^{-1}\otimes\iota_{i}^{'}B\left(\beta\right)^{-1}\right)H-e_{i,x,t}\left(\theta\right)\sigma_{i}^{-2}\iota_{i}^{'}B\left(\beta\right)^{-1}\left(\frac{\partial u_{t}\left(\theta\right)}{\partial\beta'}\right)
\end{align*}
where we have used 
\[
\frac{\partial e_{i,x,t}\left(\theta\right)}{\partial\beta'}=e_{i,xx,t}\left(\theta\right)\sigma_{i}^{-1}\left(-\left(u_{t}'\left(\theta\right)\otimes\iota_{i}^{'}\right)\left(\frac{\partial B\left(\beta\right)^{-1}}{\partial\beta'}\right)+\iota_{i}'B\left(\beta\right)^{-1}\frac{\partial u_{t}\left(\theta\right)}{\partial\beta'}\right)
\]
 and $\frac{\partial B\left(\beta\right)^{-1}}{\partial\beta'}=-\left(B'\left(\beta\right)^{-1}\otimes B\left(\beta\right)^{-1}\right)H$.

\paragraph{Expressions in the intermediate result for $l_{\sigma\beta,t}\left(\theta\right)$.}

Taking finally $\frac{\partial u_{t}\left(\theta\right)}{\partial\beta'}=\sum_{i=1}^{q}\left(u_{t-i}'B'\left(\beta\right)^{-1}\left(\theta\right)\otimes b_{i}\right)H$
into account, we obtain that {\scriptsize{}
\begin{align*}
\frac{\partial}{\partial\beta'}\left(\frac{\partial l_{t}\left(\theta\right)}{\partial\sigma}\right) & =-\left\{ \Sigma^{-3}e_{xx,t}\left(\theta\right)\mathcal{E}_{t}\left(\theta\right)\left[-\left(u_{t}'\left(\theta\right)B'\left(\beta\right)^{-1}\otimes B\left(\beta\right)^{-1}\right)H+B\left(\beta\right)^{-1}\frac{\partial u_{t}\left(\theta\right)}{\partial\beta'}\right]\right\} \sigma_{i}^{-2}\varepsilon_{i,t}\\
 & \qquad+\Sigma^{-2}E_{x,t}\left(\theta\right)\left(u_{t}'\left(\theta\right)B'\left(\beta\right)^{-1}\otimes B\left(\beta\right)^{-1}\right)H-\Sigma^{-2}E_{x,t}\left(\theta\right)B\left(\beta\right)^{-1}\left(\frac{\partial u_{t}\left(\theta\right)}{\partial\beta'}\right)\\
 & =\Sigma^{-2}\left(\Sigma^{-1}e_{xx,t}\left(\theta\right)\mathcal{E}_{t}\left(\theta\right)+E_{x,t}\left(\theta\right)\right)\left[\left(u_{t}'\left(\theta\right)B'\left(\beta\right)^{-1}\otimes B\left(\beta\right)^{-1}\right)H-B\left(\beta\right)^{-1}\left(\frac{\partial u_{t}\left(\theta\right)}{\partial\beta'}\right)\right]\\
 & =\Sigma^{-2}\left(\Sigma^{-1}e_{xx,t}\left(\theta\right)\mathcal{E}_{t}\left(\theta\right)+E_{x,t}\left(\theta\right)\right)\left[\left(u_{t}'\left(\theta\right)B'\left(\beta\right)^{-1}\otimes B\left(\beta\right)^{-1}\right)H-B\left(\beta\right)^{-1}\left(\sum_{i=1}^{q}\left(u_{t-i}'B'\left(\beta\right)^{-1}\left(\theta\right)\otimes b_{i}\right)H\right)\right]\\
 & =\Sigma^{-2}\left(\Sigma^{-1}e_{xx,t}\left(\theta\right)\mathcal{E}_{t}\left(\theta\right)+E_{x,t}\left(\theta\right)\right)\left[\left(u_{t}'\left(\theta\right)B'\left(\beta\right)^{-1}\otimes B\left(\beta\right)^{-1}\right)-\left(\sum_{i=1}^{q}\left(u_{t-i}'B'\left(\beta\right)^{-1}\left(\theta\right)\otimes B\left(\beta\right)^{-1}b_{i}\right)\right)\right]H
\end{align*}
}{\scriptsize\par}

\subsubsection{Diagonal term \textrm{\textmd{$\left(\sigma,\pi_{3}\right)$}}}

\paragraph{Intermediate result for $l_{\sigma\pi_{3},t}\left(\theta\right)$.}

We have that 
\begin{align*}
\frac{\partial}{\partial\pi_{3}'}\left(\frac{\partial l_{t}\left(\theta\right)}{\partial\sigma_{i}}\right) & =\frac{\partial}{\partial\pi_{3}'}\left(-\sigma_{i}^{-2}e_{i,x,t}\left(\theta\right)\iota_{i}^{'}B\left(\beta\right)^{-1}u_{t}\left(\theta\right)-\sigma_{i}^{-1}\right)\\
 & =-\sigma_{i}^{-2}\varepsilon_{i,t}\left(\theta\right)\left(\frac{\partial e_{i,x,t}\left(\theta\right)}{\partial\pi_{3}'}\right)-\sigma_{i}^{-2}e_{i,x,t}\left(\theta\right)\iota_{i}^{'}B\left(\beta\right)^{-1}\left(\frac{\partial u_{t}\left(\theta\right)}{\partial\pi_{3}'}\right)\\
 & =-\sigma_{i}^{-2}\varepsilon_{i,t}\left(\theta\right)\left(e_{i,xx,t}\left(\theta\right)\sigma_{i}^{-1}\iota_{i}'B\left(\beta\right)^{-1}\left(\frac{\partial u_{t}\left(\theta\right)}{\partial\pi_{3}'}\right)\right)-\sigma_{i}^{-2}e_{i,x,t}\left(\theta\right)\iota_{i}^{'}B\left(\beta\right)^{-1}\left(\frac{\partial u_{t}\left(\theta\right)}{\partial\pi_{3}'}\right)\\
 & =-\sigma_{i}^{-2}\left(\sigma_{i}^{-1}e_{i,xx,t}\left(\theta\right)\varepsilon_{i,t}\left(\theta\right)+e_{i,x,t}\left(\theta\right)\right)\iota_{i}'B\left(\beta\right)^{-1}\left(\frac{\partial u_{t}\left(\theta\right)}{\partial\pi_{3}'}\right).
\end{align*}

\paragraph{Multivariate result for $l_{\sigma\pi_{3},t}\left(\theta\right)$.}

Taking finally $\frac{\partial u_{t}'\left(\theta\right)}{\partial\pi_{3}}=-\left[\begin{pmatrix}u_{0}\left(\theta\right) & \cdots & u_{T-1}\left(\theta\right)\\
\vdots &  & \vdots\\
u_{1-q}\left(\theta\right) & \cdots & u_{T-q}\left(\theta\right)
\end{pmatrix}\otimes I_{n}\right]\mathfrak{b_{t}}$ into account, we obtain that 
\begin{align*}
\frac{\partial}{\partial\pi_{3}'}\left(\frac{\partial l_{t}\left(\theta\right)}{\partial\sigma}\right) & =-\Sigma^{-2}\left(\Sigma^{-1}e_{xx,t}\left(\theta\right)\mathcal{E}_{t}\left(\theta\right)+E_{x,t}\left(\theta\right)\right)B\left(\beta\right)^{-1}\left(\frac{\partial u_{t}\left(\theta\right)}{\partial\pi_{3}'}\right)\\
 & =\Sigma^{-2}\left(\Sigma^{-1}e_{xx,t}\left(\theta\right)\mathcal{E}_{t}\left(\theta\right)+E_{x,t}\left(\theta\right)\right)B\left(\beta\right)^{-1}w_{b,t-1}'\left(\theta\right)
\end{align*}

\subsubsection{Diagonal term $\left(\sigma,\pi_{2}\right)$}

\paragraph{Multivariate result for $l_{\sigma\pi_{2},t}\left(\theta\right)$.}

Similarly to the case $\left(\sigma,\pi_{3}\right)$, but taking $\frac{\partial u_{t}'\left(\theta\right)}{\partial\pi_{2}}=-\left[\begin{pmatrix}x_{0} & \cdots & x_{T-1}\end{pmatrix}\otimes I_{n}\right]\mathfrak{b}_{t}=-x_{b,t-1}\left(\theta\right)$
into account, we obtain that 
\begin{align*}
\frac{\partial}{\partial\pi_{2}'}\left(\frac{\partial l_{t}\left(\theta\right)}{\partial\sigma}\right) & =-\Sigma^{-2}\left(\Sigma^{-1}e_{xx,t}\left(\theta\right)\mathcal{E}_{t}\left(\theta\right)+E_{x,t}\left(\theta\right)\right)B\left(\beta\right)^{-1}\left(\frac{\partial u_{t}\left(\theta\right)}{\partial\pi_{2}'}\right)\\
 & =\Sigma^{-2}\left(\Sigma^{-1}e_{xx,t}\left(\theta\right)\mathcal{E}_{t}\left(\theta\right)+E_{x,t}\left(\theta\right)\right)B\left(\beta\right)^{-1}x_{b,t-1}'\left(\theta\right)
\end{align*}

\subsection{Elements of the Hessian involving $\lambda$}

\subsubsection{Diagonal term $\left(\lambda,\lambda\right)$}

We have $l_{\lambda\lambda,t}\left(\theta\right)=e_{\lambda\lambda,t}\left(\theta\right)$

\subsubsection{Diagonal term $\left(\lambda,\sigma\right)$}

Since $\frac{\partial}{\partial\sigma_{i}}l_{t}\left(\theta\right)=-e_{i,x,t}\left(\theta\right)\sigma_{i}^{-2}\iota_{i}^{'}B\left(\beta\right)^{-1}u_{t}\left(\theta\right)-\sigma_{i}^{-1}$,
we have $\frac{\partial}{\partial\lambda_{i}'}\left(\frac{\partial l_{t}\left(\theta\right)}{\partial\sigma_{i}}\right)=-\sigma_{i}^{-2}\varepsilon_{i,t}\left(\theta\right)e_{i,x\lambda_{i},t}\left(\theta\right)$
such that
\[
l_{\sigma\lambda,t}\left(\theta\right)=-\Sigma^{2}\mathcal{E}_{t}\left(\theta\right)e_{x\lambda,t}\left(\theta\right).
\]

\subsubsection{Diagonal term $\left(\lambda,\beta\right)$}

Directly, we obtain{\scriptsize{}
\begin{align*}
\frac{\partial}{\partial\lambda'}\left(\frac{\partial l_{t}\left(\theta\right)}{\partial\beta}\right) & =\frac{\partial}{\partial\lambda'}\left[-H'\sum_{i=1}^{q}\left(I_{n}\otimes b_{i}'B'\left(\beta\right)^{-1}\Sigma^{-1}\right)\left(\varepsilon_{t-i}\left(\theta\right)\otimes e_{x,t}\left(\theta\right)\right)-H'\left(I_{n}\otimes B'\left(\beta\right)^{-1}\Sigma^{-1}\right)\left(\varepsilon_{t}\left(\theta\right)\otimes e_{x,t}\left(\theta\right)\right)-H'vec\left(B'\left(\beta\right)^{-1}\right)\right]\\
 & =-H'\sum_{i=1}^{q}\left(I_{n}\otimes b_{i}'B'\left(\beta\right)^{-1}\Sigma^{-1}\right)\left(\varepsilon_{t-i}\left(\theta\right)\otimes e_{x\lambda,t}\left(\theta\right)\right)-H'\left(I_{n}\otimes B'\left(\beta\right)^{-1}\Sigma^{-1}\right)\left(\varepsilon_{t}\left(\theta\right)\otimes e_{x\lambda,t}\left(\theta\right)\right)
\end{align*}
}{\scriptsize\par}

\subsubsection{Diagonal term $\left(\lambda,\pi_{3}\right)$}

Directly, we obtain
\[
\frac{\partial}{\partial\lambda'}\left(\frac{\partial l_{t}\left(\theta\right)}{\partial\pi_{3}}\right)=-w_{b,t-1}\left(\theta\right)\Sigma^{-1}e_{x\lambda,t}\left(\theta\right)
\]

\subsubsection{Diagonal term $\left(\lambda,\pi_{2}\right)$}

Directly, we obtain
\[
\frac{\partial}{\partial\lambda'}\left(\frac{\partial l_{t}\left(\theta\right)}{\partial\pi_{2}}\right)=-x_{b,t-1}\left(\theta\right)B'\left(\beta\right)^{-1}\Sigma^{-1}e_{x\lambda,t}\left(\theta\right)
\]

\section{\label{sec:hessian_ulln}Verifying Uniform Convergence of Hessian}

\subsection{Expression for Hessian}

Here, we summarize the expressions derived in \ref{sec:hessian_expr}
in order to subsequently verify the respective terms satisfy a ULLN,
i.e. we need to that the $\mathbb{E}\left(\sup_{\theta\in\Theta_{0}}\left\Vert \frac{\partial^{2}\log\left(f\left(x,\theta\right)\right)}{\partial\theta\partial\theta'}\right\Vert \right)<\infty$
holds.

\begin{align*}
l_{\pi_{2}\pi_{2},t}\left(\theta\right) & =x_{b,t-1}\left(\theta\right)B'\left(\beta\right)^{-1}\Sigma^{-1}e_{xx,t}\left(\theta\right)\Sigma^{-1}B(\beta)^{-1}x_{b,t-1}'\left(\theta\right).\\
l_{\pi_{3}\pi_{3},t}\left(\theta\right) & =\left(e_{x,t}'\left(\theta\right)\Sigma^{-1}B'\left(\beta\right)^{-1}\otimes I_{n^{2}p}\right)\left(\mathcal{C}^{-1}\right)_{\left[t,\bullet\right]}\left(\mathcal{D}+\mathcal{E}\right)+w_{b,t-1}\left(\theta\right)B'\left(\beta\right)^{-1}\Sigma^{-1}e_{xx,t}\left(\theta\right)\Sigma^{-1}B\left(\beta\right)^{-1}w_{b,t-1}'\left(\theta\right).\\
l_{\pi_{2}\pi_{3},t}\left(\theta\right) & =x_{b,t-1}\left(\theta\right)B'\left(\beta\right)^{-1}\Sigma^{-1}e_{xx,t}\left(\theta\right)\Sigma^{-1}B\left(\beta\right)^{-1}w_{b,t-1}'\left(\theta\right).
\end{align*}

\begin{align*}
l_{\beta\beta,t}\left(\theta\right) & =H'K_{n,n}\left\{ \left[B'\left(\beta\right)^{-1}\Sigma^{-1}e_{x,t}\left(\theta\right)u_{t}'\left(\theta\right)B'\left(\beta\right)^{-1}\right]\otimes B\left(\beta\right)^{-1}\right\} H & \left(A_{1}\right)\\
 & \quad+H'\left\{ \left[B\left(\beta\right)^{-1}u_{t}\left(\theta\right)e_{x,t}'\left(\theta\right)\Sigma^{-1}B\left(\beta\right)^{-1}\right]\otimes B'\left(\beta\right)^{-1}\right\} K_{n,n}H & \left(A_{3}\right)\\
\\
 & \quad-H'\left\{ B\left(\beta\right)^{-1}u_{t}\left(\theta\right)\otimes\left[B'\left(\beta\right)^{-1}\Sigma^{-1}e_{xx,t}\left(\theta\right)\Sigma^{-1}B\left(\beta\right)^{-1}\left(\frac{\partial u_{t}\left(\theta\right)}{\partial\beta'}\right)\right]\right\}  & \left(A_{5}\right)\\
 & \quad-H'\left\{ u_{t}'\left(\theta\right)B'\left(\beta\right)^{-1}\otimes\left[\left(\frac{\partial u_{t}'\left(\theta\right)}{\partial\beta}\right)B'\left(\beta\right)^{-1}\Sigma^{-1}e_{xx,t}\left(\theta\right)\Sigma^{-1}B\left(\beta\right)^{-1}\right]\right\} H & \left(B_{2}\right)\\
\\
 & \quad-H'\left[B\left(\beta\right)^{-1}u_{t}\left(\theta\right)u_{t}'\left(\theta\right)B'\left(\beta\right)^{-1}\left(\theta\right)\otimes B'\left(\beta\right)^{-1}\Sigma^{-1}e_{xx,t}\left(\theta\right)\Sigma^{-1}B\left(\beta\right)^{-1}\right]H & \left(A_{4}\right)\\
 & \quad+\left(\frac{\partial u_{t}'\left(\theta\right)}{\partial\beta}\right)\left[B'\left(\beta\right)^{-1}\Sigma^{-1}e_{xx,t}\left(\theta\right)\Sigma^{-1}B\left(\beta\right)^{-1}\right]\left(\frac{\partial u_{t}\left(\theta\right)}{\partial\beta'}\right) & \left(B_{3}\right)\\
\\
 & \quad+H\left[B\left(\beta\right)^{-1}\left(\frac{\partial u_{t}\left(\theta\right)}{\partial\beta'}\right)\otimes B'\left(\beta\right)^{-1}\Sigma^{-1}e_{x,t}\left(\theta\right)\right] & \left(A_{2}\right)\\
 & \quad+\left[\frac{\partial}{\partial\beta'}\left(\frac{\partial u_{t}'\left(\theta\right)}{\partial\beta}\right)\right]B'\left(\beta\right)^{-1}\Sigma^{-1}e_{x,t}\left(\theta\right) & \left(B_{1}\right)\\
\\
 & \quad-H'K_{nn}\left(B'\left(\beta\right)^{-1}\otimes B\left(\beta\right)^{-1}\right)H & \left(C\right)
\end{align*}

\begin{align*}
l_{\beta\pi_{3},t}\left(\theta\right) & =H'\sum_{i=1}^{q}\left(\left\{ B\left(\beta\right)^{-1}w_{b,t-i}'\left(\theta\right)\right\} \otimes b_{i}'\right)B'\left(\beta\right)^{-1}\Sigma^{-1}e_{x,t}\left(\theta\right)\\
 & \quad+\left(H'\sum_{i=1}^{q}\left(B\left(\beta\right)^{-1}u_{t-i}\left(\theta\right)\otimes b_{i}'\right)\right)B'\left(\beta\right)^{-1}\Sigma^{-1}e_{xx,t}\left(\theta\right)\Sigma^{-1}B'\left(\beta\right)^{-1}w_{b,t-i}'\left(\theta\right)\\
 & \quad-H\left(\left\{ B\left(\beta\right)^{-1}w_{b,t-i}'\left(\theta\right)\right\} \otimes B'\left(\beta\right)^{-1}\Sigma^{-1}e_{x,t}\left(\theta\right)\right)\\
 & \quad-H\left(B\left(\beta\right)^{-1}u_{t}\left(\theta\right)\otimes B'\left(\beta\right)^{-1}\Sigma^{-1}e_{xx,t}\left(\theta\right)\Sigma^{-1}B'\left(\beta\right)^{-1}w_{b,t-i}'\left(\theta\right)\right)
\end{align*}

{\scriptsize{}
\begin{align*}
l_{\beta\pi_{2},t}\left(\theta\right) & =\left(H'\sum_{i=1}^{q}\left(B\left(\beta\right)^{-1}u_{t-i}\left(\theta\right)\otimes b_{i}'\right)\right)B'\left(\beta\right)^{-1}\Sigma^{-1}e_{xx,t}\left(\theta\right)\Sigma^{-1}B'\left(\beta\right)^{-1}x_{b,t-1}'\left(\theta\right)\\
 & \qquad-H\left(\left\{ B\left(\beta\right)^{-1}x_{b,t-1}'\left(\theta\right)\right\} \otimes B'\left(\beta\right)^{-1}\Sigma^{-1}e_{x,t}\left(\theta\right)\right)-H\left(B\left(\beta\right)^{-1}u_{t}\left(\theta\right)\otimes B'\left(\beta\right)^{-1}\Sigma^{-1}e_{xx,t}\left(\theta\right)\Sigma^{-1}B'\left(\beta\right)^{-1}x_{b,t-1}'\left(\theta\right)\right).
\end{align*}
}{\scriptsize\par}

\begin{align*}
l_{\sigma\sigma,t}\left(\theta\right) & =\Sigma^{-4}e_{xx,t}\left(\theta\right)\mathcal{E}_{t}^{2}\left(\theta\right)+2\Sigma^{-3}E_{x,t}\left(\theta\right)\mathcal{E}_{t}\left(\theta\right)+\Sigma^{-2}\\
l_{\sigma\beta,t}\left(\theta\right) & =\Sigma^{-2}\left(\Sigma^{-1}e_{xx,t}\left(\theta\right)\mathcal{E}_{t}\left(\theta\right)+E_{x,t}\left(\theta\right)\right)\left[\left(u_{t}'\left(\theta\right)B'\left(\beta\right)^{-1}\otimes B\left(\beta\right)^{-1}\right)-\left(\sum_{i=1}^{q}\left(u_{t-i}'B'\left(\beta\right)^{-1}\left(\theta\right)\otimes B\left(\beta\right)^{-1}b_{i}\right)\right)\right]H\\
l_{\sigma\pi_{3},t}\left(\theta\right) & =\Sigma^{-2}\left(\Sigma^{-1}e_{xx,t}\left(\theta\right)\mathcal{E}_{t}\left(\theta\right)+E_{x,t}\left(\theta\right)\right)B\left(\beta\right)^{-1}w_{b,t-1}'\left(\theta\right)\\
l_{\sigma\pi_{2},t}\left(\theta\right) & =\Sigma^{-2}\left(\Sigma^{-1}e_{xx,t}\left(\theta\right)\mathcal{E}_{t}\left(\theta\right)+E_{x,t}\left(\theta\right)\right)B\left(\beta\right)^{-1}x_{b,t-1}'\left(\theta\right)\\
l_{\lambda\lambda,t}\left(\theta\right) & =e_{\lambda\lambda,t}\left(\theta\right)\\
l_{\sigma\lambda,t}\left(\theta\right) & =-\Sigma^{2}\mathcal{E}_{t}\left(\theta\right)e_{x\lambda,t}\left(\theta\right).\\
l_{\beta\lambda,t}\left(\theta\right) & =-H'\sum_{i=1}^{q}\left(I_{n}\otimes b_{i}'B'\left(\beta\right)^{-1}\Sigma^{-1}\right)\left(\varepsilon_{t-i}\left(\theta\right)\otimes e_{x\lambda,t}\left(\theta\right)\right)-H'\left(I_{n}\otimes B'\left(\beta\right)^{-1}\Sigma^{-1}\right)\left(\varepsilon_{t}\left(\theta\right)\otimes e_{x\lambda,t}\left(\theta\right)\right)\\
l_{\pi_{3}\lambda,t}\left(\theta\right) & =-w_{b,t-1}\left(\theta\right)\Sigma^{-1}e_{x\lambda,t}\left(\theta\right)\\
l_{\pi_{2}\lambda,t}\left(\theta\right) & =-x_{b,t-1}\left(\theta\right)B'\left(\beta\right)^{-1}\Sigma^{-1}e_{x\lambda,t}\left(\theta\right)
\end{align*}

We choose the Frobenius norm, i.e. $\left\Vert A\right\Vert _{F}=\sqrt{\sum_{i=1}^{m}\sum_{j=1}^{n}a_{ij}^{2}}$
for an $\left(m\times n\right)$-dimensional matrix $A$, as particular
matrix norm because it satisfies $\left\Vert \begin{pmatrix}A_{11} & A_{12}\\
A_{21} & A_{22}
\end{pmatrix}\right\Vert _{F}^{2}=\left\Vert A_{11}\right\Vert _{F}^{2}+\left\Vert A_{12}\right\Vert _{F}^{2}+\left\Vert A_{21}\right\Vert _{F}^{2}+\left\Vert A_{22}\right\Vert _{F}^{2}$ for partitioned matrices, which in turn implies that $\mathbb{E}\left(\sup_{\theta\in\Theta_{0}}\left\Vert \begin{pmatrix}A_{11} & A_{12}\\
A_{21} & A_{22}
\end{pmatrix}\right\Vert \right)\leq\mathbb{E}\left(\sup_{\theta\in\Theta_{0}}\left\Vert A_{11}\right\Vert \right)+\mathbb{E}\left(\sup_{\theta\in\Theta_{0}}\left\Vert A_{12}\right\Vert \right)+\mathbb{E}\left(\sup_{\theta\in\Theta_{0}}\left\Vert A_{21}\right\Vert \right)+\mathbb{E}\left(\sup_{\theta\in\Theta_{0}}\left\Vert A_{22}\right\Vert \right)$. Moreover, the Frobenius norm is submultiplicative, i.e. $\left\Vert AB\right\Vert _{F}\leq\left\Vert A\right\Vert _{F}\left\Vert B\right\Vert _{F}$,
for matrices of appropriate dimension as a consequence of the Cauchy-Schwarz
inequality\footnote{Note that the Frobenius norm and the spectral norm $\left\Vert \cdot\right\Vert _{2}$
of a matrix (the matrix norm induced by the Euclidean vector norm),
satisfy the relation $\left\Vert A\right\Vert _{2}\leq\left\Vert A\right\Vert _{F}\leq\sqrt{\min\left(m,n\right)}\left\Vert A\right\Vert _{2}$
for an $\left(m\times n\right)$-dimensional matrix $A$, see e.g.
\citet{GvL13} page 72. Together with the fact that both $\left\Vert \cdot\right\Vert _{2}$
and $\left\Vert \cdot\right\Vert _{F}$ are unitarily invariant it
follows that even $\left\Vert AB\right\Vert _{F}\leq\left\Vert A\right\Vert _{2}\left\Vert B\right\Vert _{F}$
holds. This can be seen from considering the SVD of $A=U\Sigma V'$
(where $U$ and $V$ are orthogonal and $\Sigma$ is a diagonal matrix
with non-negative elements), such that $\left\Vert AB\right\Vert _{F}^{2}=\left\Vert U\Sigma V'B\right\Vert _{F}^{2}=\left\Vert \Sigma V'B\right\Vert _{F}^{2}=\sum\left|\sigma_{i}b_{ij}\right|^{2}\leq\max\left\{ \sigma_{i}\right\} \sum b_{ij}^{2}=\left\Vert \Sigma\right\Vert _{2}\left\Vert B\right\Vert _{F}=\left\Vert A\right\Vert _{2}\left\Vert B\right\Vert _{F}$.}, and satisfies $\left\Vert A_{m\times n}\otimes B_{p\times q}\right\Vert _{F}\leq C\cdot\left\Vert A_{m\times n}\right\Vert _{F}\left\Vert B_{p\times q}\right\Vert _{F}$
which follows from the equivalence of the matrix norms $\left\Vert \cdot\right\Vert _{2}$
and $\left\Vert \cdot\right\Vert _{F}$, from the fact that $\left\Vert A\otimes B\right\Vert _{2}=\left\Vert A\right\Vert _{2}\left\Vert B\right\Vert _{2}$
and that $\left\Vert I_{n}\right\Vert _{2}=1$.\footnote{To be more precise, we have 
\begin{align*}
\left\Vert A_{m\times n}\otimes B_{p\times q}\right\Vert _{F} & =\left\Vert \left(A_{m\times n}\otimes I_{p}\right)\left(I_{n}\otimes B_{p\times q}\right)\right\Vert _{F}\leq\left\Vert \left(A_{m\times n}\otimes I_{p}\right)\right\Vert _{2}\left\Vert \left(I_{n}\otimes B_{p\times q}\right)\right\Vert _{F}\leq\\
 & \leq\left\Vert A_{m\times n}\right\Vert _{2}\sqrt{\min\left(np,nq\right)}\left\Vert \left(I_{n}\otimes B_{p\times q}\right)\right\Vert _{2}=\sqrt{\min\left(np,nq\right)}\left\Vert A_{m\times n}\right\Vert _{2}\left\Vert B_{p\times q}\right\Vert _{2}\leq\\
 & \leq\sqrt{\min\left(np,nq\right)}\left\Vert A_{m\times n}\right\Vert _{F}\left\Vert B_{p\times q}\right\Vert _{F}.
\end{align*}
}

In order to verify that each of the partitioned matrices converges
(in expectation) uniformly, we note that $x_{b,t-1}(\theta),\ w_{b,t-1}\left(\theta\right)$,
the derivatives with respect to $\pi$ in the matrices $\mathcal{D}$
and $\mathcal{E}$ in the block pertaining to $l_{\pi_{3}\pi_{3},t}\left(\theta\right)$
are (causal dynamic) transformations of the (true) inputs $\left(\varepsilon_{t}\right)$.
Consider, e.g., $x_{b,t-1}\left(\theta\right)=\left(x_{t-1}\otimes b_{\theta}'(z)^{-1}\right)=\left(\begin{pmatrix}y_{t-1}\\
\vdots\\
y_{t-p}
\end{pmatrix}\otimes b_{\theta}'(z)^{-1}\right)$, where $y_{t-1}=a(z)^{-1}b(z)\varepsilon_{t-1}$ and the subscript
$\theta$ in $b_{\theta}'(z)$ is intended to emphasize that $b_{\theta}'(z)$
is a function of the parameters to be optimized while $a(z)^{-1}b(z)$
in $y_{t-1}=a(z)^{-1}b(z)\varepsilon_{t}$ corresponds to the truth.
Obviously, all elements in $x_{b,t-1}(\theta)$ are dynamic transformations
of the process $\left(\varepsilon_{t}\right)$ with geometrically
decreasing coefficients. It is easy to see that the power series (depending
on $\theta$) corresponding to the dynamic transformations converge
uniformly on the compact set $\Theta_{0}$. The same is true for $w_{b,t-1}(\theta)=\left(\begin{pmatrix}u_{t-1}\left(\theta\right)\\
\vdots\\
u_{t-q}\left(\theta\right)
\end{pmatrix}\otimes b_{\theta}'(z)^{-1}\right)$ where, e.g., $u_{t-1}\left(\theta\right)=b_{\theta}(z)^{-1}a_{\theta}(z)a(z)^{-1}b(z)\varepsilon_{t-1}$
and where the subscript $\theta$ is intended to emphasize the same
fact as before. 

The remainder of the argument is identical to the one in \citet{LMS_svarIdent16}
page 302, i.e. 
\[
\left|e_{i,x,t}\left(\theta\right)\right|,\ e_{i,x,t}^{2}\left(\theta\right),\ \left|e_{i,xx,t}\left(\theta\right)\right|,\ \left\Vert e_{i,x\lambda_{i},t}\left(\theta\right)\right\Vert ,\ \left\Vert e_{i,\lambda_{i}\lambda_{i},t}\left(\theta\right)\right\Vert 
\]
 are bounded by $C\left(1+\left\Vert u_{t}\left(\theta\right)\right\Vert ^{a_{i}}\right)$
for a generic (not always the same) constant $C$ according to assumption
\ref{assu:ULLN}. Since $u_{t}\left(\theta\right)=b_{\theta}(z)^{-1}a_{\theta}(z)a(z)^{-1}b(z)\varepsilon_{t}$
, it follows from the fact that power series with geometrically decreasing
coefficient are absolutely convergent within the unit circle and uniformly
convergent on every compact subset thereof and from Assumption \ref{assu:ULLN}
that 
\[
\mathbb{E}\left(\sup_{\theta\in\Theta_{0}}\left\Vert \frac{\partial^{2}\log\left(f\left(x,\theta\right)\right)}{\partial\theta\partial\theta'}\right\Vert \right)<\infty
\]
 holds.

\section{\label{sec:hessian_equal_opg}Expectation of Hessian Equals the Expectation
of the negative Outer Product of the Score}

The derivations are essentially identical to the ones in \citet{LMS_svarIdent16}.

\end{document}